\documentclass[11pt]{article}

\usepackage{amssymb}
\usepackage{mathrsfs}
\usepackage{graphics}
\usepackage{amsmath}
\usepackage{amsmath,amsthm}
\usepackage{fullpage}
\usepackage[numbers]{natbib}
\usepackage{setspace}
\usepackage[demo]{graphicx}
\usepackage{soul}
\usepackage{graphicx}
\usepackage{xcolor}
\usepackage{url}
\usepackage[hidelinks]{hyperref}

\usepackage{booktabs} 
\usepackage[ruled]{algorithm2e} 

\SetAlFnt{\small}
\SetAlCapFnt{\small}
\SetAlCapNameFnt{\small}
\SetAlCapHSkip{0pt}
\IncMargin{-\parindent}

\setcitestyle{authoryear}

\usepackage{mathtools,nicefrac,amsthm}

\newtheorem{theorem}{Theorem}[section]
\newtheorem*{theorem*}{Theorem}
\newtheorem{open}{ Open Question}
\newtheorem{claim}[theorem]{Claim}
\newtheorem{definition}[theorem]{Definition}   
\newtheorem{remark}[theorem]{Remark}
\newtheorem{example}[theorem]{Example}
\newtheorem{lemma}[theorem]{Lemma}
\newtheorem{corollary}[theorem]{Corollary}

\newcommand{\overv}{\overline{v}}
\newcommand{\underv}{\underline{v}}

\newcommand{\xvec}{\mathbf{x}}
\newcommand{\rvec}{\mathbf{r}}

\newcommand{\qvec}{\mathbf{q}}
\newcommand{\calF}{\mathcal{F}}
\newcommand{\calL}{\mathcal{L}}
\newcommand{\calM}{\mathcal{M}}
\newcommand{\calQ}{\mathcal{Q}}

\newcommand{\N}{\mathbb{N}}

\DeclarePairedDelimiter\abs{\lvert}{\rvert}

\DeclareMathOperator*{\argmax}{argmax}
\DeclareMathOperator*{\argmin}{argmin}
\DeclareMathOperator*{\Rev}{Rev}
\DeclareMathOperator*{\iid}{i.i.d.}

\title{Distributional Robustness:
From Pricing to Auctions}

\begin{document}


\author{ Nir Bachrach \thanks{Technion -- Israel Institute of Technology. Email: nir.bachrach@campus.technion.ac.il.}
\and Inbal Talgam-Cohen\thanks{Technion -- Israel Institute of Technology. Email: inbaltalgam@gmail.com.}
}

\maketitle

\begin{abstract}
Robust mechanism design is a rising alternative to Bayesian mechanism design, which yields designs that do not rely on assumptions like full distributional knowledge. We apply this approach to mechanisms for selling a single item, assuming that only the mean of the value distribution and an upper bound on the bidder values are known. We seek the mechanism that maximizes revenue over the worst-case distribution compatible with the known parameters. 
Such a mechanism arises as an equilibrium of a zero-sum game between the seller and an adversary who chooses the distribution, and so can be referred to as the max-min mechanism.

Carrasco et al.~[2018] derive the
max-min \emph{pricing} when the seller faces a single bidder for the item.
We go from max-min pricing to max-min \emph{auctions} by studying the canonical setting of two i.i.d.~bidders,
and show the max-min mechanism is the second-price auction with a randomized reserve.
We derive a closed-form solution for the distribution over reserve prices, as well as the worst-case value distribution, for which there is simple economic intuition. In fact we derive a closed-form solution for the reserve price distribution for any number of bidders.

Our technique for solving the zero-sum game is quite different than that of Carrasco et al.~-- it involves analyzing a discretized version of the setting, then refining the discretization grid and deriving a closed-form solution for the non-discretized, original setting. Our results establish a difference between the case of two bidders and that of $n \ge 3$ bidders.
\end{abstract}

\newpage
\section{Introduction}
\label{sec:intro}

\noindent{\bf Bayesian mechanism design.}
Traditional mechanism design centers around the Bayesian approach, in which unknown features of the environment (such as agents’ preferences) are assumed to come from a commonly-known prior distribution. The design goal is to maximize an objective like revenue in expectation over this prior.
A beautiful theory has developed for the Bayesian model;
in particular, consider one of the simplest standard mechanism design problems: a seller with a single item for sale, and bidders with privately-known values $v_1,v_2$ drawn independently from a known distribution~$F$. The design goal is to maximize the expected revenue (over the random values and any internal randomness of the mechanism). 
For the special case of a single bidder, it is not hard to deduce the optimal pricing mechanism -- set a deterministic price $r$ that maximizes the expected revenue $r(1-F(r))$. 
The celebrated theory of \citet{Myerson81} extends this result from optimal pricing to optimal auctions, yielding the mechanism with optimal expected revenue for the case of two or more bidders. 
One of the ``all-time greatest hits'' of mechanism design is the Myerson optimal auction for i.i.d.~values, which turns out to be the ubiquitous second-price auction with a deterministic reserve $r=r(F)$.


\vspace{0.05in}
\noindent{\bf Robust mechanism design.}
One drawback of the Bayesian theory is its reliance on a strong common-knowledge assumption with respect to the distribution. As Nobel laureate Robert Wilson pointed out, for economic theory to be applicable to practice we must strive to repeatedly weaken such assumptions~\cite{Wilson87}. The Wilson doctrine combined with classic robust optimization models inspired a rich literature on \emph{robust mechanism design}. 
In our single-item auction context, a distributionally-robust mechanism does not have full knowledge of the value distribution -- although such a distribution exists ``in the background''~\cite[\emph{cf}.,][]{DevanurHY15} -- and is required to perform well for a worst-case such distribution.
Robust mechanism design can be seen as bridging the worst- and average-case approaches of CS and economics, respectively,%
\footnote{Robust mechanism design can thus be viewed as part of the \emph{beyond worst-case analysis} agenda~\cite{Roughgarden20}.} 
thus holding much promise for generating new insight.

\vspace{0.05in}
\noindent{\bf Distributional robustness.}
Rather than full Bayesian knowledge, the designer is assumed to have access either to certain \emph{parameters} of the distribution (e.g.,~mean value), or to samples from it, or to additional bidders with values from the same distribution. The goal is to use these alternatives to design ``optimal'' distributionally-robust auctions, for an appropriate definition of optimality. Different notions of optimality have led to different branches of research, notably \emph{prior-independent} versus \emph{max-min} optimal mechanism design. 
Our focus in this paper is on parametric knowledge of the distribution and on max-min optimal auctions. 

The canonical problem in this area is described by
\citet{Carroll19}, as follows: 
\begin{quote}
``It can be natural to ask what happens if the designer has only partial information about the distribution, and wishes to maximize a guarantee under this partial information [...] 
For example, what happens if the seller instead does not know the distribution, but only knows the mean and an upper bound on $v$, and wishes to design a mechanism to \emph{maximize expected profit in the worst case} over all distributions consistent with this knowledge?'' \end{quote}
This is the problem we wish to make progress on in this paper.

\vspace{0.05in}
\noindent{\bf Zero-sum game perspective.}
Mathematically, the problem of designing the max-min optimal mechanism given parametric knowledge of the distribution can be modeled as a zero-sum game, where a seller-player is trying to maximize expected revenue, and an adversary-player (also referred to as ``nature'') is trying to minimize the revenue by picking the worst-case distribution compatible with the parametric knowledge. 
A \emph{Nash equilibrium} of the zero-sum game yields a stable max-min strategy for the seller, i.e., a distributionally-robust optimal auction. Many recent works on robust mechanism design have taken the zero-sum game approach~\cite[e.g.,][]{BeiGLT19,GravinL18,Carroll17,Carroll15,DuttingRT19,BabaioffFGLT20}. 

The zero-sum game model highlights two core ideas of the max-min approach: 
First, the seller should \emph{randomize} in order to hedge her uncertainty; indeed, the seller-player's equilibrium strategy is generally a \emph{mixed} one.
A second idea is to strike the right balance between relaxation of unrealistic knowledge assumptions, and good performance of the mechanism. This is achieved by carefully balancing the powers of the seller and adversary in the zero-sum game. 
If the adversary has too much power (reflecting extreme lack of knowledge on behalf of the seller), the robust auction arising from the equilibrium might be trivial and uninformative. 
An intuitive indication of striking the right balance is that the max-min strategy of the seller yields a robustly-optimal auction format that turns out to be ubiquitous in practice (a notable example is~\cite{Carroll17}). 



\vspace{0.05in}
\noindent{\bf A gap in the literature and our main result.}
The problem in the above description by \citet{Carroll17} has been addressed for several important cases. First and foremost, \citet{CarrascoFK+18} solve it when there is a \emph{single} bidder for the item, establishing the robust parallel to the optimal Bayesian \emph{pricing}. 
This parallel turns out to be a random pricing scheme where the price is distributed according to a log-uniform distribution -- in line with the above intuition of hedging against uncertainty through randomization. 

\citet{Che22} tackles a generalization to more than one bidder, but allows the bidder values to be \emph{correlated} in an unknown way (i.e., the adversary chooses the worst-case correlation among the values).
\citet{Suzdaltsev20a} also tackles a generalization beyond a single bidder; in his model the bidders are i.i.d., but the seller's strategies in the zero-sum game are limited to \emph{deterministic} auctions. To our knowledge, the case of max-min randomized auction design for more than one bidder with i.i.d.~values has not been addressed. 
In other words, the literature is missing the robust parallel to the optimal Bayesian \emph{auction}, which in this case is the “greatest hit” second-price auction with deterministic reserve.

Our main result in this paper addresses this gap in the literature for the two-bidder case.
In our zero-sum game, the adversary is \emph{not} allowed to correlate the bidder values, whereas the seller \emph{is} allowed to choose a randomized mechanism; we view this as a natural balance of powers between the players, and indeed the equilibrium of this game yields a standard and wide-spread auction format as the max-min optimum.
{\bf As our main result, we solve this game and show that the robust parallel to the optimal auction for two i.i.d.~bidders is the second-price auction with \emph{random} reserve.} We also give a closed-form solution for the seller's optimal distribution over reserve prices.

Interestingly, we can partially extend our result to any number of i.i.d.~bidders: For $n\ge3$ bidders, we find the robust-optimal distribution of the reserve price assuming the seller applies a second-price auction with reserve. However, such an auction is no longer the overall robust-optimal mechanism.

\subsection{Our Contribution and Takeaways}

\noindent{\bf Two bidders.} 
Our main theorem is the following:

\begin{theorem*}[Informal, see Theorem~\ref{thm:two_bidders:main:among_all_truthful}] 
Consider a parametric auction setting with two bidders, whose values are drawn i.i.d.~from a distribution with known mean value $\mu$ and support contained in $[0,\overv]$.
Consider the zero-sum game between the seller choosing a mechanism, and the adversary choosing a value distribution with mean $\mu$ and support contained in $[0,\overv]$, where the seller's payoff is the expected revenue.
If $M^*$ is the second-price auction with randomized reserve sampled from $Q^*$, then $(M^*,F^*)$ is an equilibrium of the zero-sum game, where $Q^*,F^*$ are CDFs defined (roughly) as follows:
\begin{itemize}
    \item $\underv$ is the unique solution in $[0, \overv]$ to $\underv\left(1+\log\nicefrac{\overv}{\underv}\right) = \mu$.
    \item $F^*$ is the (generalized) equal-revenue distribution over the support $[\underv,\overv]$, that is, $F^*(v)=1-\nicefrac{\underv}{v}$. 
    \item $Q^*$ is a normalized ratio between the log-uniform distribution and the equal-revenue distribution:
    $$
    Q^*(r)\propto
    \frac{\log\nicefrac{r}{\underv}}{\log \nicefrac{\overv}{\underv}} \cdot \frac{1}{1-\nicefrac{\underv}{r}} 
    \text{ where }
    r\in(\underv,\overv].
    $$
\end{itemize}
\end{theorem*}

\noindent The theorem provides the optimal distributionally-robust auction for a seller whose parametric knowledge is $(\mu,\overv)$, as well as the worst-case distribution $F^*$ of the bidders' values.
The key takeaway is that while a randomized reserve price is the optimal distributionally-robust pricing for a single bidder~\cite{CarrascoFK+18}, the second-price auction with a randomized reserve is the optimal distributionally-robust pricing for two bidders. 

\vspace{0.05in}
\noindent{\bf Economic interpretation: Indifference.} 
The worst-case (min-max) equilibrium distribution~$F^*$ over bidder values has a clear economic interpretation -- it always induces the same expected revenue for the seller, regardless of the reserve price realization. Indeed, this ability to ``induce indifference'' is the salient property of the equal-revenue distribution, and the reason why it emerges in the analysis of so many mechanism design settings. 
The indifference on the seller's behalf means that the adversary is ``robust'' to the seller's randomization over prices, making equal-revenue a natural candidate for the adversary's min-max optimum.
Intuitively, in the absence of knowledge (in this case about the seller's reserve price), the player (adversary) should opt for a strategy whose payoff is agnostic to the unknown details; any dependence on such details will be taken advantage of by the opponent (seller). 

Similar explanations for robustly-optimal designs that induce indifference can be found in~\cite[e.g.,][]{Carroll15,DuttingRT19,Carroll17,GravinL18,CarrascoFK+18}. 
However, as noted by~\citet{BabaioffFGLT20}, this principle of agnosticism does not always hold. In our case, the max-min distribution $Q^*$ of the seller does not induce indifference on the adversary's side.  
In other words, the seller fails to make the adversary indifferent, thus enabling the adversary to take advantage of the seller's ignorance. 

\vspace{0.05in}
\noindent{\bf Mathematical explanation of the indifference.}
Our analysis gives an intuitive mathematical explanation for why the principle of agnosticism holds for the seller but not for the adversary, and a general way to predict for which zero-sum games and which players we expect to get indifference. 
Consider a zero-sum game which has an equilibrium composed of strategies in the interior of the players' strategy spaces.%
\footnote{If the equilibrium is not composed of such strategies, the indifference can be only on a subspace where the strategies do fall in the interior.} 
Consider the game's payoff function. If $R$ is the random reserve price of the second-price auction played by the seller, and $F$ is the value distribution played by the adversary, the payoff (expected revenue) can be shown to be
$R(1-F^2(R)) + \int_{R}^{\overv} \left(1 - F(v) \right)^2 dv$.
A key observation is that the payoff is \emph{linear} in the distribution of $R$ (the seller's strategy) and \emph{quadratic} in $F$ (the adversary's strategy). 

By this observation, we expect indifference for the seller but not the adversary: Since the equilibrium is a saddle point of the payoff function (which by assumption falls in the interior), the gradients of the payoff function should be zero at the equilibrium point. Furthermore, by linearity in the seller's strategy, the gradient is a function of only the adversary's strategy. 
Putting these two together, the gradient is zero given the equilibrium strategy of the adversary for any strategy of the seller. Hence the seller's strategy does not affect her payoff -- i.e., the seller is indifferent. For similar reasons, the non-linearity of the payoff in the adversary's strategy means we do not expect indifference for the adversary.

\vspace{0.05in}
\noindent{\bf Beyond two bidders.} Our main theorem beyond two bidders is the following:

\begin{theorem*}[Informal, see Theorem~\ref{thm:multi:main:equilibrium}]
Consider a parametric auction setting with $n \ge 3$ bidders, whose values are drawn i.i.d.~from a distribution with known mean value $\mu$ and support contained in $[0,\overv]$.
Consider the zero-sum game between the seller choosing a distribution for a randomized reserve price in a second price auction, and the adversary choosing a value distribution with mean $\mu$ and support contained in $[0,\overv]$, where the seller's payoff is the expected revenue.
Then $(Q^*,F^*)$ is an equilibrium of the zero-sum game, where $Q^*,F^*$ are CDFs defined (roughly) as follows:
\begin{itemize}
    \item $\underv$ is the smaller among two solutions to $\left(n-\frac{1}{n-1}+\log\nicefrac{\overv}{\underv}\right)\underv=\left(n-1\right)^{2}\mu.$
    \item Case 1: $\underv\le \overv$.
\begin{itemize}
    \item $F^*(v)$ mixes a generalized equal-revenue distribution over the support $[\underv,\overv]$, that is $1 - \frac{\underv}{(n-1)^2 v}$, and a point mass distribution at $v_0 = \frac{\underv}{n-1}$. 
    \item $Q^*(r)$ is a normalized ratio between a truncated log-uniform distribution and $\left(F^* \right)^{n-1}$:
    $$
    Q^*(r)\propto
    \frac{n-1 - \frac{1}{n-1} + \log\nicefrac{r}{\underv}}{n-1 - \frac{1}{n-1} +\log \nicefrac{\overv}{\underv}} \cdot \left( F^*(r) \right)^{-(n-1)}
    \text{ where }
    r\in(\underv,\overv].
    $$
    For $r\in [0, \underv)$ it is simply $Q^*(r) = Q^*(\underv)$.
\end{itemize}
\item Case 2: $\underv\ge \overv$.
\begin{itemize}
    \item $Q^*$ is simply a point mass distribution on zero, which is equivalent to a simple second price auction with no reserve price.
    \item $F^*$ is a binary distribution between two possible values $\{v_0, \overv\}$, where $v_0 = \mu-\frac{\overv-\mu}{\left(n-1\right)^{2}-1}$.
\end{itemize}
\end{itemize}
\end{theorem*}
This theorem provides a solution to the seller's optimal distribution for a reserve price in a second price auction, but it does not guarantee that this is the robust optimal mechanism overall, and indeed we can find mechanisms that would achieve better results if the bidders' valuations are drawn from $F^*$ (Section \ref{subsec:multi:not_optimal}).
In this case, the indifference that nature induces on the seller is only partial, as the indifference only happens outside of $[v_0, \min\{\underv, \overv\}]$. While a second price auction with randomized reserve cannot take advantage of this, other mechanisms can abuse $F^*$ better, and get a higher expected revenue.

\subsection{Our Techniques}
\label{sub:techniques}

The method of \citet{CarrascoFK+18} does not seem to easily extend beyond a single bidder. \citet{CarrascoFK+18} characterize the transfer function of the distributionally-robust optimal pricing as the non-negative monotonic hull of a linear polynomial.%
\footnote{They also address a parametric setting where the first $N$ moments of the distribution are known, in which case the polynomial is of degree $N$.}
In the case of multiple bidders, the polynomial is no longer single-dimensional, and so applying this method seems more complicated. Instead, we suggest the following approach. 

We first limit the seller from applying general mechanisms to using only second-price auctions with a random reserve. We refer to the resulting game (in which the adversary's strategy space has not changed) as the \emph{reduced} zero-sum game.
In order to solve this game, we consider a simpler case, where the distributions that constitute the strategies of both the seller and the adversary are discrete, and the possible reserves/values are constrained to a known finite set.

For the discrete case, we only find recursive relations between adjacent points in the constrained finite set. By increasing the size of the finite set, and making the distance between the points smaller and smaller, we can take the limit of the recursive relations, and get ODEs (ordinary differential equations) that describe a solution to the non-discretized problem.
This enable us to guess a potential solution.

It remains to show that our guess based on the ODEs is an equilibrium of the reduced zero-sum game.
This is done by showing that no player can do better, even with perfect prior knowledge on the other player's chosen distribution.
In fact we show that in the two bidder case, it is an equilibrium even of the non-reduced game, in which the seller can choose any mechanism (rather than only second-price-based mechanisms). To achieve this we utilize the indifference of the seller over the interval $[0, \overv]$.

In Appendices~\ref{appndx:two_bidders_discrete}-\ref{appndx:sec:cont} we provide an alternative analysis in which we are not required to guess a solution, rather we deduce it mathematically. Appendix~\ref{appndx:sec:toolbox} in particular provides a mathematical foundation for the technique of deducing the general solution based on the discrete one. 
We also show that the idea of discretization can help solve other variants, e.g. one in which no upper bound on the values is known, by revisiting the single-bidder case solved by \citet{CarrascoFK+18}. Our analysis of this case appears in Appendix \ref{appndx:single}.

\subsection{Additional Related Work}

\noindent{\bf Robust mechanism design.}
There is a large and steadily-growing body of literature on the design of robust economic mechanisms in various domains, where the robustness is to different kinds of uncertainties (e.g., contract design robust to uncertain technologies~\cite{Carroll15}, Bayesian persuasion robust to uncertain utilities~\cite{BabichenkoTXZ21}, or auction design robust to the bidders’ attitude toward ambiguity~\citet{Kocyigit19}). The literature is too large to cover here comprehensively -- see \cite{Carroll19} for a recent survey.

\vspace{0.05in}
\noindent{\bf Distributional robustness.} 
A very active research area studies distributional robustness, with several routes being explored. 
First, one needs to choose the benchmark against which the robust mechanism's performance is measured. The main benchmarks are (1) the optimal non-robust auction that knows the distribution, and (2) our benchmark, namely the max-min robust auction that maximizes the performance against the worst-case distribution.
A third benchmark -- (3)~the welfare -- was recently considered in very interesting work of~\cite{Anunrojwong22}, which is concurrent and independent to our work. 

For the first benchmark, the performance is usually measured by a multiplicative approximation guarantee, and this branch of research is known as \emph{prior-independent} mechanism design~\cite[e.g.,][]{DhangwatnotaiRY15,Talgam20}.
An alternative measure is additive regret~\cite[e.g.,][]{BabichenkoTXZ21,Savage51}.%
\footnote{Regret in repeated settings has also been studied but is beyond our scope.} 
For the second benchmark, in this work we aim to achieve it exactly rather than approximately, but approximations of it have been considered in other contexts~\cite[e.g.,][]{BeiGLT19}). 
In some cases connections have been established between the two benchmarks~\cite{AzarM13,GiannakopoulosP20}.
For the third benchmark,~\citet{Anunrojwong22} show that the second price auction with randomized reserve minimizes the worst-case regret in comparison to extracting full welfare. In other words, it minimizes the regret due to not knowing the values and instead only knowing their supports. Their work thus establishes another, very different way in which the second price auction with randomized reserve is robustly optimal, and their robust optimality result persists across several distribution families.

Besides choosing the benchmark, another choice is whether to allow the augmentation of resources to compensate for the uncertainty, e.g.~the addition of competing bidders~\cite{BulowK96,Kocyigit19}, or alternatively whether to allow sample access to the unknown distribution~\cite[e.g.,][]{ColeR14,AllouahAB21,AllouahBB21a,BabaioffGMM18}. In this work we allow neither; closer to our approach are works that allow access to some statistic~\cite[e.g.,][]{AllouahBB21}. 

\vspace{0.05in}
\noindent{\bf Robustness to correlation.} The basic setting of distributionally robust mechanism design includes a single seller, bidder and item. A recent line of work studies \emph{more than one item}, applying robustness to the challenging domain of multi-parameter mechanism design. The robustness in this context is usually to correlation among the item values~\cite{Carroll17,GravinL18,BabaioffFGLT20}, and sometimes (in addition to correlation) to details of the marginal distributions~\cite{CheZ21,BrooksD21a,GiannakopoulosP20}.
%
The basic setting has also been generalized to the complementary case of \emph{multiple bidders}, where the robustness is to correlation among them (and possibly to distributional details as well)~\cite{Suzdaltsev20b,BeiGLT19,Che22,Kocyigit19,HeL22}.
Distributional robustness has also been studied with asymptotically-many bidders~\cite[e.g.][]{Segal03}. 

\vspace{0.05in}
\noindent{\bf Max-min mechanisms with parametric information.}
The closest work to ours is by \citet{CarrascoFK+18}. They study the same problem as ours, but with a \emph{single} bidder. The parametric information they consider is either the first moment and an upper bound on the values, or the first two moments. 
They also consider multiple moments, although in this case the max-min mechanism cannot be
given explicitly. 
Their work has generated many follow-ups for a single bidder:
Their max-min design is shown by \citet{GiannakopoulosP20} to also be a good prior-independent design; 
\citet{CarrascoFMM19} study an extension to a divisible item; 
and \citet{ChenHW21} show an alternative proof through via the minimax theorem and a novel geometric approach. 
\citet{Suzdaltsev20a} studies almost the same problem as ours -- in his model there is also more than one i.i.d.~bidder, but unlike us he only considers \emph{deterministic} mechanisms. He therefore gets a completely different design with no reserve price, which does not form an equilibrium of the zero-sum game.

To our knowledge, the first works on parametric information were by~\citet{AzarM13} and~\citet{AzarMDW13}, in the context of prior-independent mechanism design.
En route to prior-independence, \citet{AzarM13} find a deterministic max-min solution in a setting with multiple bidders and \emph{infinitely many copies} of the item -- their setting is thus closer to a single-bidder setting than to a multi-bidder one. 
\citet{AllouahB20} also study prior-independence with parametric information.

\vspace{0.05in}
\noindent{\bf Robustness for two bidders.}
Several recent works on robust mechanism design focus on the canonical two-bidder case, but in different contexts than ours: 
First, a central open question in prior-independent mechanism design is to identify the prior-independent revenue-optimal mechanism for selling a single item to two agents with i.i.d.~values. There is a line of work on this question, including \cite{HartlineJL20,FuILS15, DhangwatnotaiRY15,AllouahB20}. 
\citet{HartlineJ21} develop a method for lower-bounding prior-independent guarantees and apply it to a pair of two-bidder problems.
\citet{BergemannBM16} study max-min auction design for two bidders, but where the seller's uncertainty is about the correct model of the bidders’ beliefs on the item's quality, so the max-min optimum is very different than ours.






\section{Model}
\label{sec:model}


\subsection{Parametric Auction Settings} 

We study a simple auction setting in which a seller offers a single item to $n$ bidders. Each bidder~$i$ has a privately-known value $v_i\ge 0$ for the item, and the values are i.i.d.~samples from a distribution~$F$. Let $\vec{v}=(v_1,\dots, v_n)$. The goal of the seller (a.k.a.~\emph{designer}) is to maximize her expected revenue from selling the item to the competing bidders, where the expectation is taken over the random values as well as the internal randomness of the mechanism. 
What distinguishes our setting from standard mechanism design is that the seller has only \emph{parametric knowledge} of the distribution~$F$. 
We denote the parameters of the distribution as follows: let $\mu$ be its expectation, and let $\overv$ be an upper bound on its support, i.e.: 
$$
\mathbb{E}_{v_i\sim F}[v_i]=\mu;~~~
v_i\in [0,\overv].
$$
A parametric auction setting is thus defined by a pair $(\mu,\overv)$. Let $\calF=\calF(\mu,\overv)$ be the space of all \emph{compatible} distributions, where $F$ is compatible if it has expectation $\mu$ and support within $[0,\overv]$.

\vspace{0.05in}
\noindent{\bf The design space: Truthful auctions.} 
The auctions we consider as the seller's design space are standard:
The bidders submit bids $b_1, \dots, b_n$ for the item. The seller uses an allocation rule (possibly randomized) to map the bids $\vec{b}=(b_1, \dots, b_n)$ to an allocation of the item. A payment rule $p$ determines how much to charge the bidders in expectation: $p_i(\vec{b})$ is the expected payment of bidder $i$ given bid profile~$\vec{b}$.
The allocation and payment rules may depend on the parametric knowledge of the distribution from which the bids are independently drawn.

Bidders have quasi-linear utilities: 
bidder $i$'s expected utility is his expected value for the (possibly randomized) allocation, less his expected payment $p_i(\vec{b})$. 
We focus on \emph{truthful} -- i.e., dominant-strategy \emph{incentive compatible (IC)} and \emph{individually rational (IR)} -- auctions, where the bidders maximize their expected utilities by participating and reporting their true values to the seller (regardless of how others report).%
\footnote{Since distribution $F$ is unknown, the natural truthfulness notion in our model is dominant-strategy rather than Bayesian, since it is hard to predict how bidders would behave under uncertainty~\cite{BeiGLT19}. Even if bidders did know the distribution, one would like mechanisms that do not rely on the bidders' precise knowledge, especially since such mechanisms tend to be unnatural (i.e., asking the bidders to report the distribution, and punishing them all if the reports disagree).}
Thus we can assume $\vec{b}=\vec{v}$. 
Requiring truthfulness is without loss of generality by the revelation principle~\cite{Myerson81}. 
Denote the design space of truthful auctions for a parametric auction setting $(\mu,\overv)$ by $\calM=\calM(\mu,\overv)$. 


\vspace{0.05in}
\noindent{\bf A reduced design space: Second-price auctions with a random reserve.} A simple, ubiquitous auction format that will play a central role in our results is the \emph{second price auction with reserve}. Let $r$ (or $R$) denote the reserve price. This auction is (dominant-strategy) truthful, so we can assume that the bidders truthfully report $v_1,\dots, v_n$ as their bids. The item is allocated to a highest bidder $i^*\in \argmax_i\{v_i\}$ if the highest value $v_{i^*}$ strictly exceeds the reserve $r$ (tie-breaking among the bidders can be arbitrary). The item is not allocated if neither value clears the reserve. 
If the item is allocated, the winner's payment is the maximum between the reserve $r$ and the second-highest bid.
Note that the reserve price is allowed to be \emph{randomized}; 
we denote the distribution over reserve prices by $Q$. We are interested in studying the reduced design space of second-price auctions with a random reserve. If the seller is limited to the reduced design space, then $Q$~will be chosen depending on the seller's parametric knowledge. Denote the space of distributions over the reserve price (which constitutes the reduced design space) 
by $\calQ=\calQ(\mu,\overv$).

\subsection{Robust Revenue Performance and Zero-Sum Games}
The seller's goal is to use her parametric knowledge of the value distribution to choose an auction $M$ with optimal \emph{robust revenue performance} -- the expected revenue of $M$ when the bidders' values are drawn from the \emph{worst-case} compatible distribution $F\in\calF$. 
Following a long tradition in robust mechanism design~\cite{Wald50}, we approach this problem by treating it as a zero-sum game between two players: the seller, and a player we refer to as the \emph{adversary} or \emph{nature}. 

\vspace{0.05in}
\noindent{\bf The general zero-sum game.}
In the general zero-sum game, the seller can choose any mechanism $M$ in the design space $\calM$ as her action. The action space of the adversary is $\calF$, consisting of every compatible distribution $F$. For a pair of actions $(M,F)$, the seller's payoff is the revenue of $M$ in expectation over its internal randomness as well as over the random i.i.d.~values $v_1, \dots, v_n$ drawn from~$F$. We denote this payoff by $\Rev_{F}(M)$ where
$$
\Rev_{F}(M)=\mathbb{E}_{v_1,v_2\sim_{\iid}F}\left[\sum_i p_i(\vec{v})\right]
$$
(recall that $p$ is the payment rule of $M$; the adversary's payoff is of course $-\Rev_{F}(M)$). 

\vspace{0.05in}
\noindent{\bf The reduced zero-sum game.}
We are also interested in a related zero-sum game in which the action space of the adversary is unchanged, and the action space of the seller is reduced to include only second-price auctions with random reserves. This action space can be described as $\calQ$, where every $Q\in\calQ$ is a distribution over reserve prices.
Note that in a setting where there is a known upper bound $\overv$ on the values, there is no point in setting a reserve price above $\overv$. Therefore from now on we assume without loss that $\calQ$ contains only distributions with supports contained in $[0,\overv]$.

To define the payoff of the reduced zero-sum game, we introduce the following (slightly overloaded) notation:
let $\Rev_{F}(r)$ denote the expected revenue of the second-price auction with a fixed reserve price $r$ (the expectation is over the i.i.d.~values from $F$).
For a pair of actions $(Q,F)$, the seller's payoff is $\Rev_{F}(r)$ in expectation over the random reserve $r\sim Q$. 
We denote this payoff by 
$$
\Rev_{F}(Q) = \mathbb{E}_{r\sim Q}\left[\Rev_{F}(r)\right].
$$

\noindent{\bf Problem formulation.} 
In each of the zero-sum games, the seller's problem is to find a max-min strategy. In fact, we seek an equilibrium for each of the zero-sum games.%
\footnote{The advantage of an equilibrium is that it ensures the seller's strategy is stable, i.e., the seller would not gain from deviating.} 
We denote an equilibrium of the general zero-sum game by $(M^*,F^*)$, and refer to the problem of finding it by Problem~\eqref{eq:problem}. We denote an equilibrium of the reduced zero-sum game by $(Q^*,F^*)$, and refer to the problem of finding it by Problem~\eqref{eq:problem_reduced}:
\begin{align}
    M^*&\in \argmax_{M\in \calM} \min_{F\in \calF} \Rev_{F}(M),&
    F^*&\in \argmin_{F\in \calF} \max_{M\in \calM} \Rev_{F}(M)
    ;\tag{P1} \label{eq:problem}\\
    Q^*&\in \argmax_{Q\in \calQ} \min_{F\in \calF} \Rev_{F}(Q),&
    F^*&\in \argmin_{F\in \calF} \max_{Q\in \calQ} \Rev_{F}(Q).\tag{P2} \label{eq:problem_reduced}
\end{align}
Observe that $M^*$ (respectively, $Q^*$) maximizes the robust revenue performance over the seller's design space (respectively, reduced design space), since in both cases the worst-case $F\in\calF$ is chosen.
Thus, solving Problem~\eqref{eq:problem} achieves our ultimate goal of finding the auction with optimal robust revenue performance.

Note that it is not immediately clear that the general and the reduced zero-sum games have equilibria; there is a large supply of minimax theorems that can potentially be applied, but this is redundant as we shall explicitly construct pure-strategy equilibria for both games.

\vspace{0.05in}
\noindent{\bf A note on coinciding with the reserve price.} 
Notice that in our formulation of second-price auctions with a reserve price, the item is sold only when there is at least one bidder with a \emph{strictly} higher bid than the reserve. The strictness is not necessary for our results, but guarantees the existence of a worst-case distribution~$F$ for any distribution $Q$ over reserve prices: 
For example, if the seller (who chooses $Q$) places a probability mass on reserve price $r$, our formulation avoids a situation where the adversary (who chooses $F$) wishes to place probability mass on values as close as possible to $r$ from below, but not on $r$ itself. 
In this situation, without strictness, $F$ could only approach the worst case distribution but could never quite reach it. 

\vspace{0.05in}
\noindent{\bf Distributions of interest.}
The \emph{equal-revenue} distribution ``arises in many examples'' in the mechanism design literature~\cite[Chapter 4]{Hartline21}. In most appearances in the literature, the support of this distribution is $[1,\infty)$ or $[1,\overv]$ in truncated versions. We consider a slightly generalized version in which the support is $[\underv,\overv]$: the CDF is $F(v)=1-\nicefrac{\underv}{v}$ for $v\in [\underv,\overv)$, and $F(\overv)=1$. Another distribution of interest for us is the \emph{reciprocal} or \emph{log-uniform} distribution: the CDF is $F(v)=\frac{\log\nicefrac{v}{\underv}}{\log \nicefrac{\overv}{\underv}}$ for $v\in [\underv,\overv]$.
Throughout, all logarithms are in the natural base $e$ unless indicated otherwise.

\section{The Case of Two Bidders} \label{sec:two_bidders}
In this section we solve Problem~\eqref{eq:problem}, when there are only two bidders. We will achieve that by solving Problem~\eqref{eq:problem_reduced}, and proving that no mechanism can do better.

\begin{theorem}[Solution of Problem~\eqref{eq:problem_reduced} for two bidders] \label{thm:two_bidders:main:equilibrium}
When there are two bidders, there exists an equilibrium $\left(Q^*, F^*\right) \in \calQ \times \calF$ that solves Problem \eqref{eq:problem_reduced}. The equilibrium is described by:
\begin{align*}
    F^*(v) &= \begin{cases}
    0 & v \in [0, \underv], \\
    1-\nicefrac{\underv}{v} & v\in [\underv, \overv) ,\\
    1 & v\ge\overv
    \end{cases}, \\
    Q^*(r) &= \begin{cases}
         \left( 1 - \nicefrac{\underv}{\overv} \right) \cdot \left(\frac{1}{\log \overv - \log \underv}\right) & r\in [0, \underv],\\
        \left( 1 - \nicefrac{\underv}{\overv} \right) \cdot \left( \frac{r}{r-\underv} \right) \cdot \left( \frac{\log r - \log \underv}{ \log \overv - \log \underv} \right) & r\in (\underv, \overv],
    \end{cases}
\end{align*}
where $\underv$ is the unique solution to $\underv\left(1+\log\overv-\log\underv\right) = \mu$.
\end{theorem}

\begin{figure}[t]
\centering
\includegraphics[width=0.7\textwidth]{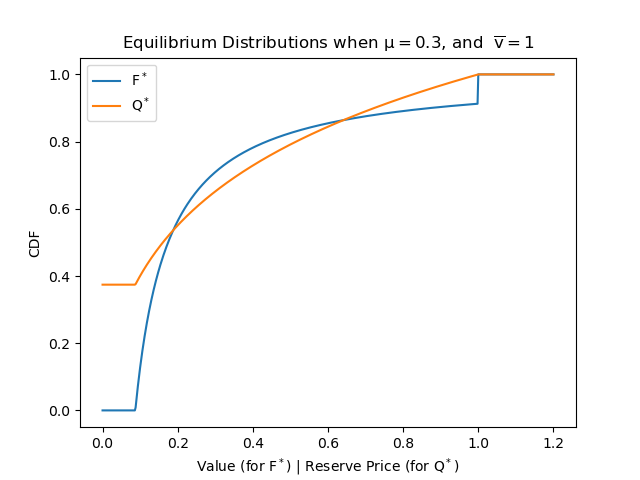}
\caption{
The solution to Problem~\eqref{eq:problem_reduced} with two bidders for the parametric settings in which $\mu=0.3$ and $\overv=1$. The pair $(F^*,Q^*)$ is an equilibrium of the reduced zero-sum game, where $F^*$ is the adversarial distribution of bidder values, and $Q^*$ is the seller's distribution over reserve prices for the second-price auction. 
$F^*$ is a generalized equal-revenue distribution with support $[\underv,\overv]$ for an appropriate $\underv$, and $Q^*$ is a ratio of a log-uniform distribution and an equal-revenue one, with the same support.
In fact, we show that $(F^*,Q^*)$ also constitutes a solution to Problem~\eqref{eq:problem}, yielding a distributionally-robust optimal auction for two bidders.}
\centering
\label{fig:two_bidders:general:equilibrium_mu0.3_overv1}
\end{figure}

Interestingly, the distribution chosen by the adversary is identical to the case of a single bidder solved by~\cite{CarrascoFK+18}. The distribution chosen by the seller is different: it has a positive mass on zero, and the rest of the mass is spread continuously on $[\underv, \overv]$.
This is illustrated (for a particular parametric setting) in Fig \ref{fig:two_bidders:general:equilibrium_mu0.3_overv1}. For a single bidder, the seller chooses the log-uniform distribution $\frac{\log \nicefrac{r}{\underv}}{\log \nicefrac{\overv}{\underv}}$. In the two bidder case, we observe that the sellers uses the log-uniform distribution divided by the equal-revenue distribution, and normalized to be in $[0,1]$.

An important and possibly surprising implication of Theorem~\ref{thm:two_bidders:main:equilibrium} is that the second-price auction with a randomized reserve sampled from $Q^*$ is an \emph{overall} robustly optimal mechanism. This constitutes our main result for two bidders and is proven in Section~\ref{subsec:two:proof}.

\begin{theorem}[Main: Solution of Problem~\eqref{eq:problem} for two bidders] \label{thm:two_bidders:main:among_all_truthful}
Let $M^*$ be the second-price auction with randomized reserve sampled from $Q^*$, where $Q^*$ is as described in Theorem \ref{thm:two_bidders:main:equilibrium}. Let $F^*$ be as described in Theorem \ref{thm:two_bidders:main:equilibrium} as well. Then, if there are exactly two bidders, then $(M^*,F^*)$ is an equilibrium that solves Problem \eqref{eq:problem}.
\end{theorem}

\vspace{0.05in}
\noindent{\bf Proof overview.}
Since the adversary is limited to choosing only distributions with a known mean value, we can use the method of Lagrange multipliers, and find a saddle point for the Lagrangian developed in Section~\ref{subsec:two:lagrangian}.
In Section~\ref{subsec:two:discrete} we consider a discrete case, where the seller's possible reserve prices, and the bidders' possible valuations are limited to a finite set.
We then find certain conditions that a saddle point of the Lagrangian would satisfy in the discrete case (Claims~\ref{clm:two_bidders:disc:nature_induction}-\ref{clm:two_bidders:disc:seller_induction}).
Afterwards, by importing these conditions to the general case, we find ODEs which we suspect that a saddle point for the non-discrete Lagrangian would satisfy, and use them to find a candidate solution $(Q^*, F^*)$ (Section~\ref{subsec:two:guess_odes}).
After finding the suspected solution $(Q^*, F^*)$, we prove in Section~\ref{subsec:two:proof} that it is a saddle point of the Lagrangian, and simultaneously prove that no other mechanism could achieve a better expected revenue if the adversary does choose distribution $F^*$. This completes the proof.

\vspace{0.05in}
\noindent{\bf A discrete-to-continuous technique of possible independent interest.}
Note that our proof, in particular Section~\ref{subsec:two:guess_odes}, uses a technique that allows us to go from the discrete case to the general case, as follows. 
As described in the proof overview, in Section~\ref{subsec:two:discrete} we find conditions which a saddle point of the Lagrangian (i.e.~an equilibrium of the discrete game) would satisfy assuming it's not on the boundary.
Roughly, these conditions describe a recursive relation that adjacent points in the discretization should satisfy to be an equilibrium. For example, if we consider points $y_1,\dots,y_{k}$ of the discrete grid, a recursive relation is of the form $y_{i+1}=\psi(y_i,\dots)$ (where the $\dots$ can be replaced by additional arguments).

In Section~\ref{subsec:two:guess_odes} we wish to go from the discrete to the general case. We do this by making the discretization finer and finer, and more importantly, translating the recursive relation into an ODE condition that holds in the limit. 
Interestingly, we do not need to solve the discrete case, prove that the limit exists, nor show that the condition on the limit necessarily holds. This is because we use this technique only to derive a guess of what the equilibrium is in the general case, and later prove that we guessed correctly. For a complete picture however, in Appendices \ref{appndx:two_bidders_discrete}-\ref{appndx:sec:cont} we prove the correctness of this technique, namely we solve the discrete case, and establish both convergence and that the resulting ODE condition on the equilibrium indeed holds.
This proof of correctness strengthens our belief that the described discrete-to-continuous technique may be of independent interest.

\subsection{The Lagrangian for Maximizing Expected Revenue} \label{subsec:two:lagrangian}
We denote by $v^{(i)}$ the $i$-th highest valuation of the bidders.
When the seller uses a second price auction with a reserve price $r$, the expected revenue is given by
\begin{align*}
    \Rev_{F}(r) &= E\left[r\cdot1_{v^{\left(2\right)}\le r<v^{\left(1\right)}}\right]+E\left[v^{\left(2\right)}\cdot1_{r<v^{\left(2\right)}}\right]\\
    &= r\left(1-F^{2}\left(r\right)\right)+\int_{r}^{\overv}\left(1-F\left(v\right)\right)^{2}dv.
\end{align*}
Therefore, when the seller randomizes a reserve price sampled from distribution $Q$, the expected revenue is
\begin{align*}
    \Rev_{F}(Q) &= \int_{-\infty}^{\infty}\left(r\left(1-F^{2}\left(r\right)\right)+\int_{r}^{\overv}\left(1-F\left(v\right)\right)^{2}dv\right)dQ\left(r\right).
\end{align*}
We can assume that $Q$'s support is contained in $[0,\overv)$, since negative reserve price is meaningless, and reserve price of $\overv$ or higher would yield now revenue.
If we also assume that $Q$ is continuous over $(0, \overv)$ and differentiable almost everywhere, we get
\begin{align}
    \Rev_{F}(Q) &= Q\left(0\right)\int_{0}^{\overv}\left(1-F\left(v\right)\right)^{2}dv+\int_{0}^{\overv}Q'\left(r\right)\left(r\left(1-F^{2}\left(r\right)\right)+\int_{r}^{\overv}\left(1-F\left(v\right)\right)^{2}dv\right)dr \nonumber \\
    &=\int_{0}^{\overv}\left(Q'\left(v\right)v\left(1-F^{2}\left(v\right)\right)+Q\left(v\right)\left(1-F\left(v\right)\right)^{2}\right)dv.\nonumber 
\end{align}
To get this we use
\begin{align*}
    \int_{0}^{\overv}\int_{r}^{\overv}Q'\left(r\right)\left(1-F\left(v\right)\right)^{2}dvdr &= \int_{0}^{\overv}\int_{0}^{v}Q'\left(r\right)\left(1-F\left(v\right)\right)^{2}drdv\\
    &=\int_{0}^{\overv}\left(1-F\left(v\right)\right)^{2}\left(Q\left(r\right)-Q\left(0\right)\right)dv.
\end{align*}

Since nature is limited to choose distributions with a mean value of $\mu$, we have the constraint
\[\int_{0}^{\overv}\left(1-F\left(v\right)\right)dv=\mu.\]
So we get the Lagrangian we get is
\begin{equation}
    \calL(Q,F) = \int_{0}^{\overv}\left(Q'\left(v\right)v\left(1-F^{2}\left(v\right)\right)+Q\left(v\right)\left(1-F\left(v\right)\right)^{2}-\lambda\left(1-F\left(v\right)\right)\right)dv. \label{eq:two_bidders:Lagrangian_expresstion}
\end{equation}

\begin{claim} \label{clm:lagrangian}
For every distribution $Q \in \calQ$, and any $\lambda$, distribution $F$ minimizes $\Rev_{F}(Q)$ over $\calF$ if and only if $F$ minimizes $\calL(Q,F)$ over $\calF$.
\end{claim}
\begin{proof}
Since $\int_0^{\overv} (1-F(v))dv = \mu$ for any $F \in \calF$, then
\[\calL(Q,F) = \Rev_{F}(Q) - \lambda \mu.\]
The difference between the two functions does not depend on $F$, and that completes the proof.
\end{proof}

\subsection{The Discrete Case} \label{subsec:two:discrete}
For simpler analysis, we consider a discretization of the problem.
Let $k> 1$ be the discretization parameter. We effectively limit the support of the seller's distribution over reserve prices and the adversary's distribution over values to $k+1$ possible options $\rvec_k=(r_{k,1},\dots, r_{k,k}, r_{k,k+1})$, where $r_{k,i} = \frac{i-1}{k}\cdot \overv$ for every $1\le i \le k+1$.

In the discrete problem, the distribution functions are step functions with jumps only on values of $\rvec_k$, so we can represent the distributions using finite vectors $\qvec, \xvec$ of size $k$, where $q_i = Q(r_{k,i})$ and $x_i = F(r_{k,i})$.
With this, we can represent the revenue as a function of $\qvec, \xvec$.

Similarly to how we got Formula \eqref{eq:two_bidders:Lagrangian_expresstion}, we get
\begin{equation}
    \calL_k(\qvec, \xvec) = \sum_{i=1}^k (q_i-q_{i-1}) r_{k,i}(1-x_i^2)  + \frac{\overv}{k} \cdot \sum_{i=1}^k q_i \left(1 - x_i \right)^2 - \lambda\frac{\overv}{k}\sum_{i=1}^k \left(1-x_i\right), \label{eq:two_bidders:Lagrangian_expresstion_discrete}
\end{equation}
where $q_0 = 0$.

We note that we can also assume that $q_k = 1$. That is because if $q_k < 1$, then there is a non zero probability for a reserve price of zero, which would yield no revenue, because the seller will never sell the item.

We can use Formula~\eqref{eq:two_bidders:Lagrangian_expresstion_discrete} to find equations a saddle point should satisfy.

\begin{claim} \label{clm:two_bidders:disc:nature_induction}
For any $\xvec = (x_1, \dots, x_k) \in [0,\infty)^k$ and for every $1 \le i \le k-1$
\[\frac{\partial}{\partial q_i}\calL_k(\qvec, \xvec) = 0 \iff x_{i+1} = \sqrt{\frac{\overv}{k} \cdot \frac{2 }{r_{k,i} +\nicefrac{\overv}{k}}(x_i - x_i^2) + x_i^2}.\]
\end{claim}
Likewise, we do the same to describe $\qvec$.
\begin{claim} \label{clm:two_bidders:disc:seller_induction}
For any $\xvec \in (0,\overv)^n$, $ \qvec \in (0,\overv)^{n-1}$, and for every $2 \le i \le k$
\[\frac{\partial}{\partial x_i}\calL_k(\qvec, \xvec) = 0 \iff q_{i-1}=q_{i}-\frac{\nicefrac{\lambda}{2}-q_{i}\left(1-x_{i}\right)}{r_{k,i}x_{i}}\cdot\nicefrac{\overv}{k},\]
and
\[\frac{\partial}{\partial x_1}\calL_k(\qvec, \xvec) = 0 \iff q_1=\frac{\lambda}{2(1-x_1)}.\]
\end{claim}

Claims~\ref{clm:two_bidders:disc:nature_induction} and \ref{clm:two_bidders:disc:seller_induction} provide a recurrence relation for a potential saddle point of $\calL_k$.

In order to continue with the proof of Theorem~\ref{thm:two_bidders:main:equilibrium}, we do not need to prove anything more for the discrete case. However, it is also possible to prove the existence of an equilibrium in the discrete case. We do it in Appendix~\ref{appndx:two_bidders_discrete}.

\subsection{Candidate Equilibrium} \label{subsec:two:guess_odes}
Based on Claims~\ref{clm:two_bidders:disc:nature_induction} and \ref{clm:two_bidders:disc:seller_induction}, we are able to guess ODEs that an equilibrium for the non-discrete problem would satisfy.
We do not need to prove the steps to achieve those ODEs, or that solving them is a saddle point of $\calL(Q,F)$. That is because the goal here is only to find candidates for an equilibrium. Only after finding the candidates (see Equations~\eqref{eq:two:F_guess} and \eqref{eq:two:Q_guess}), we prove in Section~\ref{subsec:two:proof} that they do solve Problem~\eqref{eq:problem}.

\paragraph{\bf Guesses for ODEs based on the discrete case.}

We assume that for any $k > 1$, there exists $\qvec_k^*, \xvec_k^*$ that are a saddle point for $\calL_k(\qvec, \xvec)$. We denote the corresponding distributions by $Q^*_k, F^*_k$, which means that $Q^*_k(r_{k,i}) = q_i$, and $F^*_k(r_{k,i}) = x_i$.

Next, we assume that there exist $Q^*, F^*$ such that
\begin{align*}
    \lim_{k\to\infty} Q^*_k &= Q^*, \\
    \lim_{k\to\infty} F^*_k &= F^*.
\end{align*}

Based on Claim~\ref{clm:two_bidders:disc:nature_induction}, we guess that for $F=F^*$
\begin{align*}
    \frac{dF(r)}{dr} &= \lim_{k \to \infty} \frac{F_k^*\left(r+\nicefrac{\overv}{k}\right) - F_k^*(r)}{\nicefrac{\overv}{k}} \\
    &= \lim_{\delta \to 0} \frac{\sqrt{\delta \cdot \frac{2 }{r + \delta}(F(r) - F^2(r)) + F^2(r)} - F(r)}{\delta}.
\end{align*}
This a partial derivative of the function
\begin{align*}
    \psi_F(\delta, r, y) = \sqrt{\delta \cdot \frac{2 }{r +\delta}(y - y^2) + y^2}.
\end{align*}
This is the function that describes the recursive relation in Claim~\ref{clm:two_bidders:disc:nature_induction}, since according to the claim, $F_k^*(r_{k,i+1}) = \psi_F\left( \nicefrac{\overv}{k}, r_{k,i}, F_k^*(r_{k,i}) \right)$.

This yields the ODE
\begin{align}
    \frac{d}{dr}F^*(r) &= \frac{\partial\psi_F(\delta, r, y)}{\partial \delta} \left(0, r, F^*(r)\right) \nonumber \\
    &= \frac{(1-F^*(r))F^*(r)}{rF^*(r)} \nonumber \\
    &= \frac{1-F^*(r)}{r}. \label{eq:two_bidders:ode:nature}
\end{align}
This only applies when $F^*(r) \ne 0$.\footnote{For the case where $F^*(r) = 0$, we can see that $\psi(\delta, r, 0) = 0$, and thus the derivative would be zero as well.}

We can use a similar idea for $Q^*$.
We define
\[\psi_Q(\delta, r, y,z) = y-\frac{\nicefrac{\lambda}{2}-y \left(1-z\right)}{rz}\cdot\delta.\]
According to Claim~\ref{clm:two_bidders:disc:seller_induction} 
\[Q^*_k\left(r_{k,i} - \nicefrac{\overv}{k}\right) = \psi_Q\left(\nicefrac{\overv}{k}, r_{k,i}, Q^*_k(r_{k,i}), F^*_k(r_{k,i}) \right).\]
And we guess that
\begin{align}
    \frac{dQ^*(r)}{dr} &= \lim_{k \to \infty} \frac{Q_k^*\left(r-\nicefrac{\overv}{k}\right) - F_k^*(r)}{-\nicefrac{\overv}{k}} \nonumber \\
    &= \lim_{\delta \to 0} \frac{\psi_Q\left(\delta, r, Q(r), F^*(r) \right) - \psi_Q\left(0, r, Q(r), F^*(r) \right)}{-\delta}. \nonumber \\
    &= -\frac{\partial\psi_Q(\delta, r, y, z)}{\partial \delta} \left(0, r, Q^*(r), F^*(r)\right) \nonumber \\
    &= \frac{\nicefrac{\lambda}{2}-Q^*(r) \left(1-F^*(r)\right)}{rF^*(r)}. \label{eq:two_bidders:ode:seller}
\end{align}

\paragraph{\bf Solution to the ODEs.}
A solution for Eq~\eqref{eq:two_bidders:ode:nature} is of the form
\[F^*(v) = 1 - \frac{c}{v},\]
everywhere where $F^*(v) \ne 0$.
With the assumptions of continuity over $[0, \overv)$, and a mean value of $\mu$, we get
\begin{equation}
    F^*(v) = \begin{cases}
0 & v \in [0, \underv), \\
1 - \nicefrac{\underv}{v} & v\in [\underv, \overv), \\
1 & v \ge \overv,
\end{cases} \label{eq:two:F_guess}
\end{equation}
where $\underv\left(1+\log\overv-\log\underv\right) = \mu$.

Now, we can solve for the seller's distribution.
Since the ODE is not defined where $F^*(r) = 0$, it is only defined on $(\underv, \overv]$.
We substitute $F^*(v) = 1 - \nicefrac{\underv}{v}$ in Eq~\eqref{eq:two_bidders:ode:seller}, and we get
\[\frac{dQ(r)}{dr} = \frac{\nicefrac{\lambda}{2}-Q^*(r) \cdot \nicefrac{\underv}{r}}{r - \underv}.\]
The solution to this ODE is
\[Q^*(r) = \left( \frac{r}{r-\underv} \right) \cdot \left( c + \nicefrac{\lambda}{2} \cdot \log r \right),\]
for an unknown $c$.

We know that the seller would never set a reserve price of $\overv$, so the continuity on $\overv$ is known.
Since $Q^*(\overv) = 1$, we can solve $c = 1 - \nicefrac{\underv}{\overv} - \nicefrac{\lambda}{2} \log \overv$, and thus
\[Q^*(r) = \left( \frac{r}{r-\underv} \right) \cdot \left( 1 - \nicefrac{\underv}{\overv} - \nicefrac{\lambda}{2} \cdot \log \nicefrac{\overv}{r} \right).\]

We also know that as a distribution function, $Q^*(r) \in [0, 1]$ for every $r$.
Since $\lim_{r \to \underv^+}\left( \frac{r}{r-\underv} \right) = \infty$, we know that
\begin{align}
    0 &= \lim_{r \to \underv^+} \left( 1 - \nicefrac{\underv}{\overv} - \nicefrac{\lambda}{2} \cdot \log \nicefrac{\overv}{r} \right) \nonumber \\
    & =  1 - \nicefrac{\underv}{\overv} - \nicefrac{\lambda}{2} \cdot \log \nicefrac{\overv}{\underv}, \nonumber \\
    \implies \nicefrac{\lambda}{2} &= \frac{1 - \nicefrac{\underv}{\overv}}{\log \nicefrac{\overv}{\underv}}. \label{eq:two:lambda}
\end{align}

This suggests that for $r\in (\underv, \overv]$
\[Q^*(r) = \left( 1 - \nicefrac{\underv}{\overv} \right) \cdot \left( \frac{r}{r-\underv} \right) \cdot \left( \frac{\log r - \log \underv}{ \log \overv - \log \underv} \right).\]
Now we only need to solve for $r\in [0, \underv]$.
From Claim~\ref{clm:two_bidders:disc:seller_induction}, $Q^*_k(0) = \frac{\lambda}{2\left(1-F_k^*(0)\right)}$.
Therefore
\begin{align*}
    Q^*(0) &= \frac{\lambda}{2\left(1-F^*(0)\right)} \\
    &= \frac{1 - \nicefrac{\underv}{\overv}}{\log \nicefrac{\overv}{\underv}}.
\end{align*}
Likewise $\lim_{r\to \underv^+} Q^*(r) = \frac{1 - \nicefrac{\underv}{\overv}}{\log \nicefrac{\overv}{\underv}}$, and since $Q^*(r)$ is a distribution function, it is non-decreasing.
Therefore, our guess for $Q^*$ is
\begin{equation}
    Q^*(r) = \begin{cases}
         \left( 1 - \nicefrac{\underv}{\overv} \right) \cdot \left(\frac{1}{\log \overv - \log \underv}\right) & r\in [0, \underv],\\
        \left( 1 - \nicefrac{\underv}{\overv} \right) \cdot \left( \frac{r}{r-\underv} \right) \cdot \left( \frac{\log r - \log \underv}{ \log \overv - \log \underv} \right) & r\in (\underv, \overv].
    \end{cases} \label{eq:two:Q_guess}
\end{equation}

\subsection{Proof of Theorem~\ref{thm:two_bidders:main:among_all_truthful}} \label{subsec:two:proof}
We have guessed that $Q^*, F^*$ (Equations~\eqref{eq:two:F_guess}-\eqref{eq:two:Q_guess}) are an equilibrium for Problem~\eqref{eq:problem_reduced}. We will prove this guess directly.

\begin{claim} \label{clm:two:F_is_optimal}
If the seller uses a second price auction with a randomized reserve price sampled from $Q^*$, then $F^*$ minimizes the expected revenue of the seller among all distributions with a mean value of $\mu$, and support contained in $[0, \overv]$.
\end{claim}
\begin{proof}
We have previously found the Lagrangian for this problem in Eq~\eqref{eq:two_bidders:Lagrangian_expresstion}.
We will minimize it point-wise, which means we need to minimize the integrand as a function of $F(v)$, which we denote by $h_v(z)$. So
\begin{align*}
    h_{v}\left(z\right)=Q'\left(v\right)v\left(1-z^{2}\right)+Q\left(v\right)\left(1-z\right)^{2}-\lambda\left(1-z\right).
\end{align*}
We use $\lambda=\frac{2\left(1-\nicefrac{\underv}{\overv}\right)}{\log\nicefrac{\overv}{\underv}}$ as in Eq~\eqref{eq:two:lambda}.
When $v\in (\underv, \overv]$ we use $Q\left(v\right) = \left(1-\nicefrac{\underv}{\overv}\right)\cdot\left(\frac{v}{v-\underv}\right)\cdot\left(\frac{\log v-\log\underv}{\log\overv-\log\underv}\right)$, which leads to
\begin{align*}
    h_{v}\left(z\right)=\left(\frac{1-\nicefrac{\underv}{\overv}}{\log\nicefrac{\overv}{\underv}}\right)\left(\frac{v}{v-\underv}\right)\left(\left(1-\frac{\underv\log\nicefrac{v}{\underv}}{v-\underv}\right)\left(1-z^{2}\right)+\left(\log\nicefrac{v}{\underv}\right)\left(1-z\right)^{2}-2\left(1-\nicefrac{\underv}{v}\right)\left(1-z\right)\right).
\end{align*}
After deriving by $z$, and searching for a root of the derivative, we get
\begin{align*}
    \left(1-\frac{\underv\log\nicefrac{v}{\underv}}{v-\underv}\right)z+\left(\log\nicefrac{v}{\underv}\right)\left(1-z\right)&=1-\nicefrac{\underv}{v}, \\
    \implies  z &= 1 - \nicefrac{\underv}{v}.
\end{align*}

We can also verify that the second derivative is positive, which proves that $F(v) = 1 - \nicefrac{\underv}{v}$ does minimizes the integrand point-wise, for $v\in (\underv, \overv]$.

When $v\in [0, \underv)$, we get that
\begin{align*}
    h_{v}\left(z\right) = \left(\frac{1-\nicefrac{\underv}{\overv}}{\log\nicefrac{\overv}{\underv}}\right)\left(z^{2}-1\right),
\end{align*}
which clearly has a minimum at $z = 0$, proving that $F(v) = F^*(v)$ minimizes the Lagrangian and thus the expected revenue among all distributions with support contained in $[0, \overv]$ and a mean value of $\mu$.
\end{proof}

\begin{claim} \label{clm:two:Q_is_optimal}
If the bidders' valuation is sampled from $F^*$, using a second-price auction with a randomized reserve price sampled from $Q^*$ maximizes the expected revenue among all truthful mechanisms.
\end{claim}
\begin{proof}
$F^*$ coincides with the equal-revenue distribution on $[\underv, \overv)$, which means the \emph{virtual valuation} there is always zero. 
Therefore, since by \cite{Myerson81} the expected revenue of a truthful mechanism is equal to the expected virtual welfare, it is fully determined by the allocation rule in the case where the value is exactly $\overv$. The maximum revenue is achieved if the item is always allocated to a seller with value $\overv$ if one exists.
Since this is the case for a second price auction with randomized price sampled from $Q^*$, then it maximizes the expected revenue among all truthful mechanisms.
\end{proof}

Claims \ref{clm:lagrangian}, \ref{clm:two:F_is_optimal} and \ref{clm:two:Q_is_optimal} simultaneously complete the proof for both Theorem \ref{thm:two_bidders:main:equilibrium} and Theorem \ref{thm:two_bidders:main:among_all_truthful}.

\section{Beyond Two Bidders} \label{sec:multi_bidders}

In this section we assume that are $n$ bidders, and find the optimal pricing distribution the seller can use in a second price auction with a randomized reserve price.

\begin{theorem}[Solution of Problem~\eqref{eq:problem_reduced} for $n\ge3$ bidders] \label{thm:multi:main:equilibrium}
Consider $n\ge 3$ bidders.
There exists an equilibrium $\left(Q^*, F^*\right) \in \calQ \times \calF$ that solves Problem \eqref{eq:problem_reduced}.
Let $\underv$ be the smaller among the two solutions to
\[\left(n-\frac{1}{n-1}+\log\nicefrac{\overv}{\underv}\right)\underv=\left(n-1\right)^{2}\mu.\]
If $\underv \le \overv$ then the equilibrium is described by:
\begin{align*}
    F^*(v) &= \begin{cases}
    0 & v \in \left[0, v_0\right), \\
    1-\frac{1}{\left(n-1\right)^2 } & v\in \left[v_0, \underv \right), \\
    1-\frac{v_0}{v \left(n-1\right) } & v\in [\underv, \overv), \\
    1 & v \ge \overv,
    \end{cases} \\
    Q^*(r) &= \begin{cases}
        \left(1-\frac{v_0}{\overv \left(n-1\right) }\right)^{n-1} \left( \frac{\left(n-1\right)^2}{\left(n-1\right)^2-1} \right)^{n-1} \left(\frac{n - 1 - \frac{1}{n-1}}{n - 1 - \frac{1}{n-1} + \log\nicefrac{\overv}{\underv}}\right) & r\in [0, \underv), \\
        \left(1-\frac{v_0}{\overv \left(n-1\right) }\right)^{n-1} \left( \frac{\left(n-1\right)r}{\left(n-1\right)r-v_0} \right)^{n-1} \left(\frac{n - 1 - \frac{1}{n-1} + \log\nicefrac{r}{\underv} }{n - 1 - \frac{1}{n-1} + \log\nicefrac{\overv}{\underv}}\right) & r \in [\underv, \overv],
    \end{cases}
\end{align*}
where $v_0 = \frac{\underv}{n-1}$.

If $\underv \ge \overv$, the equilibrium is described by
\begin{align*}
    F^*(v) &= \begin{cases}
    0 & v \in \left[0, v_0 \right), \\
    1-\frac{1}{(n-1)^2} & v\in [v_0, \overv), \\
    1 & v \ge \overv,
    \end{cases} \\
    Q^*(r) &= 1 \quad r \ge 0,
\end{align*}
where $v_0 = \mu-\frac{\overv-\mu}{\left(n-1\right)^{2}-1}$.
\end{theorem}
Note that for $\underv=\overv$ the two expressions for $(F^*,Q^*)$ coincide.

This theorem is true for $n\ge3$, since for $n=2$, $Q^*(r)$ as defined here is not defined for $r\in[0,\underv]$. However, the theorem does agree with Theorem~\ref{thm:two_bidders:main:equilibrium} over $(\underv, \overv]$, which means that the limit $\lim_{r\to \underv^+}Q^*(r)$ is the same as in Theorem~\ref{thm:two_bidders:main:equilibrium}.

\begin{figure}[t]
  \centering
  \begin{minipage}[t]{0.49\textwidth}
    \includegraphics[width=\textwidth]{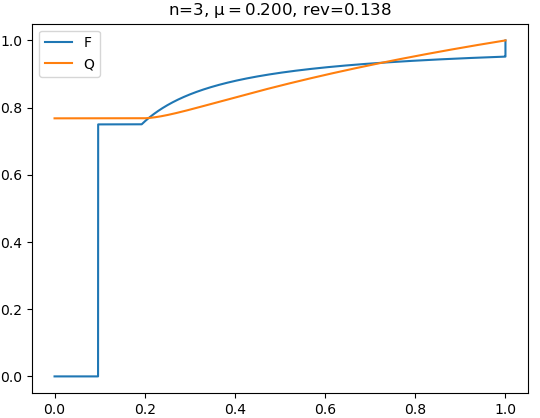}
    \caption{
        The solution $(F,Q)$ to Problem~\eqref{eq:problem_reduced} for the parametric settings in which $n=3, \mu=0.3$ and $\overv=1$.
        $F$ represents that each of the bidders has a probability of $\nicefrac{3}{4}$ to have a valuation of $v_0=0.0965$, and otherwise sample a valuation from a generalized truncated equal-revenue distribution between $\underv = 0.193$ to $1$.
        $Q$ is a ratio of a truncated log-uniform distribution and an equal-revenue distribution to the power of $n-1$, with the same support.
    }
    \label{fig:multi:low_mu}
  \end{minipage}
  \hfill
  \begin{minipage}[t]{0.49\textwidth}
    \includegraphics[width=\textwidth]{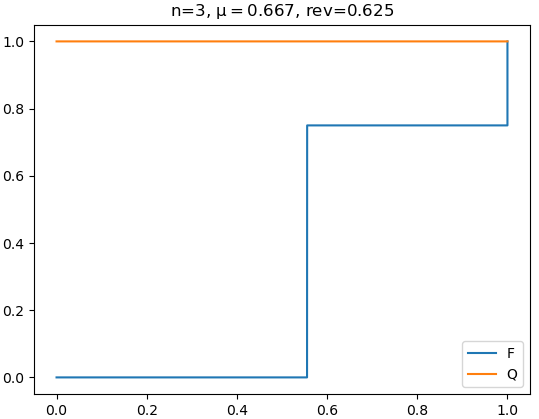}
    \caption{
        The solution to Problem~\eqref{eq:problem_reduced} for the parametric settings in which $n=3, \mu=\nicefrac{2}{3}$ and $\overv=1$.
        In this case $\mu$ is sufficiently large that the seller would always choose a reserve price of 0, and each of the bidders randomizes between $v_0=\nicefrac{5}{9}$ with probability $0.75$, and $\overv=1$.
    }
    \label{fig:multi:high_mu}
  \end{minipage}
\end{figure}

Unlike the case where there are two bidders that we have solved with Theorem~\ref{thm:two_bidders:main:equilibrium}, or the case with a single bidder solved by~\cite{CarrascoFK+18}, the adversary's optimal distribution is not a truncated equal revenue distribution, as it has a point of discontinuity at $v_0 \in (0, \overv)$.
As a result, this mechanism is not the optimal among all truthful mechanisms; we expand upon this point in Section~\ref{subsec:multi:not_optimal}.
The theorem also shows that depending on $n$, if the known upper bound is small enough (or equivalently, the mean value is large enough), then the optimal pricing distribution is to always choose a reserve price of zero, which simply results in a standard second price auction.
When the upper bound is large enough, the seller would do better by randomizing between a reserve price of zero, and a reserve price sampled from a truncated log-uniform distribution divided by $(F^*)^{n-1}$. We illustrate the optimal distributions in Figures~\ref{fig:multi:low_mu} and \ref{fig:multi:high_mu}.

\vspace{0.05in}
\noindent{\bf Proof overview.}
The proof starts in a similar way to the proof of Theorem~\ref{thm:two_bidders:main:equilibrium} by developing the Lagrangian (Section~\ref{subsec:multi:lagrangian}).
In Section~\ref{subsec:multi:discrete} we define the discrete case, and use it to find ODEs that an equilibrium would satisfy.
However, in this case, assuming continuity will result with a non monotone CDF, which is impossible (see failed attempt described in Section~\ref{subsec:multi:failed}).
This leads us to guess equilibrium distributions, which are not continuous, but depend on some unknown parameters (also in Section~\ref{subsec:multi:failed}).
After guessing the potential equilibrium, in (Section~\ref{subsec:multi:candidate}) we set the unknown parameters in such a way that will allow us to prove that both the seller and nature are playing a best response to the other player's chosen distribution (the proof appears in Section~\ref{subsec:multi:proof}).

\subsection{The Lagrangian for Maximizing Expected Revenue} \label{subsec:multi:lagrangian}
We write the formula for the expected revenue, with $n$ bidders.

\begin{claim} \label{clm:multi:revenue_expression}
For any $F \in \calF$, and distribution $Q \in \calQ$ that is continuous over $(0,\overv]$, the revenue is expressed by
\[\Rev_{F}(Q) = \int_{0}^{\overv}\left(Q'\left(v\right)v\left(1-F^{n}\left(v\right)\right)+Q\left(v\right)\left(1-nF^{n-1}\left(v\right)+\left(n-1\right)F^{n}\left(v\right)\right)\right)dv.\]
\end{claim}

Since nature is limited to a known expected value, we will minimize the Lagrangian instead, where the Lagrangian is defined as
\begin{align}
    \calL(Q,F)=& \Rev_{F}(Q) - \lambda\int_{0}^{\overv} \left(1 - F(v)\right) dv. \label{eq:multi:lagrangian}
\end{align}
Similarly to the two bidder case, it is enough to solve the problem on $\calL$ instead of $\Rev$.

\subsection{ODEs through the Discrete Case} \label{subsec:multi:discrete}
Like in the two bidder case, we consider a discretization of the problem.
Let $k> 1$ be the discretization parameter. We limit the support of the seller's distribution over reserve prices and the adversary's distribution over values to $k+1$ possible options $\rvec_k=(r_{k,1},\dots, r_{k,k}, r_{k,k+1})$, where $r_{k,i} = \frac{i-1}{k}\cdot \overv$ for every $1\le i \le k+1$.

Like in the two bidder case, we can represent the distributions using finite vectors $\qvec, \xvec$ of size $k$, where $q_i = Q(r_{k,i})$ and $x_i = F(r_{k,i})$. The Lagrangian we get is
\begin{align*}
    \calL_k(\qvec, \xvec) = \sum_{i=1}^{k}\left(\left(q_{i}-q_{i-1}\right)r_{k,i}\left(1-x_{i}^{n}\right)+\nicefrac{\overv}{k}\cdot q_{i}\left(1-nx_{i}^{n-1}+\left(n-1\right)x_{i}^{n}\right)-\nicefrac{\overv}{k}\cdot\lambda\left(1-x_{i}\right)\right).
\end{align*}
 Since we are looking for a saddle point, we find conditions for the derivative to be zero.
\begin{claim} \label{clm:multi:disc:nature_induction}
For any $\xvec = (x_1, \dots, x_k) \in [0,\infty)^k$ and for every $1 \le i \le k-1$, the Lagrangian for $n$ bidders satisfies
\[\frac{\partial}{\partial q_i}\calL_k(\qvec, \xvec) = 0 \iff x_{i+1} = \sqrt[n]{\frac{n}{r_{k,i}+\nicefrac{\overv}{k}}\left(x_{i}^{n-1}-x_{i}^{n}\right)\cdot\frac{\overv}{k}+x_{i}^{n}}.\]
\end{claim}

\begin{claim} \label{clm:multi:disc:seller_induction}
For any $\xvec \in (0,\overv)^n$, $ \qvec \in (0,\overv)^{n-1}$, and for every $2 \le i \le k$, the Lagrangian for $n$ bidders satisfies
\[\frac{\partial}{\partial x_i}\calL_k(\qvec, \xvec) = 0 \iff q_{i-1}=q_{i}-\frac{\nicefrac{\lambda}{n}-q_{i}\left(n-1\right)x_{i}^{n-2}\left(1-x_{i}\right)}{r_{k,i}x_{i}^{n-1}}\cdot\nicefrac{\overv}{k}.\]
\end{claim}
We do not prove these claims, since the proofs are almost identical to the proofs of Claims~\ref{clm:two_bidders:disc:nature_induction}-\ref{clm:two_bidders:disc:seller_induction}.

Like in the two bidder case, we find ODEs based on Claims~\ref{clm:multi:disc:nature_induction}-\ref{clm:multi:disc:seller_induction}.
We define
\[\psi_F(\delta, v, y) = \sqrt[n]{\frac{n}{v+\delta}\left(y^{n-1}-y^{n}\right)\delta+y^{n}},\]
and we get the ODE
\begin{align}
    \frac{d}{dr}F(r) &= \frac{\partial\psi_F(\delta, r, y)}{\partial \delta} \left(0, r, F(r)\right) \nonumber \\
    &= \frac{1-F(r)}{r}. \label{eq:multi:ode:nature}
\end{align}

Likewise, we define
\begin{align*}
    \psi_Q(\delta, r, y, z) = y-\frac{\nicefrac{\lambda}{n}-y\left(n-1\right)z^{n-2}\left(1-z\right)}{rz^{n-1}}\cdot\delta,
\end{align*}
and we get the ODE
\begin{align}
    \frac{d}{dr}Q(r) &= - \frac{\partial\psi_Q(\delta, r, y, z)}{\partial \delta} \left(0, r, Q(r), F(r)\right) \nonumber \\
    &= \frac{\lambda}{n r F^{n-1} (r)} - \frac{\left(n-1\right)Q(r)\left(1-F(r)\right)}{r F(r)}. \label{eq:multi:ode:seller}
\end{align}

\subsection{Failed Attempt and Intuition}
\label{subsec:multi:failed}
Here we attempt the same approach as in the case of two bidders and show that it fails. Through the failed attempt we gain intuition for the correct shape of the equilibrium.

We assume first that there exist $Q^*, F^*$ that are an equilibrium for Problem~\eqref{eq:problem_reduced}, and they solve the ODEs~\eqref{eq:multi:ode:nature}-\eqref{eq:multi:ode:seller}, everywhere where their derivative is positive.\footnote{The reason we don't assume the ODEs are correct when the derivative is zero, is that a negative derivative may produce better results, but it's impossible for a distribution function, which mean that the derivative in the discrete case may not be zero.}

If $F^*(v)$ is continuous over $[0,\overv)$ and differentiable almost everywhere on the same interval, then
\[F^*(v) = \begin{cases}
0 & v \in [0, \underv], \\
1-\nicefrac{\underv}{v} & v\in [\underv, \overv) ,\\
1 & v\ge\overv,
\end{cases} \]
where $\underv$ is defined such that the mean value would be $\mu$.
If that is the case, then we can solve for $Q^*$ similarly to how we solve it in the two bidder case, and we get
\begin{align*}
    Q^*(r) = \left(\frac{r}{r-\underv}\right)^{n-1}\left(\left(1-\nicefrac{\underv}{\overv}\right)^{n-1}+\frac{\lambda\log\nicefrac{r}{\overv}}{n}\right),
\end{align*}
for any $r\in (\underv, \overv]$. However, even if we use the $\lambda$ that solves $\left(1-\nicefrac{a}{\overv}\right)^{n-1}+\frac{\lambda\log\nicefrac{r}{\overv}}{n} = 0$, we still end up with
\[\lim_{r\to \underv^+} \left(\frac{r}{r-\underv}\right)^{n-1}\left(\left(1-\nicefrac{\underv}{\overv}\right)^{n-1}+\frac{\lambda\log\nicefrac{r}{\overv}}{n}\right) = \infty.\]

This is impossible, which suggests that our assumptions are incorrect. The most ``suspicious'' assumption is the continuity of $F^*$ over $[0, \overv)$.
To fix this, we guess that $F^*$ has mass on a point $v_0$, which together with the ODE~\eqref{eq:multi:ode:nature}, results with
\begin{align}
    F^*(v) &= \begin{cases}
    0 & v\in [0, v_0),\\
    1 - \nicefrac{a}{\underv} & v\in [v_0, \underv), \\
    1 - \nicefrac{a}{v} & v\in [\underv, \overv), \\
    1 & v\ge\overv,
    \end{cases} \label{eq:multi:F_guess}
\end{align}
for some unknown $0 \le a, v_0, \underv$.
When we use the ODE~\eqref{eq:multi:ode:seller} over $[\underv, \overv]$, with $Q^*(\overv) = 1$, and $F^*(v) = 1 - \nicefrac{a}{v}$, we get
\[Q^*(r) =  \left(\frac{r}{r-a}\right)^{n-1}\left(\left(1-\nicefrac{a}{\overv}\right)^{n-1}+\frac{\lambda}{n}\log\nicefrac{r}{\overv}\right).\]
We guess that similarly to the two bidder case, $Q^*$ has no mass in $(0,\underv)$, which means that our guess is
\begin{equation}
    Q^*(r) = \begin{cases} 
     \left(\frac{\underv}{\underv-a}\right)^{n-1}\left(\left(1-\nicefrac{a}{\overv}\right)^{n-1}+\frac{\lambda}{n} \log\nicefrac{\underv}{\overv} \right) & r \in [0,\overv), \\
     \left(\frac{r}{r-a}\right)^{n-1}\left(\left(1-\nicefrac{a}{\overv}\right)^{n-1}+\frac{\lambda}{n}\log\nicefrac{r}{\overv}\right) & r\in [\underv, \overv],
    \end{cases} \label{eq:multi:Q_guess}
\end{equation}
where $a, \underv$ are the same as in $F^*$, and $\lambda$ is unknown.

Our next goal is to find values of $a, v_0, \underv$ and $\lambda$ such that those formulas will indeed describe an equilibrium. We will also find when such values don't exist, and prove that it only happens when a deterministic reserve price of zero is the optimal pricing distribution for the second price auction.

\subsection{Candidate Equilibrium}
\label{subsec:multi:candidate}
We will now find $a,v_0, \underv, \lambda$ such that $(Q^*,F^*)$ can be an equilibrium. We will find equations that an equilibrium must satisfy. Afterwards, we will be able to prove that $(Q^*,F^*)$ is indeed an equilibrium.

Assume that there exist $a, v_0, \underv$, and $\lambda$ such that $(Q^*,F^*)$ described by formulas \eqref{eq:multi:F_guess} and \eqref{eq:multi:Q_guess} is an equilibrium. 
The next three claims (Claims \ref{clm:multi:candidate:mean}, \ref{clm:multi:candidate:equal_rev}, \ref{clm:multi:candidate:integrand_min}) provide four equations which the variables must satisfy, allowing us to find their values.

\begin{claim} \label{clm:multi:candidate:mean}
$a,v_0$ and $\underv$ satisfy
\[a\left(1-\frac{v_{0}}{\underv}-\log\nicefrac{\underv}{\overv}\right)=\mu-v_{0}.\]
\end{claim}
\begin{proof}
This is equivalent to the constraint
\[\int_0^{\overv}(1-F^*(v))dv = \mu\]
(using the definition of $F^*(v)$ according to~\eqref{eq:multi:F_guess}).
\end{proof}

If $(Q^*, F^*)$ is an equilibrium, that means that the seller cannot find a better distribution than $Q^*$ if we fix $F^*$.
Since $\Rev_{F^*}(Q^*) = E_{r\sim Q^*}\left[ \Rev_{F^*}(r) \right]$, and we know that $\Rev_{F^*}(Q^*) \ge \Rev_{F^*}(r)$ for every $r$, than we can conclude that for every $r_1, r_2$ in the support of $Q^*$
\begin{equation*}
    \Rev_{F^*}(r_1) = \Rev_{F^*}(r_2).
\end{equation*}

This requirement yields another equation according to the following claim.

\begin{claim} \label{clm:multi:candidate:equal_rev}
The following statements are equivalent
\begin{itemize}
    \item $\Rev_{F^*}(r_1) = \Rev_{F^*}(r_2)$ for every $r_1, r_2$ in the support of $Q^*$.
    \item $a,v_0$ and $\underv$ satisfy
\[a\left(n-\left(n-1\right)\nicefrac{v_{0}}{\underv}\right)=v_{0}.\]
\end{itemize}
\end{claim}
\begin{proof}
If we fix $F = F^*$ as described in \eqref{eq:multi:F_guess} by $a,v_0$ and $\underv$, the seller's best pricing distribution would be any distribution over $\argmax_r \Rev_{F}(r)$.

On $[\underv, \overv)$, we know that $F(v) = 1-\frac{a}{v}$, and thus
\[\frac{d}{dr} r(1-F^n(r)) = 1-nF^{n-1}\left(v\right)+\left(n-1\right)F^{n}\left(v\right).\]
Therefore, for every $r\in [\underv, \overv)$
\begin{align*}
    \Rev_{F}(r) &= r\left(1-F^{n}\left(r\right)\right)+\int_{r}^{\overv}\left(1-nF^{n-1}\left(v\right)+\left(n-1\right)F^{n}\left(v\right)\right)dv \\
    &= r\left(1-F^{n}\left(r\right)\right)+ \left( \overv\left(1-\lim _{v\to \overv}F^{n}\left(v\right)\right) - r\left(1-F^{n}\left(r\right)\right) \right) \\
    &= \overv\left(1 - \left(1 - \nicefrac{a}{\overv}\right)^n \right).
\end{align*}
This means that $\Rev_{F}(r_1) = \Rev_{F}(r_2)$ for any $r_1, r_2 \in [\underv, \overv)$.

Now, we only need to solve for $\Rev_{F}(0) = \Rev_{F}(\underv)$.
This is equivalent to
\[\int_0^{\underv} \left( 1-nF^{n-1}\left(v\right)+\left(n-1\right)F^{n}\left(v\right) \right) = \underv\left(1-\left(1-\frac{a}{\underv}\right)^{n}\right),\]
which is
\[v_{0}+\left(\underv-v_{0}\right)\left(1-n\left(1-\frac{a}{\underv}\right)^{n-1}+\left(n-1\right)\left(1-\frac{a}{\underv}\right)^{n}\right) = \underv\left(1-\left(1-\frac{a}{\underv}\right)^{n}\right).\]
This is equivalent to
\[a\left(n-\left(n-1\right)\nicefrac{v_{0}}{\underv}\right)=v_{0},\]
which completes the proof.
\end{proof}

To make sure that $F^*$ is nature's best response to $Q^*$, we will prove that it minimizes the integrand 
$$
Q'\left(v\right)v\left(1-F^{n}\left(v\right)\right)+Q\left(v\right)\left(1-nF^{n-1}\left(v\right)+\left(n-1\right)F^{n}\left(v\right)\right) -\lambda(1-F(v))
$$
of the Lagrangian in \eqref{eq:multi:lagrangian} point-wise, where $Q=Q^*$. The next claim yields two additional equations that make sure $F^*$ minimizes the integrand for every $v \in [0, \underv)$.

\begin{claim} \label{clm:multi:candidate:integrand_min}
If $a, v_0, \underv,$ and $\lambda$ satisfy
\begin{itemize}
    \item $\frac{\lambda}{Q^*(0)} = \frac{n}{n-1}\left(1-\frac{1}{\left(n-1\right)^{2}}\right)^{n-2}$,
    \item $a\left(n-1\right)^{2}=\underv$,
\end{itemize}
then $F=F^*$ minimizes
\[1-nF^{n-1}\left(v\right)+\left(n-1\right)F^{n}\left(v\right) -\frac{\lambda}{Q^*(0)}\left(1-F\left(v\right)\right),\]
for every $v\in[0, \underv]$.
\end{claim}

\begin{proof}
According to Lemma 3 in \cite{Suzdaltsev20a}, if $\frac{\lambda}{Q(0)} = \frac{n}{n-1}\left(1-\frac{1}{\left(n-1\right)^{2}}\right)^{n-2}$, then both $F(v) = 0$, and $F(v) = 1-\frac{1}{(n-1)^2}$ minimize the integrand pointwise over $F(v) \in [0,1]$.
So, when $\frac{\lambda}{Q(0)} = \frac{n}{n-1}\left(1-\frac{1}{\left(n-1\right)^{2}}\right)^{n-2}$, we only need to make sure that $F(v) \in \left\{0, 1-\frac{1}{(n-1)^2}\right\}$ for every $v\in [0, \underv]$.

Since $F(v) = 0$ for every $v\in[0,v_0)$, we only demand that $F(v) = 1-\frac{1}{(n-1)^2}$ for every $v\in[v_0, \underv]$, which yields $a\left(n-1\right)^{2}=\underv$.
\end{proof}

Combining Claims \ref{clm:multi:candidate:mean}, \ref{clm:multi:candidate:equal_rev}, \ref{clm:multi:candidate:integrand_min} we get four equations that $a, v_0, \underv, \lambda$ must satisfy. This allows us to solve for those variables.
We get that $\underv$ must solve
\begin{equation} \label{eq:multi:underv_equation}
    \left(n-\frac{1}{n-1}+\log\nicefrac{\overv}{\underv}\right)\underv=\left(n-1\right)^{2}\mu.
\end{equation}
This has a unique solution in $(0,\overv]$ only if $\left(n-\frac{1}{n-1}\right)\overv > \left(n-1\right)^{2}\mu$, otherwise, it has no solutions in $(0,\overv]$.
When it does have a solution, the rest of the variables are described by
\begin{align}
    v_0 &= \frac{\underv}{n-1}, \label{eq:multi:v_0_equation}\\
    a &= \frac{\underv}{(n-1)^2}, \label{eq:multi:a_equation}\\
    \lambda &= \frac{n\left(1-\nicefrac{a}{\overv}\right)^{n-1}}{n-1-\frac{1}{n-1}+\log\nicefrac{\overv}{\underv}}. \label{eq:multi:lambda_equation}
\end{align}

Now, when Eq~\eqref{eq:multi:underv_equation} does have a solution in $(0, \underv]$ we have a candidate for an equilibrium.

\subsection{Proof of Theorem \ref{thm:multi:main:equilibrium}}
\label{subsec:multi:proof}
When we use Equations \eqref{eq:multi:underv_equation}-\eqref{eq:multi:lambda_equation} and Formulas~\eqref{eq:multi:Q_guess}-\eqref{eq:multi:F_guess}, we get the same distributions as in Theorem~\ref{thm:multi:main:equilibrium}, assuming $\underv \le \overv$. 
We now prove Claims~\ref{clm:multi:valid} and~\ref{clm:multi:candidate:integrand_min_underv_to_overv} that establish these are valid distributions that form an equilibrium. The complementary case $\underv \ge \overv$ is addressed in Claim~\ref{clm:multi:proof_when_underv_ge_overv}. Combining these claims proves the theorem.

\begin{claim}
\label{clm:multi:valid}
If $\underv \in [0, \overv]$ solves Eq~\eqref{eq:multi:underv_equation}, and $a, v_0, \lambda$ are defined by Equations \eqref{eq:multi:v_0_equation}-\eqref{eq:multi:lambda_equation},
then $F^*,Q^*$ defined as in Formulas \eqref{eq:multi:Q_guess}, \eqref{eq:multi:F_guess}, are non decreasing and their image is contained in $[0,1]$.
\end{claim}

\begin{proof}
Since $0<a<\underv$, then $F^*$ is non decreasing and its image is contained in $[0,1]$.
When using these values for $Q^*$, we get (when $r\in [\underv, \overv]$)
\begin{equation}
    Q^*(r) = \left(1-\nicefrac{a}{\overv}\right)^{n-1}\left(\frac{r}{r-a}\right)^{n-1}\left(\frac{n-1-\frac{1}{n-1}+\log\nicefrac{r}{\underv}}{n-1-\frac{1}{n-1}+\log\nicefrac{\overv}{\underv}}\right).
\end{equation}
The derivative over $[\underv, \overv]$ is
\[\frac{Q^*(r)}{dr} = -\left(n-1\right)\left(1-\nicefrac{a}{\overv}\right)^{n-1}\left(\frac{r}{r-a}\right)^{n}\frac{\frac{\underv}{n-1}\log\nicefrac{r}{\underv}+\underv-r}{r^{2}\left(\left(n-2\right)n+\left(n-1\right)\log\nicefrac{\overv}{\underv}\right)}.\]

To show that this is non negative, we only need to prove that
\[\frac{\underv}{n-1}\log\nicefrac{r}{\underv}+\underv-r \le 0.\]

This is equivalent to
\[\frac{\log\nicefrac{r}{\underv}}{n-1}+1-\nicefrac{r}{\underv} \le 0,\]
which is true for $r \in [\underv, \overv]$, for any $n \ge 3$.
\end{proof}

The former claim proves that $Q^*$ and $F^*$ are distribution functions. Now we will prove that they also solve Problem~\eqref{eq:problem_reduced}.

\begin{claim} \label{clm:multi:candidate:integrand_min_underv_to_overv}
If $\underv \in [0, \overv]$ solves Eq~\eqref{eq:multi:underv_equation}, and $a, v_0, \lambda$ are defined by Equations \eqref{eq:multi:v_0_equation}-\eqref{eq:multi:lambda_equation},
and $F^*,Q^*$ are defined as in Formulas \eqref{eq:multi:Q_guess}, \eqref{eq:multi:F_guess}, then $F = F^*$ minimizes the expression
\[\frac{Q^*\left(v\right)}{dv} v\left(1-F^{n}\left(v\right)\right)+Q^*\left(v\right)\left(1-nF^{n-1}\left(v\right)+\left(n-1\right)F^{n}\left(v\right)\right) - \lambda\left(1-F\left(v\right)\right),\]
for any $v\in [\underv, \overv)$ and $F$ that its image is contained in $[0,1]$.
\end{claim}

\begin{proof}
For any $v \in [\underv, \overv)$ we define
\[h_v(z) = \frac{dQ^*\left(v\right)}{dv} v\left(1-z^n\right)+Q^*\left(v\right)\left(1-n z^{n-1} + \left(n-1\right) z^n \right) - \lambda\left(1-z\right).\]
We get
\[\frac{dh_v(z)}{dz} = -n\frac{dQ^*\left(v\right)}{dv} v z^{n-1}-n\left(n - 1\right) Q^*\left(v\right)\left(1 - z\right)z^{n-2} + \lambda.\]
After substituting for $Q^*(v)= \left(1-\nicefrac{a}{\overv}\right)^{n-1}\left(\frac{v}{v-a}\right)^{n-1}\left(\frac{n-1-\frac{1}{n-1}+\log\nicefrac{v}{\underv}}{n-1-\frac{1}{n-1}+\log\nicefrac{\overv}{\underv}}\right)$, and $z=1-\frac{a}{v}$, we get
\begin{align*}
    \frac{dh_v(z)}{dz} =& n \left(n-1\right)\left(1-\nicefrac{a}{\overv}\right)^{n-1}\left(\frac{v}{v-a}\right)
    \frac{ \left(n-1\right)a \log\nicefrac{v}{\underv}+\underv-v}{v \left(\left(n-2\right)n-\left(n-1\right)\log\nicefrac{\underv}{\overv}\right)} \\
    &- n\left(n - 1\right) 
    \left(1-\nicefrac{a}{\overv}\right)^{n-1}\left(\frac{v}{v-a}\right)\left(\frac{n-1-\frac{1}{n-1}+\log\nicefrac{v}{\underv}}{n-1-\frac{1}{n-1}+\log\nicefrac{\overv}{\underv}}\right) \cdot \frac{a}{v}
    + \lambda
\end{align*}
We dived by $ n \left(n-1\right)\left(1-\nicefrac{a}{\overv}\right)^{n-1}$, and write $\lambda_2 = \frac{\lambda}{ n \left(n-1\right)\left(1-\nicefrac{a}{\overv}\right)^{n-1}}$, and we get
\begin{align*}
    & \left(\frac{1}{v-a}\right)
    \frac{ a \log\nicefrac{v}{\underv}+\frac{\underv}{n-1} - \frac{v}{n-1}}{ n-1-\frac{1}{n-1}+\log\nicefrac{\overv}{\underv} }
    - \left(\frac{1}{v-a}\right)\left(\frac{n-1-\frac{1}{n-1}+\log\nicefrac{v}{\underv}}{n-1-\frac{1}{n-1}+\log\nicefrac{\overv}{\underv}}\right) \cdot a
    + \lambda_2 \\
    =& -\frac{1}{n-1} \cdot
    \frac{1}{ n-1-\frac{1}{n-1}+\log\nicefrac{\overv}{\underv} }
    + \lambda_2.
\end{align*}
Therefore, when $\lambda= \frac{n\left(1-\nicefrac{a}{\overv}\right)^{n-1}}{n-1-\frac{1}{n-1}+\log\nicefrac{\overv}{\underv}}$, $Q(v)= \left(1-\nicefrac{a}{\overv}\right)^{n-1}\left(\frac{v}{v-a}\right)^{n-1}\left(\frac{n-1-\frac{1}{n-1}+\log\nicefrac{v}{\underv}}{n-1-\frac{1}{n-1}+\log\nicefrac{\overv}{\underv}}\right)$, and $z=1-\frac{a}{v}$ we get
\[\frac{dh_v(z)}{dz} = 0.\]

To prove that it is a local minimum, we look at the second derivative
\begin{align*}
    \frac{d^2 h_v(z)}{dz^2} = n\left(n-1\right) \left( \left( \left( n - 1 \right) Q^*\left(v\right) - v \frac{dQ^*\left(v\right)}{dv} \right) z - \left( n - 2 \right) Q^*\left(v\right) \right) z^{n-3}.
\end{align*}
Since we only care about the sign, we only need to analyze
\begin{equation}
    \left( \left( n - 1 \right) Q^*\left(v\right) - v \frac{dQ^*\left(v\right)}{dv} \right) z - \left( n - 2 \right) Q^*\left(v\right). \label{eq:multi:second_derivative_term}
\end{equation}
Substituting  $Q^*(v) = \left(1-\nicefrac{a}{\overv}\right)^{n-1}\left(\frac{v}{v-a}\right)^{n-1}\left(\frac{n-1-\frac{1}{n-1}+\log\nicefrac{v}{\underv}}{n-1-\frac{1}{n-1}+\log\nicefrac{\overv}{\underv}}\right)$, and $z=1-\nicefrac{a}{v}$,  we get a positive value for every $v>\underv$ (See Appendix~\ref{appndx:clm:integrand_min_underv_to_overv:details} for details).

This proves that we found a local minimum. Since $\frac{d^2 h_v(z)}{dz^2}$ has only a single root in $(0,\infty)$, this is the single local minimum in this interval.
When $z=0$, we get that \eqref{eq:multi:second_derivative_term} is negative, since $Q^*(v) > 0$. That proves that the second derivative's single positive root is in $\left(0, 1 - \nicefrac{a}{v}\right)$. Therefore, $h_v$ is rising on $\left(1-\nicefrac{a}{v}, 1\right)$, which means that $h_v\left(1-\nicefrac{a}{v}\right) < h_v(1)$.

We now only need to prove that $h_v\left(1-\nicefrac{a}{v}\right) \le h_v(0)$.
To shorten the analysis ,we define
\begin{align*}
    \tilde{h}_{v}\left(z\right) =& \left(\frac{n-1-\frac{1}{n-1}+\log\nicefrac{\overv}{\underv}}{\left(1-\nicefrac{a}{\overv}\right)^{n-1}}\right)\cdot h_{v}\left(z\right) \\
    =& \left(\frac{v}{v-a}\right)^{n-1}\left(n-1-\frac{1}{n-1}+\log\nicefrac{v}{\underv}\right)\left(1-nz^{n-1}+\left(n-1\right)z^{n}\right) \\
    & -\left(\frac{v}{v-a}\right)^{n-1}\left(\frac{\left(n-1\right)a\log\nicefrac{v}{\underv}+\underv-v}{v-a}\right)\left(1-z^{n}\right)-n\left(1-z\right).
\end{align*}
It is enough to show that $\tilde{h}_v(0) - \tilde{h}_v\left(1-\nicefrac{a}{v}\right) \ge 0$, and indeed (see Appendix~\ref{appndx:clm:integrand_min_underv_to_overv:details})
\begin{align*}
    \tilde{h}_v(0) - \tilde{h}_v\left(1-\nicefrac{a}{v}\right) = \log\nicefrac{v}{\underv}-\frac{1-\nicefrac{\underv}{v}}{n-1}
    \ge 0.
\end{align*}
\end{proof}

The next claim handles the case when there is no solution $\underv$ to \eqref{eq:multi:underv_equation} in $[0, \overv]$.

\begin{claim}
\label{clm:multi:proof_when_underv_ge_overv}
If there is no solution for \eqref{eq:multi:underv_equation} in $[0, \overv]$, then $(Q^*, F^*)$ solve Problem~\eqref{eq:problem_reduced}, where
\begin{align*}
     F^*(v) &= \begin{cases}
    0 & v \in \left[0, v_0 \right), \\
    1-\frac{1}{(n-1)^2} & v\in [v_0, \overv), \\
    1 & v \ge \overv,
    \end{cases} \\
    Q^*(r) &= 1 \quad r \ge 0,
\end{align*}
where $v_0 = \mu-\frac{\overv-\mu}{\left(n-1\right)^{2}-1}$.
Furthermore, $F^*$ is nature's unique best response to the seller's $Q^*$.
\end{claim}
\begin{proof}
Since there is no solution for \eqref{eq:multi:underv_equation} in $[0, \overv]$, we know that
\[\left(n-\frac{1}{n-1}\right)\overv > \left(n-1\right)^{2}\mu,\]
since the term $\left(n-\frac{1}{n-1}+\log\nicefrac{\overv}{\underv}\right)\underv$ is increasing in $[0, \overv]$.

When the seller uses no reserve price, then $F^*$ is nature's best response according to Proposition~1 in \cite{Suzdaltsev20a}.

$F^*$ is unique since $z=0, 1-\frac{1}{(n-1)^2}$ are the only minimum points of the integrand $h_v(z)$ of Lagrangian, and therefore $F^*$ is the only CDF in $\calF$ that minimizes it point-wise. For clarity
\[h_v(z) = 1-nF^{n-1}\left(v\right)+\left(n-1\right)F^{n}\left(v\right) - \lambda(1-F(v)),\]
where $\lambda$ is given by equations \eqref{eq:multi:underv_equation}-\eqref{eq:multi:lambda_equation}

Now we need to show that if we fix $F^*$, then any reserve price higher than zero is no better than no reserve price.
To simplify our analyses, we consider the case where the highest bidder does buy the item when the reserve price is equal to her valuation.

Since the lowest valuation any bidder can have is $v_0$, then any reserve price equal or lower than $v_0$ will have no effect, so we only need to check reserve prices higher than $v_0$.
Since the nature is binary, any reserve price between $v_0$ and $\overv$ will not change the allocation of the item, so we can only benefit from increasing the reserve price. We only need to show that a reserve price of $\overv$ is no better than no reserve price.
The expected revenue for no reserve price is
\begin{align*}
    \Rev_{F^*}(0, F^*) &= v_{0}+\left(\overv-v_{0}\right)\left(1-n\left(1-\frac{1}{\left(n-1\right)^{2}}\right)^{n-1}+\left(n-1\right)\left(1-\frac{1}{\left(n-1\right)^{2}}\right)^{n}\right) \\
    &= \overv\left(1-\left(1-\nicefrac{v_{0}}{\overv}\right)\left(1+\frac{1}{n-1}\right)\left(1-\frac{1}{\left(n-1\right)^{2}}\right)^{n-1}\right).
\end{align*}
The expected revenue of reserve price $\overv$ is
\begin{align*}
    \Rev_{F^*}(\overv)  = \overv\left(1-\left(1-\frac{1}{\left(n-1\right)^{2}}\right)^{n}\right).
\end{align*}
Therefore
\begin{align*}
     & \Rev_{F^*}(0) \ge \Rev_{F^*}(\overv) \\
     \iff& \left(1-\nicefrac{v_{0}}{\overv}\right)\left(1+\frac{1}{n-1}\right)\le1-\frac{1}{\left(n-1\right)^{2}} \\
     \iff & \left(n-1\right)^{2}\mu\ge\overv\left(n-\frac{1}{n-1}\right).
\end{align*}
This is true from the assumption that Eq~\eqref{eq:multi:underv_equation} has no solutions in $[0, \overv]$.
\end{proof}

\begin{proof} [Proof of Theorem \ref{thm:multi:main:equilibrium}]
If there is a solution $\underv$ to Eq~\eqref{eq:multi:underv_equation} in $[0, \overv]$, then from Claim~\ref{clm:multi:candidate:equal_rev} we know that $\Rev_{F^*}(r_1) = \Rev_{F^*}(r_2)$, for any $r_1, r_2$ in the support of $Q^*$, which includes $\underv$.
Since $F^*$ has no mass on $(v_0, \underv)$, then any $r\in[v_0, \underv)$ cannot yield higher revenue than a reserve price of $\underv$. Likewise, since there is no mass over $[0, v_0)$, any reserve price lower than $v_0$ would have the same revenue as a reserve price of zero.
Therefore, we can conclude that $Q^*$ maximizes the expected revenue if nature uses $F^*$.
From Claims \ref{clm:multi:candidate:integrand_min}-\ref{clm:multi:candidate:integrand_min_underv_to_overv}, we get that distribution $F^*$ is nature's best response to the seller's $Q^*$, which completes the proof for the case that $\underv \in [0, \overv] $ is a solution to Eq~\eqref{eq:multi:underv_equation}.

The case where there is no solution $\underv$ to Eq~\eqref{eq:multi:underv_equation} in $[0, \overv]$ is handled by Claim~\ref{clm:multi:proof_when_underv_ge_overv}. This completes the proof.
\end{proof}

\subsection{Not an Optimal Mechanism} \label{subsec:multi:not_optimal}
Assuming an equilibrium exists, the second price auction with randomized reserve price is not the optimal mechanism when $n\ge3$, this is the case because nature does not use an equal-revenue distribution.
We now provide an example.

\begin{example} \label{ex:multi:not_optimal}
For $n=3, \overv = 1, \mu = \nicefrac{2}{3}$, Theorem~\ref{thm:multi:main:equilibrium} says that when restricted to second price auctions with randomized reserve, it is optimal for the seller to always choose a reserve price of zero, and the adversary's optimal distribution is
\begin{align*}
    F^*(v) &= \begin{cases}
    0 & v \in \left[0, \nicefrac{5}{9} \right), \\
    \nicefrac{3}{4} & v\in \left[\nicefrac{5}{9}, \overv\right), \\
    1 & v \ge \overv.
    \end{cases}
\end{align*}
The expected revenue in this case is $0.625$, see Fig~\ref{fig:multi:high_mu}.

However, the seller can do better.
Using \emph{ironing} \cite{Myerson81} (see also \cite{Hartline21}), we can find an optimal mechanism for the seller assuming the nature uses $F^*$.
This mechanism randomly picks a winner, uniformly among all the highest bidders, and allocates the item to the winner. The payment for the item is as follows
\begin{itemize}
    \item If all bidders reported a value of $\nicefrac{5}{9}$ then the payment is $\nicefrac{5}{9}$.
    \item If at least two bidders reported a value of $1$, then the payment is $1$.
    \item If only one bidder reported a value of $1$, then the payment is $\frac{23}{27}$.
\end{itemize}
Note: the mechanism handles false reports by ignoring bidders with bidding of less than $\nicefrac{5}{9}$, and treat any bidding in $\left[\nicefrac{5}{9}, 1\right)$ the same as $\nicefrac{5}{9}$. Bidding a value higher than 1 is treated the same as as bidding a value of 1.

In this mechanism, being truthful is a dominant strategy. The only case where it is different then second price, is when one bidder (assume bidder 1) has a valuation of $1$, and the other two have valuations of $\nicefrac{5}{9}$. The utility for bidder 1 of truthful bidding in this case is $1-\frac{23}{27} = \frac{4}{27}$.
By reporting less than 1, the bidder would win with probability of $\nicefrac{1}{3}$, and would pay $\nicefrac{5}{9}$. This results in an expected utility of $\nicefrac{1}{3} \cdot \left( 1 - \nicefrac{5}{9} \right) = \frac{4}{27}$ as well. Therefore, the bidders cannot gain anything by reporting falsely, and the mechanism is truthful.

The expected revenue in this case is $\frac{5}{9} \cdot \left( \frac{3}{4} \right)^3 + \frac{23}{27} \cdot 3 \cdot \frac{1}{4} \cdot \left( \frac{3}{4} \right)^2 + 1 \cdot \left( \frac{1}{4^3} + 3 \cdot \frac{1}{4^2} \cdot \frac{3}{4}\right) = \frac{3}{4}$.
Which is higher than the revenue we get from the second price auction.
\end{example}

The mechanism presented in Example~\ref{ex:multi:not_optimal} is optimal for $F^*$, but not the other way around, and it would be the case for any other mechanism that is equal to it for $F^*$. That is easy to see since nature can choose the trivial point mass distribution only on $\mu = \nicefrac{2}{3}$, which would guarantee a revenue no higher than $\nicefrac{2}{3}$.

According to Claim~\ref{clm:multi:proof_when_underv_ge_overv}, $F^*$ is nature's unique best response to $Q^*$, in the same setting of the example.
That means that every other equilibrium of Problem~\eqref{eq:problem_reduced} is of the form $(Q,F^*)$ for some $Q\in \calQ$, and it achieve the same expected revenue.
Therefore, the example here prove that no second price auction with randomized reserve can be an equilibrium for the general Problem~\eqref{eq:problem}.

\section{Discussion}

In this work we identify the second price auction with random reserve as the robustly revenue-optimal mechanism for two i.i.d.~bidders, in settings where only the mean and support of the value distribution are known. We develop an expression for the robustly optimal distribution of the random reserve, and generalize this finding to $n\ge 3$ bidders. 
This distribution constitutes part of an equilibrium of a zero-sum game between the seller, who controls the random reserve price, and an adversary who controls the value distribution. Our main technique is to start from a discretization of the setting, uncover recursive relations between the discretized probabilities, and use these relations to derive a guess of the equilibrium distributions. We also give an explanation of when to expect the equilibrium to induce indifference among one or both of the zero-sum game players. These technical contributions can possibly be of more general applicability in robustly optimal mechanism design. 

Our work raises several open questions for future research. 

\paragraph{\bf Beyond two bidders.} 
First, while for $n\ge 3$ bidders we establish the max-min optimal distribution of the reserve price when the seller is confined to second price auctions, the following is still open: 

\begin{open}
\label{opn:beyond-two}
What is the max-min optimal auction given parametric distributional knowledge for $n\ge 3$ bidders? 
\end{open}

In Section~\ref{sec:multi_bidders} we show that the answer to Open Question~\ref{opn:beyond-two} is generally \emph{not} a second price auction with randomized reserve. This raises the next question:

\begin{open}
Are there natural conditions under which the second price auction with random reserve is max-min optimal for $n\ge 3$ bidders, such as \emph{regularity} of the value distribution~\citep{Myerson81}?
\end{open}

The regularity assumption may require new techniques for the analysis, but it seems like an intuitively logical condition to try given that even for known distributions, the optimality of the second price auction with (deterministic) reserve holds only for a regular value distribution (in fact, Example~\ref{ex:multi:not_optimal} demonstrates its sub-optimality for irregular distributions). An interesting, high-level research direction is to understand whether there is a formal connection between the optimal and robustly optimal auctions. E.g., is it just a ``coincidence'' that for two bidders, the optimal and robustly optimal auction formats are quite similar (up to randomization of the reserve price)? 


\paragraph{\bf Beyond i.i.d.~bidders.} 

Our focus in this work is on the case of i.i.d.~bidders. Indeed, in revenue maximization with known value distributions, the case of bidders with values that are independent and even i.i.d.~is very central: independence enables the fundamental theory of Myerson, and precludes ``unrealistic'' results like the Cr\'emer-McLean auction~\cite{CremerM88}. Furthermore, i.i.d.~leads to the celebrated result that the second price auction with deterministic reserve is optimal. This makes the i.i.d.~case interesting for robust revenue maximization as well, but a natural next step is to go beyond this assumption.

\begin{open}
What is the max-min optimal auction given parametric distributional knowledge for independent bidders with value distributions that are not necessarily identical? 
\end{open}

As for the case of correlated bidders, this case has been studied in the robust revenue maximization literature. In particular, \citet{Che22} studies a similar setting to ours while dropping the i.i.d.~assumption. His results are as follows: If the seller chooses a second price auction with randomized reserve, then nature will strongly correlate the bidders in the sense that just one (randomly chosen) bidder will be a true contender for the item, and all other bidders will have the same value $\alpha$. In effect, nature is using the power to correlate in order to abuse the second price auction format, by eliminating the revenue-inducing competition among the bidders and leaving a single bidder ``in the game''. When nature has this power (due to dropping the i.i.d.~restriction, or even just the independence restriction), Che demonstrates that second price is not robustly optimal in general even for two bidders. (He goes on to show that under additional uncertainty about the bidders’ information, second price is robustly optimal among a subclass of mechanisms he terms “competitive”.)

An interesting direction for future work is to limit the correlation that nature can introduce among bidders, and study the optimal mechanism robust against limited correlation (rather than against extreme coordination of the bidders as above, which is arguably less realistic).

\bibliographystyle{ACM-Reference-Format}
\bibliography{abb,references}

\appendix

\section{Equilibrium in the Two Bidder Discrete Case} \label{appndx:two_bidders_discrete}
In this section we focus exclusively on the reduced zero-sum game, solving Problem~\eqref{eq:problem_reduced} in a discretized model.
In Section~\ref{subsec:two:discrete} we present two Claims (\ref{clm:two_bidders:disc:nature_induction}, \ref{clm:two_bidders:disc:seller_induction}) about a potential solution in this setting.
We will prove its existence and correctness here.

\vspace{0.05in}
\noindent{\bf The discretized model.} Let $k> 1$ be the discretization parameter. We effectively limit the support of both the seller's distribution over reserve prices and the adversary's distribution over values to $k+1$ possible options $\rvec_k=(r_{k,1},\dots, r_{k,k}, r_{k,k+1})$, where $r_{k,i} = \frac{i-1}{k}\cdot \overv$ for every $1\le i \le k+1$.
We denote by $\calF_k$ the set of all distribution functions over $\rvec_k$ with a mean value of $\mu$, and by $\calQ_k$ the set of all distribution functions over $\rvec_k$. To be precise, the functions in $\calF_k$ and $\calQ_k$ are step functions that only have jumps in points of $\rvec_k$.

\vspace{0.05in}
\noindent{\bf Our main result for the discretized model.} The following theorem guarantees the existence of a solution to Problem~\eqref{eq:problem_reduced} in the discrete case, and provides properties that the solution satisfies. 
Together with the mean value constraint, these properties uniquely describe the solution, and will allow us to solve the general case in subsequent sections.

\begin{theorem} 
\label{thm:two_bidders:disc}
For every $k > 1$, there exists an equilibrium $\left(Q^*_k, F^*_k \right) \in \calQ_k \times \calF_k$ of the reduced zero-sum game in the discretized model, where the seller and the adversary are limited to $\rvec_k$. The equilibrium satisfies the following:
\begin{itemize}
    \item $F^*_k(r_{k,i+1}) = \sqrt{\frac{\overv}{k} \cdot \frac{2 }{r_{k,i} +\nicefrac{\overv}{k}}(F^*_k(r_{k,i}) - F^*_k(r_{k,i})^2) + F^*_k(r_{k,i})^2}$ for every $1 \le i \le k-1$;
    \item $Q^*_k(r_{k,i-1}) = Q^*_k(r_{k,i}) - \frac{Q^*_k(0)(1-F^*_k(0)) - Q^*_k(r_{k,i})(1-F^*_k(r_{k,i}))}{r_{k,i} F^*_k(r_{k,i})} \cdot\nicefrac{\overv}{k}$ for every $2 \le i \le k$;
    \item $Q^*_k(r_{k,k}) = 1$.
\end{itemize}
\end{theorem}

\subsection{Existence}
We will prove that there exist non decreasing vectors $\qvec^*, \xvec^* \in [0,1]^{k}$ solving Claims~\ref{clm:two_bidders:disc:nature_induction}-\ref{clm:two_bidders:disc:seller_induction}.
\begin{claim} \label{appndx:clm:two_bidders:disc:x_star_existence}
There exists a unique vector $\xvec^* = (x_1^* ,\dots, x_k^*)$ that satisfies the following:
\begin{itemize}
    \item $\frac{\partial}{\partial q_i}\calL(\xvec^*, \qvec) = 0$ for every $1\le i \le k-1$ and every $\qvec \in \mathbb{R}^{k-1}$.
    \item $\nicefrac{\overv}{k}\cdot\sum_{i=1}^k(1 - x^*_i) = \mu$.
    \item $0 < x_1^* < x_2^* \dots <x_k^* < 1$.
\end{itemize}
\end{claim}
\begin{proof}
We define a new function $\chi_i:[0,1] \to [0,1]$ by
\begin{align*}
    \chi_1(x) &= x ;\\
    \chi_{i+1}(x) &= \sqrt{\frac{\overv}{k} \cdot \frac{2 }{r_{k,i} +\nicefrac{\overv}{k}}(\chi_i(x) - \chi_i^2(x)) + \chi_i^2(x)} \quad 1\le i \le k-1.
\end{align*}
It is not hard to see that for every $1 \le i \le k-1$ and $x\in[0,1]$
\[x \le \sqrt{\frac{\overv}{k} \cdot \frac{2 }{r_{k,i} +\nicefrac{\overv}{k}}(x - x^2) + x^2} \le 1,\]
and if $x_i\in(0,1)$ then the inequalities are strict.
From that we can see that $0 \le \chi_1(x) \le \dots \le \chi_k(x) \le 1$ for every $x \in [0,1]$, and equality can only happen at $x=0,1$.
It is also not hard to see that $\chi_i(x)$ is continuous and increasing for every $1 \le i \le k$.

We now define
\[\bar{x}_k(x) = 1 - \frac{k \cdot \mu}{\overv} + \sum_{i=1}^{k-1}(1 - \chi_i(x)).\]
Since $\chi_1(x), \dots, \chi_{k-1}(x)$ are all continuous and increasing, then $\bar{x}_k(x)$ is continuous and decreasing.
We can also get the values of $\chi_k(x)$ and $\bar{x}_k(x)$ at $x=0,1$ with ease:
\begin{align*}
    \chi_k(0) &= 0, &\chi_k(1) &= 1; \\
    \bar{x}_k(0) &= k \left(1 - \nicefrac{\mu}{\overv}  \right) > 0, &\bar{x}_k(1) &= 1 - k\cdot \nicefrac{\mu}{\overv} < 1.
\end{align*}
We now get that there is a unique value $x^* \in (0,1)$ such that $\chi_k(x^*) =\bar{x}_k(x^*) $.
We denote $x_i^* = \chi_i(x^*)$ for every $1 \le i \le k$, and $\xvec^* = (x_1^* ,\dots, x_k^*)$.
From Claim \ref{clm:two_bidders:disc:nature_induction}  we get that for every $1 \le i \le k-1$
\[\frac{\partial}{\partial q_i}\calL(\xvec^*, \qvec) = 0.\]
Since $x^*_k = \bar{x}_k(x^*_1)$, we also know that $\nicefrac{\overv}{k}\cdot\sum_{i=1}^k(1 - x^*_i) = \mu$, which concludes the proof.
\end{proof}

\begin{claim} \label{appndx:clm:two_bidders:disc:q_star_existence}
There exists a unique vector $\qvec^* = (q_1^* ,\dots, q_{k}^*)$ that satisfies the following:
\begin{itemize}
    \item $\frac{\partial}{\partial x_i}\calL(\xvec^*, \qvec^*) = 0$ for every $1\le i \le k-1$;
    \item $0 < q_1^* \dots <q_{k-1}^* < 1$.
\end{itemize}
\end{claim}

\begin{proof}
Using Claim \ref{clm:two_bidders:disc:seller_induction}, and the knowledge that $q_k = 1$ would be optimal, we can calculate the rest of $q^*$ as a function of $\lambda$, where
\begin{align*}
    q^*_{i-1}(\lambda) &= q_{i}(\lambda)-\frac{\nicefrac{\lambda}{2}-q^*_{i}(\lambda)\left(1-x^*_{i}\right)}{r_{k,i}x^*_{i}}\cdot\nicefrac{\overv}{k},\\
    q^*_k(\lambda) &= 1.
\end{align*}
It is clear that $q^*_i(\lambda)$ is decreasing for every $1 \le i \le k$.
It is also easy to see that for every $2 \le i \le k$
\[q^*_{i-1}(\lambda) < q^*_i(\lambda) \iff  q^*_i(\lambda) < \frac{\nicefrac{\lambda}{2}}{1-x^*_i}.\]
We can see that for $\lambda = 2$
\begin{align*}
    q^*_{k-1}(2) = 1-\frac{x^*_{k}}{r_{k,k}x^*_{k}}\cdot\nicefrac{\overv}{k} = 1 - \frac{\nicefrac{\overv}{k}}{1 - \nicefrac{\overv}{k}} < 1 < \frac{\nicefrac{\lambda}{2}}{1-x^*_{k-1}}.
\end{align*}
With recursion we can easily see that $q_1^*(2) < 1$ as well.
For $\lambda = 0$ we see that
\begin{align*}
    q^*_{k-1}(0) = 1 + \frac{1-x^*_{i}}{r_{k,i}x^*_{i}}\cdot\nicefrac{\overv}{k} > 1 > \frac{\nicefrac{\lambda}{2}}{1-x^*_{k-1}},
\end{align*}
and similarly we get that $q_1^*(0) > 1$.

From Claim~\ref{clm:two_bidders:disc:seller_induction}, we also know that
\[q_1^* = \frac{\nicefrac{\lambda}{2}}{1-x^*_1}.\]
Since this expression is increasing in $\lambda$, and lower than 1 for $\lambda=0$, and higher than 1 for $\lambda =2$, we get that there is a unique solution $\lambda \in (0,2)$ such that
\[q_1^*(\lambda) = \frac{\nicefrac{\lambda}{2}}{1-x^*_1}\].

In order to prove that $q_1^* \dots <q_{k-1}^*$, we assume that there exists $i$ such that $q_{i-1}^* > q_i^*$.
Since $x_i^* > x_{i-1}^*$ we get
\[q_{i-1}^* > q_i^* > \frac{\nicefrac{\lambda}{2}}{1-x_i^*} > \frac{\nicefrac{\lambda}{2}}{1-x_{i-1}^*} \implies q_{i-2}^* > q_{i-1}^* \implies q_1^* > \dots > \frac{\nicefrac{\lambda}{2}}{1-x_i^*} > \frac{\nicefrac{\lambda}{2}}{1-x_1^*},\]
which is a contradiction, since $\frac{\nicefrac{\lambda}{2}}{1-x_1^*} = q_1^*$.
\end{proof}

\subsection{Proof that the Saddle Point is an Equilibrium} 

Claims \ref{appndx:clm:two_bidders:disc:x_star_existence} and \ref{appndx:clm:two_bidders:disc:q_star_existence} only prove that the gradient is zero at $(\xvec^*, q^*)$, but they do not guarantee that it is an equilibrium. The next claim will prove that $\xvec^*$ is the adversary's best response to $\qvec^*$.

\begin{claim} \label{appndx:clm:two_bidders:disc:x_star_is_minimum}
$\xvec = \xvec^*$ from Claim \ref{appndx:clm:two_bidders:disc:x_star_existence} minimizes the Lagrangian $\calL(\xvec,\qvec^*)$, where $\qvec^*$ is given by Claim \ref{appndx:clm:two_bidders:disc:q_star_existence}.
\end{claim}
\begin{proof}
In order to show that $\xvec^*$ is a minimum, we look on the second derivative, and prove that it is always positive at $q^*$, which implies that $x^*$ is the minimum.
\begin{align*}
     \frac{\partial^2}{\partial x_i^2} \calL_k(\qvec^*, \xvec) =  -2 (q^*_i-q^*_{i-1}) r_{k,i} + 2q^*_i \cdot \nicefrac{\overv}{k}
\end{align*}
It is enough to show that
\begin{align*}
    (q^*_i-q^*_{i-1}) r_{k,i} \le q^*_i \cdot \nicefrac{\overv}{k}.
\end{align*}
For $i=1$ this is true since $r_{k,1} = 0$ (and $q^*_0 = 0$).
For $2 \le i \le k$ this is equivalent to 
\begin{align*}
    \frac{q^*_i - q^*_{i-1}}{\nicefrac{\overv}{k}}\cdot r_{k,i}\le q^*_i.
\end{align*}
From Claim \ref{clm:two_bidders:disc:seller_induction} we get
\begin{align*}
     \frac{q^*_i - q^*_{i-1}}{\nicefrac{\overv}{k}}\cdot r_{k,i} &= \frac{\nicefrac{\lambda}{2} - q^*_i(1-x_i)}{x_i} \\
     &= q^*_i - \frac{\nicefrac{q^*_i - \lambda}{2}}{x_i} \\
     &= q^*_i - \frac{q^*_i - q^*_1(1-x^*_1)}{x_i} \\
     &\le q^*_i.
\end{align*}
This completes the proof.
\end{proof}

Now we are ready to complete the proof of Theorem \ref{thm:two_bidders:disc}.

\begin{proof}[Proof of Theorem \ref{thm:two_bidders:disc}]
From Claims \ref{clm:two_bidders:disc:nature_induction}, \ref{clm:two_bidders:disc:seller_induction}, \ref{appndx:clm:two_bidders:disc:x_star_existence}, \ref{appndx:clm:two_bidders:disc:q_star_existence}
we get the unique vectors $\xvec^*$ and $\qvec^*$ such that their corresponding distributions $F^*_k \in \calF_k$ and $Q^*_k \in \calQ_k$, respectively, satisfy the same formulas as in the theorem statement.

From Claim \ref{appndx:clm:two_bidders:disc:x_star_is_minimum} we get that $F^*_k$ is the adversary's best response to $Q^*_k$.
The only thing left is to prove that $Q^*$ is the seller's best response to $F^*_k$.

From Claim \ref{appndx:clm:two_bidders:disc:x_star_existence} we get that the derivative of the expected revenue in regard to any of the seller's variables is always zero at $F^*_k$, which means that the revenue is constant, as long as there is no mass on $\overv$. Since any mass on $\overv$ will only lower the expected revenue, we get that $Q^*_k$ is a best response to $F^*_k$.
\end{proof}

\section{A Discrete-to-General Formalism} 
\label{appndx:sec:toolbox}


In the proofs of Theorems \ref{thm:two_bidders:main:equilibrium},\ref{thm:multi:main:equilibrium} we have first presented recursive relations for a saddle point in the discrete case (For example Claim~\ref{clm:two_bidders:disc:nature_induction}).
To get ODEs that the non-discrete distributions satisfy, we have defined a function that represent the recursive relations, and derived it at zero, to get the ODEs.
Since we used it to guess an equilibrium, we did not need to provide a formal proof, but this technique is proven formally here (Lemmas \ref{lemma:disc_to_cont_uniform_exists}, \ref{lemma:disc_to_cont_using_derivative}).
There Lemmas are able to provide close form expression when going from a discrete solution to a general one, without needing a close form solution to the discrete case, which we believe may be of independent interest.

In Section~\ref{sub:toolbox-defs} we define formally what we mean by a sequence of recursive solutions to the discretized settings indexed by $k$.
Section~\ref{sub:subsequence-exist} establishes conditions under which such a sequence is guaranteed to have a converging subsequence with a continuous limit.
Section~\ref{sub:limit-guarantees} provides a further guarantee on the limit (under slightly stronger conditions), namely that it must satisfy a certain ordinary differential equation \emph{(ODE)} appearing in Lemmas~\ref{lemma:disc_to_cont_using_derivative}-\ref{lemma:disc_to_cont_using_derivative:open}.

\subsection{Definitions}
\label{sub:toolbox-defs}

\begin{definition}[Recursive discrete+] \label{def:psi_recursive_discrete}
Let $\mathcal{I} \subseteq [0,\overv]$ be an interval and let $\mathcal{C},\mathcal{Z}$ be compact intervals.
Consider a sequence $\{f_k\}_{k\in\N}$ of functions from a domain $\mathcal{X}\subseteq [0,\overv]$ to a range $\mathcal{Y}$, where $f_k$ is a step function with jumps at $\rvec_k\cap \mathcal{X}$ for every $k$. 
Let $I^+_k = \left\{i \mid 1\le i \le k \text{ and } r_{k,i},r_{k,i+1} \in \mathcal{X} \right\}$ be the set of all indices such that $f_k(\cdot)$ is defined on both $r_{k,i}$ and $r_{k,i+1}$. 
The sequence of functions $\{f_k\}$ is called ``\textbf{recursive discrete+}'' if for every $k$ and $i\in I_k^+$ we can express $f_k\left(r_{k,i+1}\right)$ as follows:
\begin{equation}
    f_k\left(r_{k,i+1}\right) = \psi \left(\frac{\overv}{k},r_{k,i},f_{k}\left(r_{k,i}\right),g_{k}\left(r_{k,i}\right),a_k\right), \label{eq:def:psi_recursive_discrete+}
\end{equation}
where $\{a_k\} \subset \mathcal{C}$ is some sequence, $\{g_k\}$ is a sequence of functions from $\mathcal{X}$ to $\mathcal{Z}$ that uniformly converge to a continuous function, 
and $\psi(\delta, x, y, z, a): \mathcal{I} \times \mathcal{X} \times \mathcal{Y} \times \mathcal{Z} \times \mathcal{C} \to \mathcal{Y}$ is a continuous function such that 
for every $x,y,z,a$,
\begin{equation}
    \psi(0,x,y,z,a) = y. \label{eq:def:psi_recursive_discrete_delta_0}
\end{equation}
\end{definition}

Similarly to the definition of ``recursive discrete+'' we say that a sequence $\{f_k\}_{k\in\N}$ is ``\textbf{recursive discrete-}'', if it satisfies the same conditions, but instead of condition \eqref{eq:def:psi_recursive_discrete+} it satisfies
\begin{equation}
    f_k\left(r_{k,i-1}\right) = \psi \left(\frac{\overv}{k},r_{k,i},f_{k}\left(r_{k,i}\right),g_{k}\left(r_{k,i}\right),a_k\right).
    \label{eq:def:psi_recursive_discrete-}
\end{equation}
for every $k$ and $i\in I_k^-$, where $I^-_k = \left\{i \mid 1\le i \le k \text{ and } r_{k,i},r_{k,i-1} \in \mathcal{X} \right\}$ the set of all indices where $f_k$ is defined on both $r_{k,i}$ and $r_{k,i-1}$.
If a sequence $\{f_k\}_{k\in\N}$ is either recursive discrete+ or recursive discrete- we simply say it is  "\textbf{recursive discrete}".

The last two arguments of $\psi$ are optional, and $\psi$ can use only the first three arguments; This is the case for $\{F^*_k\}$, where the value of $F^*_k(r_{k,i+1})$ is only dependent on $\delta = \frac{\overv}{k}, x = r_{k,i}, y=F^*_k(r_{k,i})$.
Those arguments are used for $\{Q^*_k\}$, where the value at $r_{k,i-1}$ is dependent on $\frac{\overv}{k}, r_{k,i}, Q^*_k(r_{k,i})$ and on $F^*_k(r_{k,i})$ and $Q^*_k(0)(1-F^*_k(0))$ as well.

\subsection{Existence of a Uniformly-Convergent Subsequence with a Continuous Limit}
\label{sub:subsequence-exist}
In this section we prove that a recursive discrete sequence, as defined above, is guaranteed to contain uniformly converges subsequences, that converge into a continuous function, when the corresponding function  $\psi$ is Lipschitz.

\begin{lemma} \label{lemma:disc_to_cont_uniform_exists}
Let $\mathcal{I}, \mathcal{X}, \mathcal{Y}, \mathcal{Z}, \mathcal{C}$ be intervals where $0 \in \mathcal{I}$, $\mathcal{I}, \mathcal{X} \subseteq [0,\overv]$, and $\mathcal{Y}, \mathcal{Z}, \mathcal{C}$ are compact.
Let $\left\{ f_{k}\right\}$ be a recursive discrete sequence for a suitable function $\psi: \mathcal{I} \times \dots \times \mathcal{C} \to \mathcal{Y}$, and corresponding sequences $\{g_k\}$ and $\{a_k\}$, where $\{g_k\}$ converges uniformly to a continuous function $g:\mathcal{X} \to \mathcal{Z}$.
Denote by $S$ the space of functions from $\mathcal{X}$ to $\mathcal{Y}$ endowed by the uniform metric $d_S(h_1,h_2) = \sup_{x\in \mathcal{X}} \abs*{h_1(x) - h_2(x)}$.
If $\psi\left(\delta,x,y,z,a\right)$ is Lipschitz in $\delta$, then the sequence $\left\{ f_{k}\right\}$ is relatively compact in $S$, and any limit is continuous.
\end{lemma}

\begin{proof}
We assume that $\{f_k\}$ is recursive discrete+, and the proof for recursive discrete- is similar.
We define $I_k \coloneqq I_k^+$ where $I_k^+$ is defined the same way as in Definition \ref{def:psi_recursive_discrete}.
$\psi$ is $c$-Lipschitz in $\delta$, which means that for every $k>1$ such that $\frac{\overv}{k} \in \mathcal{I}$ every and $i \in I_k$
\begin{align*}
    & \abs*{f_{k}\left(r_{k,i+1}\right) - f_{k}\left(r_{k,i}\right)} \\
    = & \abs*{\psi\left(\frac{\overv}{k},r_{k,i},f_{k}\left(r_{k,i}\right),g_{k}\left(r_{k,i}\right),a_{k}\right)-\psi\left(0,r_{k,i},f_{k}\left(r_{k,i}\right),g_{k}\left(r_{k,i}\right),a_{k}\right)} \\
    \le & \frac{\overv}{k}\cdot c.
\end{align*}
This is an upper bound for the height of the jumps in $f_k$. Since in every interval $[x_1, x_2]$ there is no more than $\left\lceil \frac{k}{\overv} \cdot \abs*{x_1 - x_2} \right\rceil \le \frac{k}{\overv} \cdot \abs*{x_1 - x_2} + 1$ jumps, we get that for every $x_1, x_2 \in [0, \overv]$
\[\abs*{f_{k}\left(x_1\right) - f_{k}\left(x_2\right)} \le c\abs*{x_1 - x_2} + \frac{\overv}{k}\cdot c.\]
If $\mathcal{X}$ is closed, and therefore compact, we use a generalization for non-continuous functions of the Arzelà–Ascoli theorem~\cite[see, e.g.,][]{WikiArzelaAscoli21}. We conclude that $\{f_k\}$ is relatively compact in $S$, and any limit is continuous.

If $\mathcal{X}$ is not compact, extend the definition of $f_k$ for each $k$ to $\overline{\mathcal{X}}$, the closure of $\mathcal{X}$. It is compact since $\overline{\mathcal{X}} \subseteq [0,\overv]$.
We define $\overline{f}_k$ to be identical to $f_k$ on $\mathcal{X}$.
If a point $x' \in \overline{\mathcal{X}} \setminus \mathcal{X}$, we simply define $\overline{f}_k(x') = \lim_{x \to x'}f_k(x)$, which exists since $f_k$ is a step function and hence has no more than $k+1$ points of discontinuity.
The inequality 
\[\abs*{\overline{f}_k\left(x_1\right) - \overline{f}_k\left(x_2\right)} \le c\abs*{x_1 - x_2} + \frac{\overv}{k}\cdot c\]
is still true for every $x_1, x_2 \in [0, \overv]$, which allows us to use the same generalization for non-continuous functions for Arzelà–Ascoli theorem.
Since the sequence $\{\overline{f}_k\}$ is relatively compact in $S$, so is the sequence $\{f_k\}$.
\end{proof}

\subsection{The Limit Must Satisfy an ODE}
\label{sub:limit-guarantees}

In the previous section, Lemma \ref{lemma:disc_to_cont_uniform_exists} gives the existence of a uniformly convergent subsequence with a continuous limit. Under slightly stronger assumptions, we now prove a property that any limit must satisfy. 
One way to use the resulting stronger lemma (Lemma~\ref{lemma:disc_to_cont_using_derivative}) is to combine the property it establishes with other known properties that any limit must satisfy, such that together this combination of properties is only satisfied by one function; this will prove convergence of the sequence. 
We also show an almost equivalent version (Lemma~\ref{lemma:disc_to_cont_using_derivative:open}) with less limitations.

\begin{lemma} \label{lemma:disc_to_cont_using_derivative}
Let $\mathcal{I}, \mathcal{X}, \mathcal{Y}, \mathcal{Z}, \mathcal{C}$ be compact intervals where $0 \in \mathcal{I}$, and $\mathcal{I}, \mathcal{X} \subseteq [0,\overv]$.
Let $\left\{ f_{k}\right\}$ be a recursive discrete sequence for a suitable function $\psi: \mathcal{I} \times \dots \times \mathcal{C} \to \mathcal{Y}$, and corresponding sequences $\{g_k\}$ and $\{a_k\}$ where $\{g_k\}$ converges uniformly to a continuous function $g:\mathcal{X} \to \mathcal{Z}$.
Denote by $S$ the space of functions from $\mathcal{X}$ to $\mathcal{Y}$ endowed by the uniform metric $d_S(h_1,h_2) = \sup_{x\in \mathcal{X}} \abs*{h_1(x) - h_2(x)}$.

If $\psi\left(\delta,x,y,z,a\right)$ is differentiable in $\delta$ and $\frac{\partial}{\partial \delta}\psi\left(\delta,x,y,z,a\right)$ is continuous, then the sequence $\left\{ f_{k}\right\}$ is relatively compact in $S$, and every limit $f$ is differentiable and satisfies
\[f'(r) = \begin{cases}
    \frac{\partial \psi\left(\delta,x,y,z,a\right)}{\partial\delta} \left(0,r,f(r),g(r),a\right) & \{f_k\}\text{ is recursive discrete+} \\
    -\frac{\partial \psi\left(\delta,x,y,z,a\right)}{\partial\delta} \left(0,r,f(r),g(r),a\right) & \{f_k\}\text{ is recursive discrete-}
\end{cases},\]
where $a$ is some limit of the subsequence $\{a_{k_j}\}$ and $f_{k_j}$ is the subsequence the converges into~$f$.
\end{lemma}

\begin{proof}
We prove this for the case of recursive discrete+. The case of recursive discrete- is very similar.
Since the derivative of $\psi$ in $\delta$ is continuous on a compact space, $\psi$ is Lipschitz in $\delta$, and from Lemma \ref{lemma:disc_to_cont_uniform_exists} we get that $\{f_k\}$ is relatively compact in $S$, and any limit is continuous.
For every $k$ we define the function $\tilde{f}:\mathcal{X} \to \mathcal{Y} $ to be equal to $f_k$ on $\rvec_k$, but $\tilde{f}_k$ is continuous, and linear on $[r_{k,i}, r_{k,i+1}]$ for every $i \in I_k$.
It is easy to see that for every subsequence $\left\{f_{k_j}\right\}$ that converges uniformly to a continuous function $f$, the subsequence $\left\{\tilde{f}_{k_j}\right\}$ also converges uniformly to $f$.
For every $k$  such that $\frac{\overv}{k} \in \mathcal{I}$, and every $i \in I_k$, and every $r\in [r_{k,i}, r_{k,i+1})$
\begin{align*}
    \tilde{f}_k'(r) &= \frac{f_{k}\left(r_{k,i+1}\right) - f_{k}\left(r_{k,i}\right)}{\frac{\overv}{k}} \\
    &=\frac{\psi\left(\frac{\overv}{k},r_{k,i},f_{k}\left(r_{k,i}\right),g_{k}\left(r_{k,i}\right),a_{k}\right)-\psi\left(0,r_{k,i},f_{k}\left(r_{k,i}\right),g_{k}\left(r_{k,i}\right),a_{k}\right)}{\frac{\overv}{k}} \\
    &= \frac{\partial \psi\left(\delta,x,y,z,a\right)}{\partial \delta} \left(\delta',r_{k,i},f_{k}\left(r_{k,i}\right),g_{k}\left(r_{k,i}\right),a_{k}\right),
\end{align*}
where $\delta' \in \left[0,\frac{\overv}{k}\right]$ and $\tilde{f}_k'(r)$ is defined to be the right derivative if $r=r_{k,i,k}$.
The derivative $\frac{\partial \psi\left(\delta,x,y,z,a\right)}{\partial \delta}$ is continuous on a compact space, and thus uniformly continuous.
Let $\{k_j\}$ be a sequence such that $\{f_{k_j}\}$ converges uniformly to a continuous function $f$ and $\{a_{k_j}\}$ converges to some $a\in \mathcal{C}$.
We get that for every $\varepsilon > 0$, for any large enough $j$, and $r\in \mathcal{X}$.
\[\abs*{\tilde{f}_{k_j}'(r) - \frac{\partial \psi\left(\delta,x,y,z,a\right)}{\partial \delta} \left(0,r,f\left(r\right),g\left(r\right),a\right)} < \varepsilon.\]
Therefore, the sequence $\left\{\tilde{f}_{k_j}' \right\}$ is uniformly convergent, and it converges to $f'$.%
\footnote{This is known when the functions of the sequence are differentiable on the entire domain. While this is not the case here, the functions are continuous and piecewise linear, which suffices. This is true since for every continuous and piecewise linear function $g$, the equality $g(x) =  g(x_0) + \int_{x_0}^x g'(t)dt$ holds, and since $\lim_{k \to \infty}\int_{\mathcal{X}'} \tilde{f}_k'(t)dt = \int_{\mathcal{X}'} \lim_{k \to \infty}\tilde{f}_k'(t)dt$ for every interval $\mathcal{X}' \subseteq \mathcal{X}$, we conclude that $\lim_{k \to \infty}\tilde{f}_k'(x) = f'(x)$.} 
We conclude that for every $r \in \mathcal{X}$
\[f'(r) = \lim_{j\to \infty} \tilde{f}_{k_j}'(r) = \frac{\partial \psi\left(\delta,x,y,z,a\right)}{\partial \delta} \left(0,r,f\left(r\right),g\left(r\right),a\right).\]
\end{proof}

Lemma \ref{lemma:disc_to_cont_uniform_exists} requires $\psi$ to be Lipschitz over all the domain, and Lemma \ref{lemma:disc_to_cont_using_derivative} only works when we limit the domain of the sequence $\{f_k\}$ to a compact interval.
By representing a general interval as a countable infinite union of compact intervals, and applying the lemmas to each of those intervals, we can get a similar result with less limitations, as follows:

\begin{lemma}\label{lemma:disc_to_cont_using_derivative:open}
Let $\mathcal{I}, \mathcal{X}, \mathcal{Y}, \mathcal{Z}, \mathcal{C}$ be any intervals where $0 \in \mathcal{I}$, $\mathcal{I}, \mathcal{X} \subseteq [0,\overv]$, and $\mathcal{C}$ is compact.
Let $\left\{ f_{k}\right\}$ be a recursive discrete sequence for a suitable function $\psi: \mathcal{I} \times \dots \times \mathcal{C} \to \mathcal{Y}$, and corresponding sequences $\{g_k\}$ and $\{a_k\}$ where $\{g_k\}$ converges uniformly to a continuous function $g:\mathcal{X} \to \mathcal{Z}$.

If for every compact interval $\tilde{\mathcal{X}} \subseteq \mathcal{X}$, there are compact intervals $\tilde{\mathcal{Y}} \subseteq \mathcal{Y}$ and $\tilde{\mathcal{Z}} \subseteq \mathcal{Z}$ such that for every large enough $k$ $f_k\left( \tilde{\mathcal{X}} \right) \subseteq \tilde{\mathcal{Y}}$ and $g_k\left( \tilde{\mathcal{X}} \right) \subseteq \tilde{\mathcal{Z}}$, then the following is true:
\begin{enumerate}
    \item If $\psi\left(\delta,x,y,z,a\right)$ is locally Lipschitz in $\delta$, then any subsequence of $\left\{ f_{k}\right\}$ has a locally uniform convergent subsequence, and any limit is continuous. \label{lemma:disc_to_cont_using_derivative:lipschitz:open}
    \item If $\psi\left(\delta,x,y,z,a\right)$ is differentiable in $\delta$ and $\frac{\partial}{\partial \delta}\psi\left(\delta,x,y,z,a\right)$ is continuous, then then any subsequence of $\left\{ f_{k}\right\}$ has a a locally uniform convergent subsequence, and every limit $f$ is differentiable and satisfies
    \[f'(r) = \begin{cases}
        \frac{\partial \psi\left(\delta,x,y,z,a\right)}{\partial\delta} \left(0,r,f(r),g(r),a\right) & \{f_k\}\text{ is recursive discrete+} \\
        -\frac{\partial \psi\left(\delta,x,y,z,a\right)}{\partial\delta} \left(0,r,f(r),g(r),a\right) & \{f_k\}\text{ is recursive discrete-}
    \end{cases},\]
    where $a$ is some limit of the subsequence $\{a_{k_j}\}$ and $f_{k_j}$ is the subsequence the converges into~$f$. \label{lemma:disc_to_cont_using_derivative:continuous_derivative:open}
\end{enumerate}
\end{lemma}

\begin{proof}
Any interval can be represented as an increasing union of compact intervals, so there exist $J_1 \subseteq J_2 \subseteq \dots $ such that they are all compact and $\bigcup_{i=1}^\infty \mathcal{X}$. By applying Lemma \ref{lemma:disc_to_cont_uniform_exists} to all of those, with suitable $\tilde{\mathcal{Y}}, \tilde{\mathcal{Z}}$ intervals, we get that for all of those intervals, any subsequence has a convergent subsequence, any locally Lipschitz function on a compact set is Lipschitz.

For $J_1$ we take a convergent subsequence $\{f_{1,k}\}_{k\in\N}$, and for every $i\ge 1$ we define the subsequence $\{f_{i+1,k}\}_{k\in\N}$ to be a subsequence of $\{f_{i,k}\}_{k\in\N}$ that is convergent on $J_{i+1}$, and thus it converges on $J_{i+1}$ to a continuous function, and agrees with the previous subsequences on $J_1, \dots, J_i$ respectively.

We get that the subsequence $\left\{f_{k,k}\right\}_{k\in\N}$ is convergent on $\mathcal{X}$, and converges to a continuous function. This completes the proof for Case~\eqref{lemma:disc_to_cont_using_derivative:lipschitz:open} of the lemma.
We do the same thing for Case~\eqref{lemma:disc_to_cont_using_derivative:continuous_derivative:open}, but apply Lemma~\ref{lemma:disc_to_cont_using_derivative} on each interval, to get a convergent subsequence where the limit is differentiable and satisfies the ODE presented by Lemma~\ref{lemma:disc_to_cont_using_derivative}. This completes the proof for Case~\eqref{lemma:disc_to_cont_using_derivative:continuous_derivative:open} of the lemma.
\end{proof}

\begin{remark} \label{remark:lemmas_work_on_sub_sequence}
Lemmas~\ref{lemma:disc_to_cont_uniform_exists}-\ref{lemma:disc_to_cont_using_derivative:open} are also true if we apply them to a subsequence $\left\{ f_{k_j} \right\}$ and each of the conditions need only be true for every $k_j$. 
\end{remark}

\section{Alternative Proof for Two Bidders}
\label{appndx:sec:cont}

We will show here how the get a direct proof for the Theorem~\ref{thm:two_bidders:main:equilibrium} (and by extension Theorem~\ref{thm:two_bidders:main:among_all_truthful} as well).

\vspace{0.05in}
\noindent{\bf Proof overview.}
Consider the functions $Q^*_k$ and $F^*_k$ we get from Theorem \ref{thm:two_bidders:disc} for each $k$.
In Section~\ref{sub:converge} we will prove that the sequences $\left\{Q^*_k\right\}$ and $\left\{F^*_k\right\}$ are convergent, and converge to $Q^*$ and $F^*$, respectively. This is the most technically-challenging part of the proof, and is made possible by our toolbox from Section~\ref{appndx:sec:toolbox}. 
In Section~\ref{sub:equil} we will prove that $(Q^*, F^*)$ is an equilibrium: For one side, we will show that any distribution of the bidders can be approximated well enough using $\rvec_k$ for sufficiently large $k$. For the other side, we use Claim~\ref{clm:two:Q_is_optimal}, as it only utilize the indifference of the seller over $[0, \overv]$. We conclude that any other distribution cannot gain a lower expected revenue when the seller uses $Q^*$.

\subsection{Proof of Convergence}
\label{sub:converge}
We define sequences $\left\{Q^*_k\right\}$ and $\left\{F^*_k\right\}$, such that $\left(Q^*_k, F^*_k \right)$ is an equilibrium for the discrete setting with parameter $k$ that we get from Theorem \ref{thm:two_bidders:disc}. We show in this section that $\left\{F^*_k\right\}$ converges to $F^*$ (Claim~\ref{clm:two_bidders:cont:F_k_convergent}) and $\left\{Q^*_k\right\}$ converges to $Q^*$ (Claim~\ref{clm:two_bidders:cont:Q_k_convergent}). For the convergence of $\left\{F^*_k\right\}$ we use Lemma~\ref{lemma:disc_to_cont_using_derivative:open}; for the convergence of $\left\{Q^*_k\right\}$ we use both Lemma~\ref{lemma:disc_to_cont_using_derivative:open} and Lemma~\ref{lemma:disc_to_cont_uniform_exists} as it provides guarantees uniform converges.

\begin{claim}\label{clm:two_bidders:cont:F_k_convergent}
The sequence $\left\{F^*_k\right\}_{k\in \N}$ is uniformly convergent to $F^*$ where
\[F^*(v) = \begin{cases}
    0 & v \in [0, \underv] \\
    1-\frac{\underv}{v} & v\in [\underv, \overv) \\
    1 & v = \overv
    \end{cases},\]
Where $\underv$ is the solution to $\underv\left(1+\log\overv-\log\underv\right) = \mu$.
\end{claim}
\begin{proof}
From Theorem \ref{thm:two_bidders:disc}, for each $k>1$  and $1 \le i \le k-1$ the function $F^*_k(v)$ satisfies:
\[F^*_k(r_{k,i+1}) = \sqrt{\frac{\overv}{k} \cdot \frac{2 }{r_{k,i} +\frac{\overv}{k}}(F^*_k(r_{k,i}) - F^*_k(r_{k,i})^2) + F^*_k(r_{k,i})^2};\]
So we define a function $\psi(\delta, r, y): [0,\overv] \times (0,\overv] \times [0,1] \to [0,1]$ by
\[\psi(\delta, r, y) = \sqrt{\frac{2\delta}{r +\delta}(y - y^2) + y^2}.\]
For each $k>1$  and $1 < i \le k-1$ the function $F^*_k(v)$ satisfies:
\begin{equation}
    F^*_k(r_{k,i+1}) = \psi\left(\frac{\overv}{k}, r_{k,i}, F^*_k(r_{k,i}) \right). \label{eq:proof:clm:two_bidders:F_k_convergent:psi_induction}
\end{equation}
For this proof, we would like to change the value of $F^*_k$ on $r_{k,k+1}=\overv$ from $1$ to $\psi\left(\frac{\overv}{k}, r_{k,k}, F^*_k(r_{k,k}) \right)$, so property $\eqref{eq:proof:clm:two_bidders:F_k_convergent:psi_induction}$ is true for $k$ as well.
For every $k$, we define a new function $\tilde{F}_k$ by $\tilde{F}_k(r) = F^*_k(r)$ for every $r\in[0,\overv)$ and $\tilde{F}(\overv) = \psi\left(\frac{\overv}{k}, r_{k,k}, \tilde{F}(r_{k,k}) \right)$.
Now, for every $k>1$ and $1 < i \le k$
\[ \tilde{F}_k(r_{k,i+1}) = \psi\left(\frac{\overv}{k}, r_{k,i}, \tilde{F}_k(r_{k,i}) \right). \label{eq:proof:clm:two_bidders:F_k_convergent:psi_induction_tilde}\]
This will allow us get locally uniform convergence on a compact interval $[0,\overv]$, which means uniform convergence.

In order to prove that the sequence $\left\{\tilde{F}_k\right\}$ is uniformly convergent, we need to analyze $\psi$.
We can see is that for every $y\in [0,1]$ and every $r\in (0,\overv]$
\begin{equation}
    \psi(0, r, y) = y, \label{eq:proof:clm:two_bidders:F_k_convergent:psi_0_r_y_is_y}
\end{equation}
$\psi$ is continuous, and properties \eqref{eq:proof:clm:two_bidders:F_k_convergent:psi_0_r_y_is_y} and \eqref{eq:proof:clm:two_bidders:F_k_convergent:psi_induction} fit into Definition \ref{def:psi_recursive_discrete}'s conditions \eqref{eq:def:psi_recursive_discrete_delta_0} and \eqref{eq:def:psi_recursive_discrete+} respectively, which mean that $\left\{\tilde{F}_k\right\}$ is recursive discrete.
In order to prove convergence, we need to analyze $\psi$ further.
\[\frac{\partial}{\partial \delta} \psi(\delta, r, y) = \frac{(y-y^2)\left( \frac{2}{r+\delta} - \frac{2\delta}{(r+\delta)^2} \right)}{2\sqrt{\frac{2\delta}{r +\delta}(y - y^2) + y^2}} = \frac{y(1-y)\left( \frac{1}{r+\delta} - \frac{\delta}{(r+\delta)^2} \right)}{\psi(\delta, r, y)}.\]
We can see that it is non negative when $\delta \in [0,\overv]$, $r\in (0,\overv]$, and $y\in(0,1]$.
At $y=0$ this equation is not defined, but $\psi(\delta, r, 0) = 0$ for every $\delta$ and $r$, which means that the derivative at $y=0$ is always 0.
The derivative at $\delta = 0$ and $r,y>0$ is $\frac{1-y}{r}$ which implies that the derivative is not continuous at $\delta = 0$ and $y=0$.
However, for any $r>0$ we can get an upper bound by using $\psi(\delta, r, y) \ge y$
\[\frac{\partial}{\partial \delta} \psi(\delta, r, y) \le \left(1-y\right)\left(\frac{1}{r+\delta}-\frac{\delta}{(r+\delta)^{2}}\right) \le \frac{1-y}{r}.\]
This upper bound implies that $\psi$ is locally Lipschitz.

Since $\left\{\tilde{F}_k\right\}$ is recursive discrete, and the corresponding $\psi$ is locally Lipschitz, we can apply Lemma \ref{lemma:disc_to_cont_using_derivative:open}.\ref{lemma:disc_to_cont_using_derivative:lipschitz:open} (for $\mathcal{X} = (0,\overv]$, and $\mathcal{Y} = [0,1]$) and we get that any subsequence of $\left\{\tilde{F}_k\right\}$ has a a locally uniform convergent subsequence on $(0,\overv]$, and each limit is continuous.

We take an arbitrary locally uniform convergent subsequence $\left\{\tilde{F}_{k_j}\right\}$, and denote its limit by $\tilde{F}$.
For every $j$ we know that $\int_0^{\overv} (1-\tilde{F}_{k_j}(v))dv = \mu$, which implies that $\int_0^{\overv} (1-\tilde{F}(v))dv = \mu$ as well.
We denote by $\underv = \inf_{r\in(0,\overv]}\left\{ \tilde{F}(r) > 0\right\}$.
This is the lowest value that satisfies $\frac{\partial}{\partial \delta} \psi(\delta, r, \tilde{F}(r))$ is continuous on $r \in (\underv, \overv]$.

Since $\left\{\tilde{F}_{k_j}\right\}$ is convergent to $\tilde{F}$, and all of those are non decreasing functions, then on any compact sub-interval of $(\underv, \overv]$ there is a positive minimum value $a$ of $\tilde{F}$ over that sun-interval, which implies that for any large enough $j$, $F_{k_j}$ is at least $\frac{a}{2}$ on the same compact sub-interval.
Together with the continuous partial derivative, and Remark \ref{remark:lemmas_work_on_sub_sequence}, we can use Lemma \ref{lemma:disc_to_cont_using_derivative:open}.\ref{lemma:disc_to_cont_using_derivative:continuous_derivative:open} with the intervals $r\in (\underv, \overv]$ and $y\in (0,1]$, which implies that for any $r\in (\underv, \overv]$
\begin{align*}
    \frac{d}{dr}\tilde{F}(r) &= \frac{\partial\psi(\delta, r, y)}{\partial \delta} \left(0, r, \tilde{F}(r)\right) \\
    &= \frac{1-\tilde{F}(r)}{r}.
\end{align*}
This implies that
\[\tilde{F}(r) = 1 - \frac{c}{r} \quad r\in (\underv, \overv]\]
for some constant $c$.
The function $\tilde{F}$ is defined from $(0, \overv]$ to $(0,1]$, continuous, and equal to 0 outside of $(\underv, \overv]$, which suggest that
\[0 = \lim_{r\to \underv^+}\tilde{F}(r) = 1 - \frac{c}{\underv}.\]
Therefore $\underv = c$ and we get that
\[\tilde{F}(r) = \begin{cases}
    0 & r \in (0, \underv] \\
    1-\frac{\underv}{r} & r\in [\underv, \overv]
    \end{cases}.\]
To get $\underv$ we use the equality $\int_0^{\overv} (1-\tilde{F}(r))dv = \mu$, and we get that $\underv$ is the unique solution to
\[\underv\left(1+\log\overv-\log\underv\right) = \mu.\]
Since the functions $\tilde{F}_k$ are all positive and non decreasing, we get that $\lim_{k\to \infty}\tilde{F}_k(0) = 0$ since it must be non negative and $\tilde{F}_k$ is non decreasing;
Thus, we can define $\tilde{F}(0) = 0$.
The convergence on $[0,\overv]$ is uniform since it is locally uniform convergence to a continuous function on a compact set.

We have taken an arbitrary limit $\tilde{F}$ and proved that it is a specific function. Since any subsequence has a uniformly convergent subsequence, the sequence $\{\tilde{F}_k\}_{k\in \N}$ is uniformly convergent to $\tilde{F}$.

Since $\{\tilde{F}_k\}$ and $\{F^*_k\}$ are identical on $[0,\overv)$, we get uniform convergence there for $\{F^*_k\}$. Since $F^*_k(\overv) = 1$ for every $k$, we get convergence there as well. And we conclude that $\{F^*_k\}$ convergence uniformly to $F^*$ that is defined in the claim's description.
\end{proof}

\begin{claim}  \label{clm:two_bidders:cont:Q_k_convergent}
The sequence $\left\{Q^*_k\right\}_{k\in \N}$ is convergent to $Q^*$ where
\[Q^*(r) = \begin{cases}
         \left( 1 - \nicefrac{\underv}{\overv} \right) \cdot \left(\frac{1}{\log \overv - \log \underv}\right) & r\in [0, \underv]\\
        \left( 1 - \nicefrac{\underv}{\overv} \right) \cdot \left( \frac{r}{r-\underv} \right) \cdot \left( \frac{\log r - \log \underv}{ \log \overv - \log \underv} \right) & r\in (\underv, \overv]
    \end{cases},\]
Where $\underv$ is the solution to $\underv\left(1+\log\overv-\log\underv\right) = \mu$.
\end{claim}

\begin{proof}
From Theorem \ref{thm:two_bidders:disc}, for each $k>1$  and $2 \le i \le k-1$ the function $Q^*_k(v)$ satisfies:
\[Q^*_k(r_{k,i-1}) = Q^*_k(r_{k,i}) + \left[Q^*_k(r_{k,i})(1-F^*_k(r_{k,i})) - Q^*_k(0)(1-F^*_k(0)) \right] \cdot\frac{\overv}{k r_{k,i} F^*_k(r_{k,i})};\]
So we define a function $\psi(\delta, r, y,z,a)$ by
\[\psi(\delta, r, y,z,a) = y + \left[y(1-z) - a) \right] \cdot\frac{\delta}{rz},\]
and now for each $k>1$  and $2 \le i \le k-1$ the function $F^*_k(v)$ satisfies:
\begin{equation}
    Q^*_k(r_{k,i-1}) = \psi\left(\frac{\overv}{k}, r_{k,i}, Q^*_k(r_{k,i}), F^*_k(r_{k,i}), Q^*_k(0)(1-F^*_k(0)) \right). \label{eq:proof:clm:two_bidders:Q_k_convergent:psi_induction}
\end{equation}
This fits property \eqref{eq:def:psi_recursive_discrete-}, since $\psi$ also satisfies property \eqref{eq:def:psi_recursive_discrete_delta_0} from the same definition, and from Claim~\ref{clm:two_bidders:cont:F_k_convergent} $\{F^*_k\}$ is uniformly convergent to a continuous function $F^*$ on $[0,\overv)$.
We get that the sequence $\left\{Q^*_k\right\}$ is recursive discrete- on $\mathcal{I} = [0,\overv], \mathcal{X} = (\underv, \overv), \mathcal{Y}=[0,1], \mathcal{Z} = (0,1]$, and $\mathcal{C} = [0,1]$.
We defined $\mathcal{X}$ to be $(\underv, \overv)$, because $F^*$ is not continuous at $\overv$, and because $F^*(v) = 0$ for every $v\in[0,\underv]$, so $\psi(\cdot, \underv, \cdot, F^*(v), \cdot)$ is not defined.

Next, we differentiate by $\delta$,
\begin{align*}
    \frac{\partial}{\partial\delta} \psi(\delta, r, y,z,a) =\frac{y(1-z) - a}{rz},
\end{align*}
which is continuous when $r,z>0$.

For every compact sub-interval $I \subset (\underv, \overv)$, we can find a minimum $\underline{z}$ of $F^*$ over $I$. Therefore $\left[\nicefrac{\underline{z}}{2}, 1\right]$ is a compact sub-interval of $\mathcal{Y}$ that contains $F^*_k(I)$ for every large enough $k$.

According to Lemma \ref{lemma:disc_to_cont_using_derivative:open}.\ref{lemma:disc_to_cont_using_derivative:continuous_derivative:open} any subsequence of $\left\{Q^*_k\right\}$ has a convergent subsequence. We take such convergent subsequence $\left\{Q^*_{k_j}\right\}_{j\in\N}$, an it is convergent to a continuous function $Q^*$ on $\mathcal{X} = (\underv, \overv)$, and it satisfies
\begin{align*}
    \frac{dQ^*(r)}{dr} &= -\frac{\partial \psi\left(\delta,r,y,z,a\right)}{\partial\delta} \left(0,r,Q^*(r),F^*(r),a\right) \\
    &= -\frac{Q^*(r)(1-F^*(r)) - a}{r\cdot F^*(r)} \\
    &= \frac{a - Q^*(r)\cdot \frac{\underv}{r}}{r-\underv},
\end{align*}
where $a$ is some limit of the sequence $\left\{ Q^*_{k_j}(0)(1-F^*_{k_j}(0)) \right\}_{j\in \N}$.
The family of solutions for this ODE is
\[Q^*(r) = \frac{r\left( a \log r + c \right)}{r-\underv}\qquad r\in (\underv, \overv),\]
For some constant $c$.
From Theorem \ref{thm:two_bidders:disc}, for every $k$
\[ Q^*_k(0)(1-F^*_k(0)) = (k-1)(1 - Q^*_k(r_{k,k-1})) F^*_k(r_{k,k}) + (1-F^*_k(r_{k,k})).\]
Since $Q^*_k(0)(1-F^*_k(0)) \in [0,1]$ for every $k$, and $\lim_{k\to \infty }F^*_k(r_{k,k}) = 1-\nicefrac{\underv}{\overv} > 0$. Therefore, we get that $\lim_{k\to \infty} Q^*_k(r_{k,k-1}) = 1$.

We can see that $\psi$ is Lipschitz (in $\delta$) for $r \in [v', \overv)$ for every $v' > \underv$, so from Lemma \ref{lemma:disc_to_cont_uniform_exists}, $\{Q^*_{k_j}\}_{j\in\N}$ has a uniformly convergent subsequence on $[v', \overv)$ (and not just locally uniform).
Combined with $\lim_{k\to \infty} Q^*_k(r_{k,k-1}) = 1$, we get that $\lim_{r \to \overv}Q^*(r) = 1$.
We can use this to find $c$,
\begin{align*}
    &\lim_{r \to \overv}Q^*(r) = 1 \\
    \implies& \frac{\overv \left( a \log \overv + c \right)}{\overv-\underv} = 1 \\
    \implies& c = 1 - \nicefrac{\underv}{\overv} - a \log \overv.
\end{align*}
So far, we have that
\[Q^*(r) = \frac{r\left( a (\log r - \log \overv) + 1 - \nicefrac{\underv}{\overv} \right)}{r-\underv}\qquad r\in (\underv, \overv).\]
As a limit of non decreasing function in $\mathcal{Y} = [0,1]$, we get that $Q^*$ is non decreasing in $[0,1]$.
Therefore, the limit $\lim_{r\to \underv^+} Q^*(r)$ exists, and it is in $[0,1]$. This implies that
\[\lim_{r\to \underv^+} r\left( a (\log r - \log \overv) + 1 - \nicefrac{\underv}{\overv} \right) = 0.\]
Which means that
\[a = \frac{1 - \nicefrac{\underv}{\overv}}{\log \overv - \log \underv},\]
and thus, for every $r \in (\underv, \overv)$
\begin{align*}
    Q^*(r) &= \left( 1 - \nicefrac{\underv}{\overv} \right) \cdot \left( \frac{r}{r-\underv} \right) \cdot \left( \frac{\log r - \log \overv}{ \log \overv - \log \underv} + 1\right)\\
    &=\left( 1 - \nicefrac{\underv}{\overv} \right) \cdot \left( \frac{r}{r-\underv} \right) \cdot \left( \frac{\log r - \log \underv}{ \log \overv - \log \underv} \right).
\end{align*}
Since $Q^*_k(\overv) = 1$ for every $k$, then by extending $Q^*$ to $\overv$, and defining
\[Q^*(\overv) = \left( 1 - \nicefrac{\underv}{\overv} \right) \cdot \left( \frac{\overv}{\overv-\underv} \right) \cdot \left( \frac{\log \overv - \log \underv}{ \log \overv - \log \underv} \right) = 1,\]
we get convergence on $\overv$ as well.
$a$ was defined as some arbitrary limit of $\left\{ Q^*_{k_j}(0)(1-F^*_{k_j}(0) \right\}_{j\in \N}$, since we have proved that $a=\frac{1 - \nicefrac{\underv}{\overv}}{\log \overv - \log \underv}$, which is only dependent on $\overv$ and $\underv$, then every limit is identical, and the subsequence is convergent. Since $\lim_{k \to \infty} F^*_k(0) = 0$, we get
\[\lim_{j\to \infty} Q^*_{k_j}(0) = \lim_{j\to \infty} Q^*_{k_j}(0)(1-F^*_{k_j}(0) = a.\]
Likewise, we can see that $lim_{r \to \underv^+} Q^*(r) = a$. Since $Q^*_k$ is non decreasing for every $k$, we conclude that for every $r\in [0,\underv]$
\[\lim_{j \to \infty} Q^*_{k_j}(r) = a,\]
which allows us to extend the definition of $Q^*$ to the entire interval $[0,\overv]$.

We have defined $Q^*$ as an arbitrary limit of $\left\{Q^*_{k}\right\}$, and proved it is a specific function. Since every subsequence has a convergent subsequence, we conclude that $\lim_{k \to \infty} Q^*_k = Q^*$ where
\[Q^*(r) = \begin{cases}
         \left( 1 - \nicefrac{\underv}{\overv} \right) \cdot \left(\frac{1}{\log \overv - \log \underv}\right) & r\in [0, \underv]\\
        \left( 1 - \nicefrac{\underv}{\overv} \right) \cdot \left( \frac{r}{r-\underv} \right) \cdot \left( \frac{\log r - \log \underv}{ \log \overv - \log \underv} \right) & r\in (\underv, \overv]
    \end{cases}.\]
\end{proof}

\subsection{Proof of Equilibrium in the Limit}
\label{sub:equil}
While we have proved convergence of the equilibrium points of the discrete case to $Q^*$ and $F^*$, it is not trivial that $(Q^*, F^*)$ is in itself an equilibrium of the general case.

We can however prove Claim~\ref{clm:two:F_is_optimal} directly using the convergence and correctness in the discrete case.

Consider random variables $R_k \sim Q^*_k$ for every $k$, and $R \sim Q^*$, where $Q^*_k$, and $Q^*$ are described by Theorem \ref{thm:two_bidders:disc} and Claim \ref{clm:two_bidders:cont:Q_k_convergent}.
Since $Q^*_k \to Q^*$, then $R_k \to R$ in distribution.
In order to prove Claim \ref{clm:two:F_is_optimal}, we need to prove that for any $F \in \calF$,
\[\mathbb{E}\left[\Rev_{F}(R)\right] \ge \mathbb{E}\left[\Rev_{F^*}(R)\right].\]

\begin{claim} \label{clm:two_bidders:cont:limit_of_general_F_over_Q*}
For every distribution function $F$ over $[0,\overv]$,
\[\lim_{k \to \infty} \mathbb{E}\left[ \Rev_{F}(R_k) \right] = \mathbb{E}\left[ \Rev_{F}(R) \right].\]
\end{claim}
\begin{proof}
For a given distribution function $F$, the function $\Rev_{F}$ is continuous at $0$ and at every point where $F$ is continuous. Since $F$ is monotone, then the set of all discontinuities is countable.
$Q^*$ only point of discontinuity is 0, and therefore $\Pr_{r\sim Q^*}(F \text{ is not continuous at } r) = 0$.
Now, we can use Portmanteau Theorem which completes the proof.
\end{proof}
Before we prove that $F^*$ is the adversary's best response, we need to prove a rather intuitive claim.

\begin{claim}
For $F^*_k$ as defined in the discrete case
\[\lim_{k \to \infty} \mathbb{E}\left[ \Rev_{F_k^*}(R_k) \right] = \mathbb{E}\left[ \Rev_{F^*}(R) \right].\]
\end{claim}

\begin{proof}
Since $F^*_k \to F^*$ uniformly, then for any $\varepsilon > 0$, for large enough $k$
\[\sup_{v\in [0,\overv]} \abs*{F^*(v) - F^*_k(v)} < \varepsilon.\]
Therefore, for any $r\in [0,\overv]$:
\begin{align*}
    \abs*{\Rev_{F^*}(r) - \Rev_{F^*_k}(r)} &= \abs*{ r \left( F^{*2}_k(r) - F^{*2}(r)\right) + \int_{r}^{\overv} \left(1 - F^*(v) \right)^2 - \left(1 - F^*_k(v) \right)^2 dv }\\
    &\le 2\overv \abs*{F^*_k(r) - F^*(r)} + 2\int_{r}^{\overv} \abs*{F^*_k(v) - F^*(v)}dv \\
    &\le 4\overv \varepsilon.
\end{align*}
From Claim \ref{clm:two_bidders:cont:limit_of_general_F_over_Q*}, we get that:
\[\lim_{k \to \infty} E\left[ \Rev_{F^*}(R_k) \right] = E\left[ \Rev_{F^*}(R) \right].\]
Therefore, for large enough $k$:
\begin{align*}
    \abs*{E\left[ \Rev_{F^*}(R) \right] - E\left[ \Rev_{F_k^*}(R_k) \right]} \le& \abs*{E \left[ \Rev_{F^*}(R) \right] - E\left[ \Rev_{F^*}(R_k) \right]} \\
    &+ \abs*{E\left[ \Rev_{F^*}(R_k) \right] - E\left[ \Rev_{F_k^*}(R_k) \right]} \\
    \le& \varepsilon + 4\overv\varepsilon.
\end{align*}
\end{proof}

Now, we can check more general distributions, and prove that for every $F$,
\[\mathbb{E}\left[ \Rev_{F^*}(R) \right] \ge \mathbb{E}\left[ \Rev_{F}(R) \right].\]

\begin{claim} \label{clm:two_bidders:cont:general_F_is_close_to_some_F_k}
For every $F\in \calF$ there exists a sequence $\{F_k\}$ where $F_k \in \calF_k$ for every $k$, and for every $\varepsilon > 0$, for sufficiently large $k$,
\[\abs*{\mathbb{E} \left[ \Rev_{F}(R_k) \right] - \mathbb{E}\left[ \Rev_{F_k}(R_k) \right]} < \varepsilon.\]
\end{claim}
\begin{proof}
Fix $F \in \calF$ and $\varepsilon \in (0,1)$.
For any $k$ we define $F_k \in \calF_k$ by $F_k(r_{k,i}) = \frac{k}{\overv}\cdot \int_{r_{k,i}}^{r_{k,i+1}} F(v)dv$ for every $1 \le i \le k$. We immediately get that $F(r_{k,i}) \le F_k(v) \le F(r_{k,i+1})$ for every $1 \le i \le k$ and $v\in [r_{k,i}, r_{k,i+1})$, and thus
\[\abs*{F(v) - F_k(v)} \le F(r_{k,i+1}) - F(r_{k,i}).\]
Denote by $I_{\varepsilon, k} = \left\{ i \in \{1,\dots,k\} \mid F\left(r_{k,i+1}\right) - F(r_{k,i}) > \varepsilon \right\}$. It is easy to see that $\abs*{I_{\varepsilon, k}} < \frac{1}{\varepsilon}$.

For every reserve price $r$ and $k>\frac{\overv}{\varepsilon^2}$, we can get an upper bound to the difference in the revenue between $F$ and $F_k$,
\begin{align*}
    \abs*{\Rev_{F}(r) - \Rev_{F_k}(r)} &= \abs*{r \left( F^2_k(r) - F^2(r)\right) + \int_{r}^{\overv} \left(1 - F(v) \right)^2 - \left(1 - F_k(v) \right)^2 dv } \\
    &\le 2r \abs*{F_k(r) - F(r)} + 2\int_r^{\overv} \abs*{F_k(v) - F(v)}dv \\
    &\le 2r \abs*{F_k(r) - F(r)} + 2\overv \varepsilon + 2\sum_{r_{k,i} \in I_{\varepsilon, k}}\int_{r_{k,i}}^{r_{k,i+1}} \abs*{F_k(v) - F(v)}dv \\
    &\le 2r \abs*{F_k(r) - F(r)} + 2\overv \varepsilon + \frac{2\overv}{k\varepsilon} \\
    &\le 2r \abs*{F_k(r) - F(r)} + 4\overv \varepsilon.
\end{align*}

Since $Q^*_k$ is uniformly convergent to $Q^*$ on $[\underv(1+4\varepsilon), \overv]$ (Lemma \ref{lemma:disc_to_cont_uniform_exists} on $[\underv(1+4\varepsilon), \overv)$ and adding another point $\overv$ to the domain), which is Lipschitz, we can get an upper bound for $\Pr(R_k = r_{k,i})$, such that for large enough $k$ and any $r_{k,i} \in [\underv(1+4\varepsilon), \overv]$
\[\Pr(R_k = r_{k,i}) < \varepsilon^2.\]
Also
\[\sum_{r_{k,i} \in \rvec_k \cap [0,\underv\left(1+4\varepsilon \right)]} \Pr(R_k = r_{k,i}) = F_k^*\left(\underv\left(1+4\varepsilon \right)\right) \le 1-\frac{1}{1+4\varepsilon} + \varepsilon < 5\varepsilon.\]
Therefore
\begin{align*}
    \mathbb{E}\left[\abs*{F_k(R_k) - F(R_k)} \right] &= \sum_{i=1}^k \Pr(R_k = r_{k,i})\abs*{F_k(r_{k,i}) - F(r_{k,i})} \\
    &< 5\varepsilon + \sum_{i=1}^k \Pr(R_k = r_{k,i})\varepsilon + \varepsilon^2\sum_{i\in I_{\varepsilon, k}} \abs*{F_k(r_{k,i}) - F(r_{k,i})}\\
    &< 5\varepsilon + \varepsilon + \varepsilon^2 \cdot \frac{1}{\varepsilon}\\
    &= 7\varepsilon.
\end{align*}

Putting it all together, for large enough $k$
\begin{align*}
    \abs*{\mathbb{E} \left[ \Rev_{F}(R_k) \right] - \mathbb{E}\left[ \Rev_{F_k}(R_k) \right]} &\le \mathbb{E} \left[ \abs*{\Rev_{F}(R_k) -  \Rev_{F_k}(R_k)} \right] \\
    &< 4\overv\varepsilon +  2\overv \mathbb{E}\left[\abs*{F_k(R_k) - F(R_k)} \right]\\
    &< 18\overv\varepsilon.
\end{align*}
\end{proof}

\begin{corollary} \label{cor:two_bidders:cont:limit_of_revenue_of_genral_F}
For any $F \in \calF$ there is a sequence $\{F_k\}$ where $F_k \in \calF_k$ for every $k$ such that
\[\lim_{k \to \infty} \mathbb{E}[\Rev_{F_k}(R_k)] = \mathbb{E}[\Rev_{F}(R)].\]
\end{corollary}
This is easy to see from Claims \ref{clm:two_bidders:cont:limit_of_general_F_over_Q*} and \ref{clm:two_bidders:cont:general_F_is_close_to_some_F_k}.

\begin{corollary}
For any $F \in \calF$,
\[\mathbb{E}[\Rev_{F}(R)] \ge \mathbb{E}[\Rev_{F^*}(R)].\]
\end{corollary}
\begin{proof}
Fix $F \in \calF$. We use the sequence $\{F_k\}$ from Corollary \ref{cor:two_bidders:cont:limit_of_revenue_of_genral_F}. Since $F^*_k$ is the adversary's optimal response to $Q_k^*$ then $\mathbb{E}[\Rev_{F_k}(R_k)] \ge \mathbb{E}[\Rev_{F^*_k}(R_k)$ for any $k$. And thus
\[\mathbb{E}[\Rev_{F}(R)] = \lim_{k \to \infty} \mathbb{E}[\Rev_{F_k}(R_k)] \ge \lim_{k \to \infty} \mathbb{E}[\Rev_{F_k^*}(R_k)] = \mathbb{E}[\Rev_{F^*}(R)].\]
\end{proof}
This Completes the direct proof for Claim~\ref{clm:two:F_is_optimal}. The rest of the proof is identical to Section~\ref{subsec:two:proof} (from Claim~\ref{clm:two:Q_is_optimal}).

\section{Missing Proofs from Section~\ref{sec:two_bidders}}
\subsection*{Proof of Claim \ref{clm:two_bidders:disc:nature_induction}}
\begin{proof}
For every $1 \le i \le k-1$
\begin{align*}
    \frac{\partial}{\partial q_i}\calL_k(\xvec, \qvec) &= r_{k,i}(1-x_i^2) + \frac{\overv}{k} \left(1 - x_i \right)^2 - r_{k,i+1}(1-x_{i+1}^2) \\
    &= r_{k,i+1}x_{i+1}^2 -\frac{2\overv}{k}\cdot x_i + (r_{k,i+1} - 2r_i)x_i^2.
\end{align*}
Since $x_i \ge 0 $ for all $1\le i \le k-1$
\begin{align*}
    \frac{\partial}{\partial q_i}\Rev(\xvec, \qvec) &= 0\\
    \iff x_{i+1} &= \sqrt{\frac{\frac{2\overv}{k}\cdot x_i + \left( r_i - \frac{\overv}{k} \right)x_i^2}{r_{k,i} + \frac{\overv}{k}}}\\
    &= \sqrt{\frac{\overv}{k} \cdot \frac{2 }{r_{k,i} +\nicefrac{\overv}{k}}(x_i - x_i^2) + x_i^2} .
\end{align*}
\end{proof}

\subsection*{Proof of Claim \ref{clm:two_bidders:disc:seller_induction}}
\begin{proof}
Differentiating $\calL_k(\qvec, \xvec)$ by $x_i$ is
\[\frac{\partial}{\partial x_i}\calL_k(\qvec, \xvec) = -2(q_i-q_{i-1}) r_{k,i}x_i  -  2\nicefrac{\overv}{k} \cdot q_i \left(1 - x_i \right) + \lambda \cdot \nicefrac{\overv}{k}.\]
Since $r_{k,i}, x_i > 0$ for every $2\le i \le k$, we can easily get
\[\frac{\partial}{\partial x_i}\calL_k(\qvec, \xvec) = 0 \iff q_{i-1}=q_{i}-\frac{\nicefrac{\lambda}{2}-q_{i}\left(1-x_{i}\right)}{r_{k,i}x_{i}}\cdot\nicefrac{\overv}{k}.\]

Since $r_{k,1} = 0$, we get
\[\frac{\partial}{\partial x_1}\calL_k(\qvec, \xvec) = -  2\nicefrac{\overv}{k} \cdot q_1 \left(1 - x_1 \right) + \lambda \cdot \nicefrac{\overv}{k},\]
which means that
\[\frac{\partial}{\partial x_1}\calL_k(\qvec, \xvec) = 0 \iff q_1=\frac{\lambda}{2(1-x_1)}.\]
\end{proof}

\section{Missing Proofs from Section~\ref{sec:multi_bidders}}
\subsection*{Proof of Claim \ref{clm:multi:revenue_expression}}
\begin{proof}
For a fixed price $r$ the expected revenue is given by
\begin{align*}
    \Rev_{F}(r) &= E\left[r\cdot1_{v^{\left(2\right)}\le r<v^{\left(1\right)}}\right]+E\left[v^{\left(2\right)}\cdot1_{r<v^{\left(2\right)}}\right] \\
    &=r\left(1-F^{n}\left(r\right)\right)+\int_{r}^{\overv}\left(1-nF^{n-1}\left(v\right)+\left(n-1\right)F^{n}\left(v\right)\right)dv.
\end{align*}

When the seller randomizes a reserve price using distribution $Q \in \calQ$ that is continuous over $(0,\overv]$ and differentiable almost everywhere, the expected revenue would be
\begin{align*}
    \Rev_{F}(Q) =& \int_{-\infty}^{\overv}\left(r\left(1-F^{n}\left(r\right)\right)+\int_{r}^{\overv}\left(1-nF^{n-1}\left(v\right)+\left(n-1\right)F^{n}\left(v\right)\right)dv\right)dQ\left(r\right). \\
    =~& Q(0)\int_{0}^{\overv}\left(1-nF^{n-1}\left(v\right)+\left(n-1\right)F^{n}\left(v\right)\right)dv \\
    &+ \int_{0}^{\overv}Q'(r)\left(r\left(1-F^{n}\left(r\right)\right)+\int_{r}^{\overv}\left(1-nF^{n-1}\left(v\right)+\left(n-1\right)F^{n}\left(v\right)\right)dv\right)dr \\
    =&\int_{0}^{\overv}\left(Q'\left(v\right)v\left(1-F^{n}\left(v\right)\right)+Q\left(v\right)\left(1-nF^{n-1}\left(v\right)+\left(n-1\right)F^{n}\left(v\right)\right)\right)dv.
\end{align*}
\end{proof}

\subsection{Details for Proof of Claim~\ref{clm:multi:candidate:integrand_min_underv_to_overv}}
\label{appndx:clm:integrand_min_underv_to_overv:details}
\subsubsection*{\bf Second derivative is positive.}

We substitute $Q^*(v) = \left(1-\nicefrac{a}{\overv}\right)^{n-1}\left(\frac{v}{v-a}\right)^{n-1}\left(\frac{n-1-\frac{1}{n-1}+\log\nicefrac{v}{\underv}}{n-1-\frac{1}{n-1}+\log\nicefrac{\overv}{\underv}}\right)$, and $z=1-\nicefrac{a}{v}$, in Eq~\eqref{eq:multi:second_derivative_term}, we get

\begin{align*}
&    \left(\left(n-1\right)Q^{*}\left(v\right)-v\frac{dQ^{*}\left(v\right)}{dv}\right)z-\left(n-2\right)Q^{*}\left(v\right) \\
=&   \left(n-1\right)\left(1-\nicefrac{a}{\overv}\right)^{n-1}\left(\frac{v}{v-a}\right)^{n-1}\left(\frac{n-1-\frac{1}{n-1}+\log\nicefrac{v}{\underv}}{n-1-\frac{1}{n-1}+\log\nicefrac{\overv}{\underv}}\right)\left(1-\nicefrac{a}{v}\right) \\
	&+v\left(1-\nicefrac{a}{\overv}\right)^{n-1}\left(\frac{v}{v-a}\right)^{n-1}\left(\frac{\left(n-1\right)a\log\nicefrac{v}{\underv}+\underv-v}{v\left(v-a\right)\left(n-1-\frac{1}{n-1}+\log\nicefrac{\overv}{\underv}\right)}\right)\left(1-\nicefrac{a}{v}\right) \\
	&-\left(n-2\right)\left(1-\nicefrac{a}{\overv}\right)^{n-1}\left(\frac{v}{v-a}\right)^{n-1}\left(\frac{n-1-\frac{1}{n-1}+\log\nicefrac{v}{\underv}}{n-1-\frac{1}{n-1}+\log\nicefrac{\overv}{\underv}}\right) \\
=&	\left(n-1\right)\left(1-\nicefrac{a}{\overv}\right)^{n-1}\left(\frac{v}{v-a}\right)^{n-2}\left(\frac{n-1-\frac{1}{n-1}+\log\nicefrac{v}{\underv}}{n-1-\frac{1}{n-1}+\log\nicefrac{\overv}{\underv}}\right) \\
	&+\left(1-\nicefrac{a}{\overv}\right)^{n-1}\left(\frac{v}{v-a}\right)^{n-1}\left(\frac{\left(n-1\right)a\log\nicefrac{v}{\underv}+\underv-v}{v\left(n-1-\frac{1}{n-1}+\log\nicefrac{\overv}{\underv}\right)}\right) \\
	&-\left(n-2\right)\left(1-\nicefrac{a}{\overv}\right)^{n-1}\left(\frac{v}{v-a}\right)^{n-1}\left(\frac{n-1-\frac{1}{n-1}+\log\nicefrac{v}{\underv}}{n-1-\frac{1}{n-1}+\log\nicefrac{\overv}{\underv}}\right).
\end{align*}
We divide by $\frac{\left(1-\nicefrac{a}{\overv}\right)^{n-1}\left(\frac{v}{v-a}\right)^{n-1}}{n-1-\frac{1}{n-1}+\log\nicefrac{\overv}{\underv}}>0$ and get
\begin{align*}
    & \left(n-1\right)\left(1-\nicefrac{a}{v}\right)\left(n-1-\frac{1}{n-1}+\log\nicefrac{v}{\underv}\right) \\
    &+\frac{\left(n-1\right)a\log\nicefrac{v}{\underv}+\underv-v}{v}\\
    &-\left(n-2\right)\left(n-1-\frac{1}{n-1}+\log\nicefrac{v}{\underv}\right) \\
    =&\left(\left(n-1\right)\left(1-\nicefrac{a}{v}\right)-\left(n-2\right)\right)\left(n-1-\frac{1}{n-1}+\log\nicefrac{v}{\underv}\right) \\
    &+\left(n-1\right)\nicefrac{a}{v}\cdot\log\nicefrac{v}{\underv}-\left(1-\nicefrac{\underv}{v}\right) \\
    =& \left(1-\left(n-1\right)\nicefrac{a}{v}\right)\left(n-1-\frac{1}{n-1}+\log\nicefrac{v}{\underv}\right)+\left(n-1\right)\nicefrac{a}{v}\cdot\log\nicefrac{v}{\underv}-1+\nicefrac{\underv}{v} \\
    =& \left(n-1-\frac{1}{n-1}+\log\nicefrac{v}{\underv}\right)-\left(n-1\right)\nicefrac{a}{v}\cdot\left(n-1-\frac{1}{n-1}\right)-1+\nicefrac{\underv}{v} \\
    =& \left(n-1-\frac{1}{n-1}+\log\nicefrac{v}{\underv}\right)-\nicefrac{a}{v}\cdot\left(\left(n-1\right)^{2}-1\right)-1+\nicefrac{\underv}{v} \\
    =& \left(n-1-\frac{1}{n-1}+\log\nicefrac{v}{\underv}\right)-\left(1-\nicefrac{a}{v}\right)-\nicefrac{a}{v}\cdot\left(n-1\right)^{2}+\nicefrac{\underv}{v}.
\end{align*}
We substitute $a=\frac{\underv}{(n-1)^2}$ (Eq~\eqref{eq:multi:a_equation}) and we get
\begin{align*}
    \left(n-1-\frac{1}{n-1}+\log\nicefrac{v}{\underv}\right)-\left(1-\frac{\underv}{v\left(n-1\right)^{2}}\right) > \log\nicefrac{v}{\underv}-\left(1-\nicefrac{\underv}{v}\right) \ge 0.
\end{align*}

\subsubsection*{\bf Details for $h_v\left(1-\nicefrac{a}{v}\right) \le h_v(0)$.}
First we express $\tilde{h}_{v}\left(0\right)$ by
\begin{align*}
    \tilde{h}_{v}\left(0\right)=& \left(\frac{v}{v-a}\right)^{n-1}\left(n-1-\frac{1}{n-1}+\log\nicefrac{v}{\underv}\right) \\
    &-\left(\frac{v}{v-a}\right)^{n-1}\left(\frac{\left(n-1\right)a\log\nicefrac{v}{\underv}+\underv-v}{v-a}\right)-n.
\end{align*}
Next we express $\tilde{h}_{v}\left(1-\nicefrac{a}{v}\right)$ by
\begin{align*}
    \tilde{h}_{v}\left(1-\nicefrac{a}{v}\right)=&\left(\frac{v}{v-a}\right)^{n-1}\left(n-1-\frac{1}{n-1}+\log\nicefrac{v}{\underv}\right)\left(1-n\left(1-\nicefrac{a}{v}\right)^{n-1}+\left(n-1\right)\left(1-\nicefrac{a}{v}\right)^{n}\right) \\
    &-\left(\frac{v}{v-a}\right)^{n-1}\left(\frac{\left(n-1\right)a\log\nicefrac{v}{\underv}+\underv-v}{v-a}\right)\left(1-\left(1-\nicefrac{a}{v}\right)^{n}\right)-n\cdot\nicefrac{a}{v}.
\end{align*}

Now we can calculate $h_v\left(1-\nicefrac{a}{v}\right)$.
\begin{align*}
    &\tilde{h}_{v}\left(0\right)-\tilde{h}_{v}\left(1-\nicefrac{a}{v}\right) \\
    =&\left(\frac{v}{v-a}\right)^{n-1}\left(n-1-\frac{1}{n-1}+\log\nicefrac{v}{\underv}\right)\left(n\left(1-\nicefrac{a}{v}\right)^{n-1}-\left(n-1\right)\left(1-\nicefrac{a}{v}\right)^{n}\right) \\
    &- \left(\frac{v}{v-a}\right)^{n-1}\left(\frac{\left(n-1\right)a\log\nicefrac{v}{\underv}+\underv-v}{v-a}\right)\left(1-\nicefrac{a}{v}\right)^{n}-n\left(1-\nicefrac{a}{v}\right) \\
    =& \left(n-1-\frac{1}{n-1}+\log\nicefrac{v}{\underv}\right)\left(n-\left(n-1\right)\left(1-\nicefrac{a}{v}\right)\right) \\
    &-\left(\frac{\left(n-1\right)a\log\nicefrac{v}{\underv}+\underv-v}{v-a}\right)\left(\frac{v-a}{v}\right)-n\left(1-\nicefrac{a}{v}\right) \\
    =& \left(n-1-\frac{1}{n-1}+\log\nicefrac{v}{\underv}\right)\left(n-\left(n-1\right)\left(1-\nicefrac{a}{v}\right)\right) \\
    &- \left(\frac{\left(n-1\right)a\log\nicefrac{v}{\underv}}{v}-\left(1-\nicefrac{\underv}{v}\right)\right)-n\left(1-\nicefrac{a}{v}\right) \\
    =& \left(n-1-\frac{1}{n-1}+\log\nicefrac{v}{\underv}\right)\left(1+\left(n-1\right)\nicefrac{a}{v}\right) \\
    &- \left(n-1\right)\nicefrac{a}{v}\cdot\log\nicefrac{v}{\underv}+\left(1-\nicefrac{\underv}{v}\right)-n\left(1-\nicefrac{a}{v}\right) \\
    =&-\frac{1}{n-1}+\log\nicefrac{v}{\underv}+\left(1-\frac{1}{\left(n-1\right)^{2}}\right)\left(n-1\right)^{2}\nicefrac{a}{v}-\nicefrac{\underv}{v}+n\cdot\nicefrac{a}{v}.
\end{align*}
We substitute $a=\frac{\underv}{(n-1)^2}$ (Eq~\eqref{eq:multi:a_equation}) and we get
\begin{align*}
    & -\frac{1}{n-1}+\log\nicefrac{v}{\underv}+\left(1-\frac{1}{\left(n-1\right)^{2}}\right)\nicefrac{\underv}{v}-\nicefrac{\underv}{v}+\frac{n}{\left(n-1\right)^{2}}\nicefrac{\underv}{v} \\
    =& \frac{1}{n-1}+\log\nicefrac{v}{\underv}+\frac{1}{n-1}\nicefrac{\underv}{v} \\
    =& \log\nicefrac{v}{\underv}-\frac{1-\nicefrac{\underv}{v}}{n-1} \\
    \ge& 0.
\end{align*}

\section{Revisiting the Single bidder Case through Discretization}
\label{appndx:single}

In this section we solve a "simpler" setting (in a sense made precise below) where there is only one bidder.
This case is also solved by \cite{CarrascoFK+18}.
We show how we can use discretization to find the equilibrium, both when the seller knows the bidder's expected value and an upper bound on the values (Section~\ref{subsec:single:first_moment_and_upper_bound}), and when the seller knows the bidder's expected value and second moment (Section~\ref{subsec:single:only_two_moments}). In order find the equilibrium for the latter, we first solve a case where an upper bound is also known (Section~\ref{subsec:single:two_moments_and_upper_bound}). For this problem, we find that the seller has a certain freedom in the choice of the pricing distribution, which is expressed by a family of distributions, all of which yield the same expected revenue.
We are able to identify a specific member of this family, which does not rely on the upper bound on the valuations to guarantee the same expected revenue. This solves the case where the seller only knows the first two moments, and not an upper bound.

Our starting point is the expression for expected revenue. For one bidder with a valuation sampled from a distribution $F$, the expected revenue of a seller who sells the item for a price of $r$, is
\[\Rev_{F}(r) = r\left( 1-F(r) \right).\]
When the seller randomizes a price $r$, using distribution $Q$, the expected revenue is
\begin{equation}
    \Rev_{F}(Q) = E_{r\sim R}[\Rev_{F}(r)] =\int_0^\infty r\left( 1-F(r) \right)dQ(r). \label{single:expected_revnue_form}
\end{equation}

Since the expected revenue function is linear in terms of both the seller's distribution and the adversary's distribution, it is expected that at equilibrium, both the seller and the adversary are indifferent over the support (see discussion in Section~\ref{sec:intro} on indifference).
This makes the single bidder setting easier to analyze.

\subsection{Single bidder: Upper Bound and First Moment Constraint}
\label{subsec:single:first_moment_and_upper_bound}
We assume the adversary can only choose distributions with a mean value of $\mu$, and their support is contained in $[0,\overv]$.
Let $\calF_{\mu}^{\le\overv}$ be the set of distributions which have expectation $\mu$, and support within $[0,\overv]$.
For $F\in \calF_{\mu}^{\le\overv}$ and $Q\in \calQ$, the expected revenue 
The problem
\begin{equation}
    \max_{Q\in \calQ} \min_{F\in \calF_{\mu}^{\le\overv}} \int_0^\infty r\left( 1-F(r) \right)dQ(r). \label{prblm:single:mu_overv}
\end{equation}

\subsubsection{Discrete Case}

For $k\ge 1$, we limit the possible prices and valuations to $k+1$ values $\rvec_k = \left(r_{k,1}, \dots, r_{k,k+1}\right)$, such that $r_{k,k+1} = \overv$.
When the context is clear, we will write $r_i$ instead of $r_{k,i}$.
Unlike the case for two bidders, in this case, using a price of zero, will always yield an expected revenue of zero, so we assume $r_{k,1} = \underv > 0$, for some value $0 < \underv < \mu$ that we will determine when we solve for the general, non discrete case, to make sure is would be a tight lower bound for the support of the distribution the adversary would choose in that general case.

For every $k\ge 1$ and $1\le i \le k+1$, we assign $r_{k,i} = \alpha^{i-1} \underv$ where $\alpha = \left( \nicefrac{\overv}{\underv}\right)^{\nicefrac{1}{k}}$.\footnote{A simple analysis with a general vector of possible values shows that when the ratio between consecutive prices is identical, we get simpler expressions. This was not the case for two bidders, which is immediately clear since there a price of zero is used by the seller.}
Given distributions $Q_k, F_k$ over $\rvec_k$, where $F_k$ has mean of $\mu$, we define the vectors $\qvec_k = \left(q_1, \dots, q_{k+1} \right)$ and $\xvec_k = \left(x_1, \dots, x_{k+1} \right)$ such that $q_i = Q_k(r_i)$ and $x_i = F_k(r_i)$. We also define $q_0 = 0$.

The expected revenue of a mechanism where the seller randomizes a price using $Q_k$ and the bidder's valuation is sampled from $F_k$ is
\[\Rev(\qvec_k, \xvec_k) = \sum_{i=1}^k (q_i - q_{i-1})r_i (1-x_i).\]
The mean value of $F_k$ is
\begin{align*}
    \int_0^{\overv} (1-F(v))dv &= r_1 + \sum_{i=1}^k(r_{i+1}- r_i) (1-x_i) \\
    &= \underv + (\alpha-1)\sum_{i=1}^kr_i(1-x_i).
\end{align*}
Therefore the problem we solve for each $k$ is
\begin{align}
    \max_{\qvec_k} \min_{\xvec_k} \ \ & \sum_{i=1}^k (q_i - q_{i-1})r_i (1-x_i) \label{single:discrete:maxmin:rev}\\
    s.t \ \ & q_1 \le \dots q_k = q_{k+1} = 1 \label{single:discrete:maxmin:q_nondec} \\
    & x_1 \le \dots \le x_{k+1} = 1 \label{single:discrete:maxmin:x_nondec} \\
    & (\alpha-1)\sum_{i=1}^k r_i(1-x_i) = \mu - \underv \label{single:discrete:maxmin:mean_constraint}
\end{align}
We note that we force $q_k= q_{k+1} = 1$ since the seller could never seller the item with a price of $r_{k+1}$.
We allow values lower than zero in this case, but we will make sure that when we take the limit to get the general case, the image of the distribution function would be contained in $[0,1]$. We do not necessarily get negative values, but we do not prove otherwise.
The Lagrangian we get for this problem is
\[\calL(\qvec_k, \xvec_k, \lambda_1) = \Rev(\qvec_k, \xvec_k) - \lambda_1 (\alpha-1)\sum_{i=1}^k r_i(1-x_i).\]
We will ignore the inequalities constraints in $\eqref{single:discrete:maxmin:q_nondec}$ and \eqref{single:discrete:maxmin:x_nondec}, so we could assume that if a solution exists, then it would satisfy
\begin{align}
    \frac{\partial \calL(\qvec_k, \xvec_k, \lambda_1)}{\partial q_i} &= 0 &1\le i \le k-1 \label{single:discrete:derivative_q_is_0}\\
    \frac{\partial \calL(\qvec_k, \xvec_k, \lambda_1)}{\partial x_i} &= 0 & 1 \le i \le k. \label{single:discrete:derivative_x_is_0}
\end{align}
We will then see that the solution we find do satisfy those inequalities.

\begin{claim}  \label{clm:single:disc:x_form}
    If $\qvec_k, \xvec_k$ solves Problem \eqref{single:discrete:maxmin:rev}-\eqref{single:discrete:maxmin:mean_constraint}, then
    \[x_i = 1 - \frac{\mu - \underv}{kr_i(\alpha-1)} \qquad \forall 1\le i\le k.\]
\end{claim}

\begin{proof}
We look at the derivative by $q_i$ for every $1\le i \le k-1$,
\begin{align*}
    \frac{\partial \calL(\qvec_k, \xvec_k, \lambda_1)}{\partial q_i} =  r_i(1-x_i) - r_{i+1}(1-x_{i+1}).
\end{align*}
From equality \eqref{single:discrete:derivative_q_is_0} we get that $r_i(1-x_i) = r_{i+1}(1-x_{i+1})$ for every $1 \le i \le k-1$, which means that for every $1\le i\le k$
\[x_i = 1 -  \frac{r_1}{r_i}(1-x_1) = 1 -  \frac{\underv}{r_i}(1-x_1).\]
Using the constraint \eqref{single:discrete:maxmin:mean_constraint} we get
\begin{align*}
    \mu - \underv &= (\alpha-1)\sum_{i=1}^k r_i(1-x_i)\\
    &= \underv(\alpha-1)\sum_{i=1}^k(1-x_1) \\
    &= k\underv(\alpha-1)(1-x_1).
\end{align*}
And thus
\begin{align*}
    1-x_1 &= \frac{\mu - \underv}{k\underv(\alpha-1)} \\
    \implies x_i &= 1 - \frac{\mu - \underv}{kr_i(\alpha-1)}.
\end{align*}
We can notice that indeed $x_1 \le \dots \le x_k \le 1$.
\end{proof}

\begin{claim} \label{clm:single:disc:q_form}
    If $\qvec_k, \xvec_k$ solves Problem \eqref{single:discrete:maxmin:rev}-\eqref{single:discrete:maxmin:mean_constraint}, then for every $1\le i \le k$
    \[q_i =  \frac{i}{k}.\]
\end{claim}
\begin{proof}
The derivative by $x_i$ for every $1\le i \le k$ is
\[\frac{\partial \calL(\qvec_k, \xvec_k, \lambda_1)}{\partial x_i} =  -(q_i - q_{i-1})r_i + \lambda_1 r_i(\alpha-1).\]
The derivative is zero only when
\[q_i - q_{i-1} = \lambda_1 (\alpha-1).\]
Since $q_0 = 0$, we get that for every $1\le i\le k$
\begin{align*}
    q_i = i\lambda_1 (\alpha-1).
\end{align*}
We know that $q_k=1$, which means
\[k\lambda_1 (\alpha-1) = 1.\]
So we can solve $\lambda_1 = \frac{1}{k(\alpha-1)}$,  and get
And get
\[q_i = \frac{i}{k}.\]
We can also see that the inequalities $0 \le q_1 \le \dots \le q_k$, are satisfied.
\end{proof}

\subsubsection{Mean and Upper Bound Solution}
For every $k\ge 1$, there is a a solution $(\qvec_k, \xvec_k)$ that is described by Claims \ref{clm:single:disc:x_form},  and \ref{clm:single:disc:q_form}.
Those correspond to functions $Q^*_k, F_k^*$ where $Q^*_k(r_i) = q_i$, and $F_k^*(r_i) = x_i$ for every $1\le i \le k+1$.
We defined the possible values in the discrete case to be between in $[\underv, \overv]$ where $\underv > 0$ is a parameter we can control, so we can set it in such a way that the adversary's best distribution's support would be in $[\underv, \overv]$.
\begin{claim} \label{clm:single:gen:F_form}
    For arbitrary $\underv > 0$ the sequence $\left\{F^*_k\right\}_{k\in \N}$ is pointwise convergent to $F^*$,where
    \[F^*(v) = \begin{cases}
    0 & v\in [0,\underv), \\
    1-\frac{\mu - \underv}{v \log\left(\nicefrac{\overv}{\underv}\right)} & v \in [\underv, \overv), \\
    1 & v \ge \overv.
    \end{cases}\]
\end{claim}
\begin{proof}
We know that $r_{k,1} \in [0,\overv]$ for every $k$, so the sequence $\{r_{k,1}\}_{k\in \N}$ has a convergent subsequence. We will assume that the sequence itself is convergent, however all of the steps can be done on every convergent subsequence, which would all yield the same limit, which would prove that the sequence is indeed convergent.

Since $\alpha = \left(\nicefrac{\overv}{\underv}^{\nicefrac{1}{k}} \right)$, we get that
\[\lim_{k\to \infty} k(\alpha - 1) = \log\left(\nicefrac{\overv}{\underv}\right).\]
From these limits, and from Claim \ref{clm:single:disc:x_form}, we get
\begin{align*}
    F^*(v) = 1- \frac{\mu - \underv}{v \log\left(\nicefrac{\overv}{\underv}\right)}.
\end{align*}
\end{proof}
\begin{claim} \label{clm:single:gen:Q_form}
    For arbitrary $\underv > 0$ the sequence $\left\{Q^*_k\right\}_{k\in \N}$ is pointwise convergent to $Q^*$,where
    \[Q^*(v) = \begin{cases}
        0 & r \in [0, \underv), \\
        \frac{\log r - \log \underv}{\log \overv - \log \underv} & r\in [\underv, \overv].
    \end{cases}\]
\end{claim}

We now need to set $\underv$ such that neither the seller nor the adversary would prefer having any mass in $[0,\underv)$.
We can use the fact that both the adversary and the seller are indifferent over $[\underv, \overv]$, when the other uses $Q^*$ or $F^*$ respectively.
\begin{claim} \label{clm:single:mu_upper_bound:sol}
    If $\underv$ solves $\frac{\mu - \underv}{\log \left(\nicefrac{\overv}{\underv} \right)} = \underv$, then $Q^*$, and $F^*$ defined in Claims \ref{clm:single:gen:F_form}, \ref{clm:single:gen:Q_form}, are a solution to Problem \eqref{prblm:single:mu_overv}.
\end{claim}
\begin{proof}
When the seller uses distribution $Q$ to determine the price, and the bidder samples valuation from distribution $F$, the expected revenue is given by
\[Rev_F(Q) = \mathbb{E}_{r\sim Q}\left[\Rev_{F}(r)\right] = \int_0^\infty r\left(1-F(r)\right)dQ(r).\]
If $\underv = \frac{\mu - \underv}{ \log\left(\nicefrac{\overv}{\underv}\right)}$, then for $F=F^*$ we get a revenue of
\begin{align*}
    Rev_{F^*}(Q) &= \int_0^\infty r\left(1-F^*(r)\right)dQ(r) \\
    &= \int_0^{\underv}rdQ(r) + \int_{\underv}^{\overv} \frac{\mu - \underv}{ \log\left(\nicefrac{\overv}{\underv}\right)}dQ(r) \\
    &\le \underv Q(\underv) +\underv \left(Q(\overv) - Q(\underv) \right) \\
    &\le \underv.
\end{align*}
Likewise, if $\underv = \frac{\mu - \underv}{ \log\left(\nicefrac{\overv}{\underv}\right)}$, then for $Q=Q^*$, and $F$ with a mean value of $\mu$, and a support contained in $[0,\overv]$, we get a revenue of
\begin{align*}
    Rev_{F}(Q^*) &= \int_0^\infty r\left(1-F(r)\right)dQ^*(r) \\
    &= \int_0^\infty r\left(1-F(r)\right)(Q^*)'(r)dr \\
    &= \int_{\underv}^{\overv} \frac{\left(1-F(r)\right)}{\log \left( \nicefrac{\overv}{\underv} \right)} dr \\
    &= \frac{\mu - \int_0^{\underv}(1-F(r))dr}{\log \left( \nicefrac{\overv}{\underv} \right)} \\
    & \ge \frac{\mu - \underv}{\log \left( \nicefrac{\overv}{\underv} \right)} \\
    &= \underv.
\end{align*}
\end{proof}

\subsection{Single bidder: Upper Bound and Two Moment Constraints}
\label{subsec:single:two_moments_and_upper_bound}
In order to solve the case where only the mean value is known, and an upper bound on the second moment, we first solve a simpler problem.
In this problem, we assume that the seller knows an upper bound $\mu_2$, on the second moment of the bidder's distribution, and also an upper bound $\overv$ on the valuations, and the mean value $\mu$ of the distribution.
Without the constraint on the second moment, we have solved this problem, and found that for the equilibrium, the bidder sample her valuation using $F^*$ (Claim~\ref{clm:single:mu_upper_bound:sol}).
We solve for the cases where $\mu_2$ is equal to the second moment of $F^*$, which means that
\begin{align*}
    \mu_2 &= \int_0^{\overv}2v(1-F^*(v))dv\\
    &= \underv^2 + \int_{\underv}^{\overv}2\underv dv \\
    &= 2\overv \underv - \underv^2.
\end{align*}.
With this $\mu_2$, we define $\calF_{\mu,\mu_2}^{\le \overv}$ to be the set of all distributions over $[0,\overv]$ with a mean value of $\mu$, and a second moment of $\mu_2$.
And the new problem is
\begin{equation}
    \max_{Q\in \calQ} \min_{F\in \calF_{\mu, \mu_2}^{\le\overv}} \int_0^\infty r\left( 1-F(r) \right)dQ(r). \label{prblm:single:mu_mu2_overv}
\end{equation}

Clearly, the solution for Problem \eqref{prblm:single:mu_overv} is also a solution for Problem \eqref{prblm:single:mu_mu2_overv}. However, with the new constraint, we can find a family of solutions for the seller, that can grantee the same expected revenue. Afterwards, we will identify a specific distribution in this family that solves that guarantee the same revenue even when we remove the upper bound on the valuations.

\subsubsection{Mean, Upper Bound, Second Moment: Discrete Case}
For every $k\ge1$, we solve for a discrete case with $k+1$ possibilities, similarly to the case without the additional constraint on the second moment. We also define $\underv$ to be solution of $\frac{\mu - \underv}{\log \left(\nicefrac{\overv}{\underv} \right)} = \underv$, like in Claim \ref{clm:single:mu_upper_bound:sol}.
In order to make sure the adversary acts in the same way as without the constraint, we allow the second moment to be higher than $\mu_2$ for each $k$, and the upper bound on the second moment we will use is the second moment of $F^*_k$, where $F^*_k$ is the adversary's previous solution, according got Claim \ref{clm:single:disc:x_form}.
Which means
\begin{align*}
    \mu_{2,k} &= \int_0^{\overv}2v(1-F^*_k(v))dv \\
    &= \underv^2 + \sum_{i=1}^k \int_{r_i}^{r_{i+1}} 2v(1-x_i^*)dv \\
    &=\underv^2 + \sum_{i=1}^k (1-x_i^*)(r_{i+1}^2 - r_i^2) \\
    &= \underv^2 + \sum_{i=1}^k \frac{\mu - \underv}{kr_i(\alpha-1)}(\alpha^2 - 1)r_i^2 \\
    &= \underv^2 + \frac{(\mu - \underv)(\alpha+1)}{k}\sum_{i=1}^k r_i \\
    &= \underv^2 + \frac{(\mu - \underv)(\overv - \underv)(\alpha+1)}{k(\alpha-1)}.
\end{align*}
The new problem is identical to \ref{single:discrete:maxmin:rev}-\ref{single:discrete:maxmin:mean_constraint}, with the additional constraint
\begin{equation} \label{single:discrete:maxmin:mu2:second_moment}
    (\alpha^2 - 1)\sum_{i=1}^k (1-x_i)r_i^2 \le \mu_{2,k} - \underv^2.
\end{equation}
That changes the Lagrangian for this problem to be
\[\calL(\qvec_k, \xvec_k, \lambda_1, \lambda_2) = \Rev(\qvec_k, \xvec_k) - \lambda_1 (\alpha-1)\sum_{i=1}^k r_i(1-x_i) - \lambda_2(\alpha^2 - 1)\sum_{i=1}^k (1-x_i)r_i^2.\]
The solution for the adversary remains completely identical.
But for the seller, we get a bit more freedom.

\begin{claim}  \label{clm:single:mu_upper_bound_mu2:disc:q_form}
    If $\qvec_k, \xvec_k$ solves Problem \eqref{single:discrete:maxmin:rev}-\eqref{single:discrete:maxmin:mean_constraint}, with constraint \eqref{single:discrete:maxmin:mu2:second_moment} as well, then for every $1\le i \le k$
    \[q_i =  i\lambda_1 (\alpha-1) + \lambda_2 (\alpha + 1)(\alpha r_i - \underv),\]
    where $\lambda_1 = \frac{1 - \lambda_2 (\alpha + 1)(\overv - \underv)}{k(\alpha-1)}$.
\end{claim}
\begin{proof}
The derivative by $x_i$ for every $1\le i \le k$ is
\[\frac{\partial \calL(\qvec_k, \xvec_k, \lambda_1 , \lambda_2)}{\partial x_i} =  -(q_i - q_{i-1})r_i + \lambda_1 r_i(\alpha-1) + \lambda_2 r_i^2(\alpha^2-1).\]
The derivative is zero only when
\[q_i - q_{i-1} = \lambda_1 (\alpha-1) + \lambda_2 r_i(\alpha^2-1).\]
Since $q_0 = 0$, we get that for every $1\le i\le k$
\begin{align*}
    q_i &=\sum_{j=1}^i \lambda_1 (\alpha-1) + \lambda_2 r_j(\alpha^2-1) \\
    &= i\lambda_1 (\alpha-1) + \lambda_2 \underv(\alpha^2-1) \sum_{j=1}^i \alpha^{j-1} \\
    &= i\lambda_1 (\alpha-1) + \lambda_2 (\alpha + 1)(\alpha r_i - \underv).
\end{align*}
We know that $q_k=1$, which means
\[k\lambda_1 (\alpha-1) + \lambda_2 (\alpha + 1)(\overv - \underv) = 1.\]
And we can solve $\lambda_1$,
\[\lambda_1 = \frac{1 - \lambda_2 (\alpha + 1)(\overv - \underv)}{k(\alpha-1)}.\]
\end{proof}

\subsubsection{Mean, Upper Bound, Second Moment: General Case}
In the discrete case, we get freedom that depends on $\lambda_2$, which yields a family of solutions.
This translates to the general case as well, as we take the limit.

From Claim \ref{clm:single:mu_upper_bound_mu2:disc:q_form}, for each $k$, the solution $Q^*_k$ for the discrete problem satisfies
\[Q_k^*(r_i) = \frac{i}{k}(1-\lambda_2(\alpha+1)(\overv-\underv)) + \lambda_2 (\alpha + 1)(\alpha r_i - \underv).\]
And thus, $\left\{Q_k^*\right\}_{k\in \N}$ is pointwise convergent to $Q^*$ where
\[Q^*(r) = \frac{\log \nicefrac{r}{\underv}}{\log \nicefrac{\overv}{\underv}}(1-2\lambda_2(\overv-\underv)) + 2\lambda_2(r-\underv).\]
Similarly, when $k\to \infty$
\begin{equation}
    \lambda_1 = \frac{1-2\lambda_2(\overv - \underv)}{\log(\nicefrac{\overv}{\underv})}. \label{eq:single:2_moments_and_upper:lambda_1_2_equality}
\end{equation}
And we get
\[Q^*_{\lambda_2}(r) = \lambda_1 \log \nicefrac{r}{\underv} + 2\lambda_2 (r-\underv).\]

We just need to make sure that $Q^*$ is a distribution function, which means we need to make sure that the formula we got is non-decreasing in $[\underv, \overv]$ and that its image over the same interval is contained in $[0,1]$.

Since $Q^*(\underv) = 0$, and $Q^*(\overv) = 1$ regardless of $\lambda_2$, we just need to check that it is non decreasing.

The derivative of $Q^*$ is
\[\frac{dQ^*_{\lambda_2}(r)}{dr} = \frac{\lambda_1}{r} + 2\lambda_2.\]
So we demand
\[\lambda_1 \ge -2r\lambda_2 \qquad \forall r \in [\underv, \overv].\]

\begin{claim}
When the seller randomizes a price using $Q^*_{\lambda_2}$, for $\lambda_1, \lambda_2$ satisfying
\begin{align*}
    \lambda_1 &= \frac{1-2\lambda_2(\overv - \underv)}{\log(\nicefrac{\overv}{\underv})}, \\
    \lambda_1 &\ge -2r\lambda_2 \qquad \forall r \in [\underv, \overv],
\end{align*}
she guarantees an expected revenue of at least $\underv$, when the bidder samples her valuation from $F\in \calF_{\mu, \mu_2}^{\le \overv}$.
\end{claim}
\begin{proof}
The seller best response cannot have any mass below $\underv$, since any such distribution is dominated by a distribution that doesn't. We can see that since transferring the any mass below $\underv$ to $\underv$ would not increase the expected revenue, while increasing the expected value and decreasing the variance. This allows us to lower the values higher than $\mu$ to get to the right expected value, without increasing the revenue, while still decreasing second moment.
Therefore, we need only to check cases where the support is fully contained in $[\underv, \overv]$.

Let be $F \in \calF_{\mu, \mu_2}^{\le \overv}$ such that the support of $F$ is contained in $[\underv, \overv]$.
\begin{align*}
    Rev_{F}(Q^*) &= \int_0^\infty r\left(1-F(r)\right)(Q^*)'(r)dr \\
    &= \int_{\underv}^{\overv} \left(\lambda_1 \left(1-F(r)\right) + 2\lambda_2 r \left(1-F(r)\right) \right)dr \\
    &= \lambda_1 (\mu - \underv)  + \lambda_2 (\mu_2 - \underv^2) \\
    &= \frac{1-2\lambda_2(\overv - \underv)}{\log(\nicefrac{\overv}{\underv})} (\mu - \underv) + \lambda_2 (2\overv \underv - 2\underv^2) \\
    &= \underv-2\lambda_2\underv(\overv - \underv) + 2\lambda_2 \underv (\overv - \underv) \\
    & = \underv.
\end{align*}
\end{proof}

\subsection{Single bidder: Two Out of Three Constraints}
\label{subsec:single:only_two_moments}
In subsection~\ref{subsec:single:two_moments_and_upper_bound} we have introduced the second moment constraint, where $\mu_2$  was set to be:
\[\mu_2 = 2\overv \underv - \underv^2,\]
where $\underv$ is the unique solution to
\[\mu = \underv \left(1 + \log \left(\nicefrac{\overv}{\underv} \right)\right).\]
In this section, we assume only knowledge of then mean value $\mu$, and an upper bound on the second moment $\mu_2$.

We define $(\underv, \overv)$ to be the solution to:
\begin{align*}
    \mu_2 &= 2\overv \underv - \underv^2, \\
    \mu &= \underv \left(1 + \log \left(\nicefrac{\overv}{\underv} \right)\right).
\end{align*}
If the bidder could not have any values higher than $\overv$, we would have end up with the same solution as in subsection~\ref{subsec:single:two_moments_and_upper_bound}, which is:
\begin{align*}
    Q^*_{\lambda_2}(r) &= \begin{cases}
    0 & r \in [0,\underv] \\
    \lambda_1 \log \nicefrac{r}{\underv} + 2\lambda_2 (r-\underv) & r\in(\underv, \overv] \\
    1 & r > \overv
    \end{cases}\\
    F^*(v) &= \begin{cases}
    0 & v\in [0,\underv] \\
    1 - \nicefrac{\underv}{v} & v\in (\underv, \overv) \\
    1 & v\ge \overv
    \end{cases} \\
    \lambda_1 &= \frac{1-2\lambda_2(\overv - \underv)}{\log(\nicefrac{\overv}{\underv})}, \\
    \lambda_1 &\ge -2r\lambda_2 \qquad \forall r \in [\underv, \overv].
\end{align*}
We now explore how we can set $\lambda_2$ such that the adversary would not benefit from having mass above $\overv$.

Fix $F \in \calF_{\mu, \le \mu_2}$ (that is $F$ only assumed to have an expectation $\mu$, and a second moment that isn't higher than $\mu_2$).
We also assume w.o.l.g that $F$ doesn't have any mass lower than $\underv$, since if it does, than there exists $F_2 \in \calF_{\mu, \le \mu_2}$ that doesn't have any mass lower than $\underv$, and yield no higher expected revenue.

If $F$ has a second moment of exactly $\mu_2$, we get that the expected revenue the seller guarantees herself when using $Q^*$ is:
\begin{align*}
    Rev_{F}(Q^*) &= \int_0^\infty r\left(1-F(r)\right)(Q^*)'(r)dr \\
    &= \int_{\underv}^{\overv} \left(\lambda_1 \left(1-F(r)\right) + 2\lambda_2 r \left(1-F(r)\right) \right)dr \\
    &= \lambda_1 \left(\mu - \underv - \int_{\overv}^\infty (1-F(r))dr \right)  + \lambda_2 \left(\mu_2 - \underv^2 - \int_{\overv}^\infty 2r(1-F(r))dr\right) \\
    &= \lambda_1 (\mu - \underv) + \lambda_2 (2\overv \underv - 2\underv^2) - \int_{\overv}^\infty (1-F(r))(\lambda_1 + 2\lambda_2r)dr\\
    &= \frac{1-2\lambda_2(\overv - \underv)}{\log(\nicefrac{\overv}{\underv})} (\mu - \underv) + \lambda_2 (2\overv \underv - 2\underv^2) - \int_{\overv}^\infty (1-F(r))(\lambda_1 + 2\lambda_2r)dr\\
    &= \underv-2\lambda_2\underv(\overv - \underv) + 2\lambda_2 \underv (\overv - \underv) - \int_{\overv}^\infty (1-F(r))(\lambda_1 + 2\lambda_2r)dr\\
    & = \underv - \int_{\overv}^\infty (1-F(r))(\lambda_1 + 2\lambda_2r)dr.
\end{align*}
Therefore, we would require $\lambda_1 \le -2\lambda_2r$ for every $r\ge \overv$.
We know that $\lambda_1 \ge -2\lambda_2\overv$, and thus we can set
\[ \lambda_1 = -2\lambda_2 \overv .\]
In that case, we get
\[\frac{1-2\lambda_2(\overv - \underv)}{\log(\nicefrac{\overv}{\underv})} = -2\lambda_2 \overv.\]
We can solve $\lambda_2$ and $\lambda_1$ to be
\begin{align}
    \lambda_2 &= \frac{1}{2((\overv - \underv)-\overv \log(\nicefrac{\overv}{\underv}))}, \\
    \lambda_1 &= -\frac{\overv}{(\overv - \underv)-\overv \log(\nicefrac{\overv}{\underv})}.
\end{align}
Note that $\lambda_2 < 0$, and thus the inequalities  
\begin{align*}
     \lambda_1 &\ge -2r\lambda_2 \qquad \forall r \in [\underv, \overv], \\
     \lambda_1 &\le -2r\lambda_2 \qquad \forall r \ge \overv
\end{align*}
are satisfied.

\begin{claim}
If the seller randomized the price by sampling from $Q^*$ defined below, then she guarantees an expected revenue of at least $\underv$ (defined below) if the bidder uses a distribution $F$ with a mean value $\mu$, and a second moment no larger than $\mu_2$.
\[Q^*(r) = \begin{cases}
0 & r \in [0,\underv] \\
\lambda_1 \log \nicefrac{r}{\underv} + 2\lambda_2 (r-\underv) & r\in(\underv, \overv] \\
1 & r > \overv
\end{cases}\]
where
\begin{align*}
    \lambda_2 &= \frac{1}{2((\overv - \underv)-\overv \log(\nicefrac{\overv}{\underv}))}, \\
    \lambda_1 &= -\frac{\overv}{(\overv - \underv)-\overv \log(\nicefrac{\overv}{\underv})},
\end{align*}
and $\underv, \overv$ are the solution to
\begin{align*}
    \mu_2 &= 2\overv \underv - \underv^2, \\
    \mu &= \underv \left(1 + \log \left(\nicefrac{\overv}{\underv} \right)\right).
\end{align*}
\end{claim}
\begin{proof}
As explained previously, we can assume that $F$ doesn't have any mass below $\underv$.
We note that $\lambda_2$ is always negative, and $\lambda_1 + 2\lambda_2 r$ is non-positive for every $r\ge \overv$.
The expected revenue is
\begin{align*}
    Rev_{F}(Q^*) &= \int_0^\infty r\left(1-F(r)\right)(Q^*)'(r)dr \\
    &= \int_{\underv}^{\overv} \left(\lambda_1 \left(1-F(r)\right) + 2\lambda_2 r \left(1-F(r)\right) \right)dr \\
    &\ge \lambda_1 \left(\mu - \underv - \int_{\overv}^\infty (1-F(r))dr \right)  + \lambda_2 \left(\mu_2 - \underv^2 - \int_{\overv}^\infty 2r(1-F(r))dr\right) \\
    &= \lambda_1 (\mu - \underv) + \lambda_2 (2\overv \underv - 2\underv^2) - \int_{\overv}^\infty (1-F(r))(\lambda_1 + 2\lambda_2r)dr\\
    &= \underv - \int_{\overv}^\infty (1-F(r))(\lambda_1 + 2\lambda_2r)dr\\
    &\ge \underv.
\end{align*}
\end{proof}

\end{document}